\documentclass{article}

\usepackage{graphicx}
\usepackage{times}
\usepackage{bm}
\usepackage{amssymb}
\usepackage{amsmath}
\usepackage{amsthm}
\usepackage{varwidth}
\usepackage{paralist}
\usepackage[inline]{enumitem}
\usepackage{hyperref}
\usepackage{authblk}

\newtheorem{runningexample}{Example}
\newtheorem{proposition}{Proposition}
\newtheorem{examples}{Examples}
\newtheorem{example}{Example}
\newtheorem{obs}{Observation}
\newtheorem{definition}{Definition}
\newtheorem{theorem}{Theorem}
\newtheorem{corollary}{Corollary}
\newtheorem{lemma}{Lemma}
\newtheorem{claim}{Claim}

\input{lata_preamble}

\newcommand{\proofclosure}{
\begin{proof}
\label{proof_th_closure}
Let $\pA_1$ and $\pA_2$ be PIAs that accept the languages $L(\pA_1)$ and $L(\pA_2)$, respectively.
In the following, we describe an automaton $\pA$ 
which accepts the corresponding language operation.

\paragraph*{Union.} 
$L(\pA_1) \cup L(\pA_2)$ is accepted by a pebble-intervals automaton $\pA$ which 
on input $w$, \new{uses a silent transition to move to the initial
  state of either $\pA_1$  of $\pA_2$ (intuitively, it guesses whether $w \in L(\pA_1)$ or $w \in L(\pA_2)$)} and then simulates the corresponding automaton on $w$.

\paragraph*{Concatenation.}
We show that $L(\pA_1) L(\pA_2)$ is accepted by a pebble-intervals automaton $\pA$.
Essentially, $\pA$ guesses where the partition is in the input, then it
simulates the corresponding automaton on each segment of the input.
More precisely, given input $w$, if $\varepsilon \in L(\pA_1)\cap L(\pA_2)$, 
the automaton $\pA$ guesses whether $w$ is the empty string, and if so moves to an accepting state using a silent transition. Otherwise, 
$\pA$ places a pebble $p$ on an arbitrary position $\ell$ in $w$.
$\pA$ remembers the letter $w(\ell)$ in its state, as it will need to simulate $\pA_2$ as if it has read the letter $w(\ell)$ in the first position in the input to $\pA_2$. 
$\pA$ then simulates $\pA_1$ on the prefix $w(1) \ldots w(\ell-1)$
by replacing all move transitions $\Move{k}{i}{\rhd}$ with $\Move{k}{i}{p}$.  
Similarly, 
$\pA$ then simulates $\pA_2$ on the suffix $w(\ell) \ldots w(n)$ by replacing all move transitions 
$\Move{k}{\lhd}{i}$ with $\Move{k}{p}{i}$. 
However, for this part of the input we need some new silent transitions to simulate $\pA_2$ placing a pebble on the position $p$ is on
and reading the letter $w(\ell)$.
$\pA$ accepts if both simulations accept.

\paragraph*{Kleene-$\star$.}
We show that $L(\pA_1)^\star$ is accepted by a pebble-intervals automaton $\pA$.
This automaton works similarly to the concatenation case, except it uses two pebble to enclose the substring it simulate an automaton on,
and it non-deterministically chooses when to move on to the next segment of the word. 
The automaton $\pA$ guesses whether the input $w$ is the empty string or not. If so, it goes into an accepting state. 
Otherwise, $\pA$ begins to simulate $\pA_1$ on contiguous substrings of $w$ in rounds, using pebbles $p,p'$ to enclose the substrings. 
In the first round, $\pA$ places $p$ anywhere and then places $p'$ somewhere to the right of $p$.
Then $\pA$ simulates $\pA_1$ on the substring between the pebbles by replacing $\Move{k}{\lhd}{i}$ transitions with $\Move{k}{p}{i}$ transitions, and 
$\Move{k}{i}{\rhd}$ transitions with $\Move{k}{i}{p'}$ transitions. In addition, $\pA$ uses new silent transitions to simulate $\pA_1$ placing pebbles on the positions
of $p$ and $p'$.
The subsequent rounds are similar, except that they begin with $\pA$ placing $p$ to the right of $p'$.
Note that for a run to be successful, the placement of $p$ in the first round must be on the first position of the input,
and in subsequent rounds, the placement of $p$ must be immediately to the right of $p'$.
$\pA$ decides non-deterministically to finish a round or continue it whenever the simulated state of $\pA_1$ is accepting
(otherwise $\pA$ must continue the round). $\pA$ is at an accepting state whenever $\pA_1$ is at an accepting state.

\paragraph*{Shuffle.}
The automaton $\pA$ uses two disjoint sets of pebbles. $\pA$ uses one set of pebbles to simulate $\pA_1$ on some substring of $w$, which is not necessarily contiguous,
and the other set of pebbles to simulate $\pA_2$ on the substring composed of the remaining positions. 
$\pA$ accepts if both simulations accept.

\paragraph*{Iterated shuffle.}
This is similar to the case of the shuffle, except that $\pA$ performs several simulations of $\pA_1$ in rounds.
At the beginning of each round, $\pA_1$ is set to its initial state,
and all pebbles are considered available. $\pA$ simulates $\pA_1$ on an arbitrary subset of the unread positions of the input. 
If $\pA_1$ goes into an accepting state, $\pA$ non-deterministically chooses whether to continue the simulation of $\pA_1$,
or whether a shuffled string was accepted by $\pA_1$ and this is the end of a round. 
In the latter case $\pA$ either starts a new round or it goes into an accepting state.
\end{proof}
}

\newcommand{\proofnonclosure}{
\begin{proof}
\label{proof_th_intersection}
(Part 1.)
We show that if pebble-intervals automata were effectively closed under intersection with regular languages,
we could decide an undecidable problem. Namely, the Minsky halting problem. 
A Minsky machine \cite{minsky1967computation} is a sequence of labeled instructions
\begin{flalign*}
\begin{array}{ll}
\texttt{0}: &\mathtt{comm}_0 \\ 
\texttt{1}: &\mathtt{comm}_1 \\ 
\vdots & \vdots \\
\texttt{n-1}: &\mathtt{comm}_{n-1} \\ 
\mathtt{HALT}: &\mathtt{halt}	
\end{array}
\end{flalign*}
where each of the first $n$ instructions is either an $\texttt{inc}_\texttt{c}$ instruction of the form:
\begin{flalign*}
\begin{array}{ll}
\texttt{i:} & \texttt{c:= c+1; goto i+1}	
\end{array}	
\end{flalign*}
or a conditional $\texttt{dec}_\texttt{c}\texttt{(j)}$ instruction of the form:
\begin{flalign*}
\begin{array}{ll}
\texttt{i:} & \texttt{if c=0 then goto j} \\
						& \texttt{else c:=c-1; goto i+1}
\end{array}
\end{flalign*}
and $\texttt{c}$ is one of the registers of the machine.

A \emph{trace} of a Minsky machine $M$ is a sequence $\texttt{t}_1,\ldots,\texttt{t}_m$ of labels where $\texttt{t}_1 = \texttt{0}$, 
$\texttt{t}_m = \texttt{HALT}$, and for
$0 < k < m$, if $\texttt{t}_{k-1} = \texttt{i}$ is an $\texttt{inc}_\texttt{c}$ instruction, then $\texttt{t}_{k} = \texttt{i+1}$,
and if it is a $\texttt{dec}_\texttt{c}\texttt{(j)}$ instruction, then either $\texttt{t}_{k} = \texttt{i+1}$, or $\texttt{t}_{k} = \texttt{j}$.

Note that the language of traces of a Minsky machine $M$ is regular.

A trace $\texttt{t}_1,\ldots,\texttt{t}_m$ is \emph{feasible} if the following holds. 
If $\texttt{t}_{k-1} = \texttt{i}$ for a $\texttt{dec}_\texttt{c}\texttt{(j)}$ instruction, and $\texttt{j} \neq \texttt{i+1}$,
then $\texttt{t}_{k} = \texttt{j}$ iff $\texttt{c} = 0$ at step $k-1$ during the run of $M$. 
Note that $c=0$ holds exactly at the points
where the number of 
$\texttt{inc}_\texttt{c}$ instructions is equal to the number of $\texttt{dec}_\texttt{c}$ instructions.

The problem of deciding whether a given Minsky machine $M$ with two registers
$c_0,c_1$ halts when started with $\texttt{c}_0 = \texttt{c}_1 = 0$ is referred to here as the \emph{Minsky halting problem}. 
It is known that the Minsky halting problem is undecidable \cite{schroeppel1972two}.
In other words, it is undecidable whether there is a feasible trace of $M$.

Now fix a Minsky machine $M$ and denote the language of its traces by $\texttt{T}$.

Let $\Sigma = \{\texttt{halt},\texttt{inc}_\texttt{c},\texttt{dec}_\texttt{c},\texttt{jmp}_\texttt{c} \mid \texttt{c} \in \{\texttt{c}_0,\texttt{c}_1\}\}$.
We define a mapping $h: \texttt{T} \rightarrow \Sigma^\star$ from traces to $\Sigma^\star$ which 
will also produce a regular language.
Let $\texttt{t}_1 \ldots \texttt{t}_m$ be a trace. For $\texttt{t}_{k-1} = \texttt{i}$ where $0 < k < m$:
\begin{itemize}
	\item 
	If $\texttt{i}$ is an $\texttt{inc}_\texttt{c}$ instruction, then $f(\texttt{t}_{k-1},\texttt{t}_{k}) = \texttt{inc}_\texttt{c}$.
	\item
	If $\texttt{i}$ is a $\texttt{dec}_\texttt{c}\texttt{(j)}$ instruction with $\texttt{j} \neq \texttt{i+1}$,
	then
	\[
	f(\texttt{t}_{k-1},\texttt{t}_{k}) = 
	\begin{cases}
	\texttt{dec}_\texttt{c} & \text{if $\texttt{t}_{k} = \texttt{i+1}$,}\\
	\texttt{jmp}_\texttt{c} & \text{if $\texttt{t}_{k} = \texttt{j}$}
	\end{cases}
	\]
	\item
	If $\texttt{i}$ is a $\texttt{dec}_\texttt{c}\texttt{(i+1)}$ instruction,
	then
	$f(\texttt{t}_{k-1},\texttt{t}_{k}) \in \{\texttt{dec}_\texttt{c},\texttt{jmp}_\texttt{c}\}$.
\end{itemize}
Now define $h(\texttt{t}_1 \ldots \texttt{t}_m)$ as the string $f(\texttt{t}_1,\texttt{t}_2)f(\texttt{t}_2,\texttt{t}_3)\ldots f(\texttt{t}_{m-1},\texttt{t}_m) \, \texttt{halt}$.

Denote the language resulting from applying $h$ to all the traces by $\texttt{Inst}_M = \{h(\texttt{t}) \mid \texttt{t} \in \texttt{T} \}$, and note $\texttt{Inst}_M$ is regular.
This language describes the traces of a machine. In order to decide whether there exists a feasible trace of $M$,
we need to be able to test if the appearances of the $\texttt{jmp}_\texttt{c}$ letter are in appropriate positions. 
For this purpose, we next define a pebble-intervals language by shuffling two pebble-intervals languages given by a context-free grammar,
which can distinguish between appropriately-placed and inappropriately-placed $\texttt{jmp}_\texttt{c}$ letters.
The intersection of the languages will consist of the feasible traces, meaning it will be non-empty if and only
if the given Minsky machine halts. 

We now define a context free grammar which produces sequences of 
$\texttt{inc}_{\texttt{c}}$ and $\texttt{dec}_{\texttt{c}}$ which are
well matched, along with $\texttt{jmp}_{\texttt{c}}$ letters which are appropriately placed,
that is, if $\texttt{jmp}_{\texttt{c}}$ appear in position $i$ of a string, then there is 
an equal number of 
$\texttt{inc}_{\texttt{c}}$ and $\texttt{dec}_{\texttt{c}}$ in the prefix containing the first $i$ positions.
Define the context free grammar $\Gamma_\texttt{c}$ as follows:
\begin{align*}
\begin{array}{rcl}
A & ::= & B \ \texttt{halt} \\
B & ::= & BB  \mid C B C \mid D\\
C & ::= & \varepsilon \mid \texttt{jmp}_{\texttt{c}} \ C \\
D & ::= & \texttt{inc}_{\texttt{c}} \ D \ \texttt{dec}_{\texttt{c}} \mid DD \mid \varepsilon
\end{array}
\end{align*}

\begin{claim}
\label{cl_grammar_is_pia}
The language produced by the grammar $\Gamma_\texttt{c}$ is PI.
\end{claim}
\begin{proof}[Proof of Claim~\ref{cl_grammar_is_pia}]
Essentially, the automaton simulates $\pA_{Dyck}$ on substrings of its input in order
verify that every jump is preceded by the appropriate number of increments and decrements.
These substrings are enclosed between pebbles 1 and 2 (or the ends of the input). The automaton $\pA_{Dyck}$
is simulated using a disjoint set of pebbles.

We describe a successful run of the automaton for one of the registers.
We automaton places pebble $0$ on the last position to read $\texttt{halt}$.
Then it guesses whether there are any $\texttt{jmp}_{\texttt{c}}$ on the input. 
If not, then it simply simulates $\pA_{Dyck}$ on the whole input.
Otherwise, pebble $1$ is placed to the left of $0$ to read $\texttt{jmp}_{\texttt{c}}$, and  $\pA_{Dyck}$
is simulated on the substring in the interval between the beginning of the input and pebble $1$.
Let $i \in \{1,2\}$ be the pebble placed at the beginning of the current round.
In a new round, the automaton guesses whether there is $\texttt{jmp}_{\texttt{c}}$ to the right of pebble $i$,
and if so, it moves pebble $3-i$ to the right of pebble $i$ to read $\texttt{jmp}_{\texttt{c}}$, and simulates 
$\pA_{Dyck}$ on the substring in the interval between $3-i$ and $i$.
If the automaton guesses there is no $\texttt{jmp}_{\texttt{c}}$ to the right of pebble $i$, it simulates $\pA_{Dyck}$
on the remainder of the input, in the interval between $i$ and $0$ and does not enter any new rounds.
\end{proof}

We resume the proof of Theorem~\ref{th_intersection}.
Define
$L_{\texttt{c}_0}$ as the language produced by $\Gamma_{\texttt{c}_0}$, and $L_{\texttt{c}_1}$ as the language produced by $\Gamma_{\texttt{c}_1}$.
Finally, define $\texttt{L}$ as the shuffle between $L_{\texttt{c}_0}$ and $L_{\texttt{c}_1}$, and note
that by Claim~\ref{cl_grammar_is_pia} and Theorem \ref{th_closure}, it is a pebble-intervals language.

Therefore, intersecting $\texttt{L}$ with $\texttt{Inst}_M$ would result in exactly the feasible traces. 
Now assume for contradiction that pebble-intervals languages are effectively closed under intersection with regular languages.
We describe a procedure for deciding the Minsky halting problem. Given $M$, generate and automaton $\pA_{\mathit{inst}}$ for 
$\texttt{Inst}_M$. Using
the assumed effective procedure for intersection with regular languages, now generate an automaton $\pA_M$ for $L_M = \texttt{L} \cap \texttt{Inst}_M$.
To decide the Minsky halting problem for $M$, test $\pA_M$ for emptiness.

Since emptiness of pebble-intervals automata is decidable, we contradict undecidability of the Minsky halting problem, 
so we conclude
that pebble-intervals languages are not effectively closed under intersection, even with regular languages.

(Part 2.)
We show we can build an automaton for $(L_M)^c$ and conclude that if we could effectively build a complement automaton
for any pebble-intervals language, then in particular we could build one 
for $((L_M)^c)^c = L_M$, which we could test for emptiness and again solve the Minsky halting problem.

Automaton for $(L_M)^c$: since $L_M = \texttt{L} \cap \texttt{Inst}_M$, we have that $(L_M)^c = \texttt{L}^c \cup (\texttt{Inst}_M)^c$. 
Note that if $\texttt{L}^c$ and $(\texttt{Inst}_M)^c$ are PI languages, then by Theorem \ref{th_closure}, their union $(L_M)^c$ is
a pebble-intervals language.
$(\texttt{Inst}_M)^c$ is regular and therefore also a pebble-intervals language, so it remains to show that 
$\texttt{L}^c$ is accepted by a pebble-intervals automaton. 
\begin{claim}
\label{cl_complement_L_pi}
$\texttt{L}^c$ is a pebble-intervals language.
\end{claim}
\begin{proof}[Proof of Claim~\ref{cl_complement_L_pi}]
\label{proof_cl_complement_L_pi}
Note that for every $w \in \texttt{L}^c$, it holds that there is a prefix of $w$ ending with some $\texttt{jmp}$ such that the number of increments and decrements are not equal. 
The automaton guesses for which register this happens, the prefix, and whether it is the number of increments or decrements that is larger. Then it verifies its guesses using $\pA_{\mathit{Dyck}}$.
\end{proof}
This concludes the proof of Theorem~\ref{th_intersection}.
\end{proof}

}

\newcommand{\proofcorunivinc}{
\begin{proof}
\label{proof_cor_universality_inclusion}
We have seen in the proof above that we can effectively build an automaton $\pA$ for $(L_M)^c$.
We have that $\pA$ has a universal language if and only if $L_M = \emptyset$. 
Thus if we could test for universality, we could again solve the Minsky halting problem.

Since universality easily reduces to inclusion, undecidability of the inclusion problem follows.
\end{proof}
}

\newcommand{\proofgooddecoration}{
\begin{proof}
\label{proof_prop_good_decoration}
If $\cD$ is the empty word $\emptyset_{\dat{\Xi}}$, then the empty word $\emptyset_{\compatTasks}$ over the vocabulary
$\compatTasks$ is the unique $\cD$-task word. It is easy to verify that 
$\emptyset_{\compatTasks}$  is a
completed $\cD$-task word if and only if $\cD \models \varphi_\exists$.
From now on, we assume $\cD$ is not the empty word. 
Note $\cD \models \varphi_\varepsilon$. 

Assume $\cD \models \varphi_\exists$.
Let $\cT$ have the same universe and order relations as $\cD$.
For every $d \in D$, we define the unique $\ts_d \in \compatCompletedTasks$ for which $\cT \models \ts_d(d)$ as follows.
Let $a \in [A]$ be such that $\cD \models \xi_a(d)$.
Since $\cD \models \varphi_\exists$, we have $\cD \models \bigwedge_{b=1}^{B}{\exists y \, \bigvee_{c=1}^{C}{\theta_{a b c}(d,y)}}$. 
Therefore there exist $d_b \in D: b\in B$ and $c_b \in [C]: b \in [B]$
such that $\cD \models \theta_{a b c_b}(d,d_b)$ 
for every $b \in [B]$.
Let $\ts_d = \{C_\theta \mid \theta \in \omega_d\}$,
and set $\cT \models \ts_d(d)$.

The set $\omega_d = \{\theta_{a b c_b} \mid b \in [B]\}$ is a witness type set and $\ts_d$ realizes $\omega_d$.
Hence, $\cT \models \forall x \, \bigvee_{\ts \in \compatCompletedTasks} \ts(x)$.
For every $d\in D$, $\cD\models \xi^\omega(d)$,
and for every $\theta \in \omega_d$, $P_\theta \notin \ts_d$ and $\cD \models \exists y \, \theta(d,y)$. 
Hence, $\cT$ is a $\cD$-task word.

For the other direction, let $\cT$ be completed $\cD$-task word. 
Let $d \in D$ and $a\in [A]$ such that $\cD \models \xi_a(d)$.
Since $\cT$ is a completed task word, there exists 
$\ts_d \in \compatCompletedTasks$ such that $\cT \models \ts(d)$. 
Let $\omega_d$ be the witness type set such that $\ts$ realizes $\omega_d$. 
There exist $c_b \in [C]: b \in [B]$ such that $\omega_d = \{ \theta_{a b c_b} \mid b \in [B] \}$
and $\ts_d = \{ C_{\theta_{a b c_b}} \mid b \in [B] \}$. 
Since $\cT$ is a task word, $\cD \models \exists y\, \theta(d,y)$.

Consequently, 
$\cD \models \bigwedge_{b=1}^{B}{\exists y \, \bigvee_{c=1}^{C}{\theta_{a b c}(d,y)}}$ for every $d\in D$ and $a\in [A]$, and hence $\cD \models \varphi_\exists$. 

\end{proof}

}

\newcommand{\proofuniquetasked}{

Before proving the claim, we recall the conditions of Definition~\ref{def_tasked_word}.
In a \emph{$\cD$-task word}
every $d \in D$ with $\cT \models \ts(d)$ satisfies: 
	\begin{enumerate} 
	\item \label{def_tasked_word_a}
	   $\cD \models \xi^{\omega(\ts)}(d)$,
	\item \label{def_tasked_word_b}
	   for every $\theta \in \omega(\ts)$, $P_\theta \in \ts$ iff $\cD \models \neg \exists y \, \theta(d,y)$.
	\item \label{def_tasked_word_d}
	   $C_\theta \in \ts  \implies \cD \models \exists y \, \theta(d,y)$ 
	\end{enumerate}

\begin{proof}
First we show the existence of such  $\cT_1$. 
We denote the universe of $\cT_1$ by $T_1$.
Let $\cT_1$ be the $\trim{\cD}{1}$-task word given as follows.
For every $d\in T_1$, let $\ts\in\compatTasks$ be such that $\cT_1\models \ts(d)$, 
and let $\ts_1$ be:
\[
\begin{array}{lllllllll}
\ts_1 & = & \{C_\theta \mid \theta \in \omega(\ts), \trim{\cD}{1} \models \exists y \, \theta(d,y)\} \cup \\ &&
\{P_\theta \mid \theta \in \omega(\ts), \trim{\cD}{1} \models \neg \exists y \, \theta(d,y)\}.	
\end{array}
\]
Then $\cT_1\models\ts_1(d)$. Since $\omega(\ts_1)=\omega(\ts)$, Condition~\ref{def_tasked_word_b} in Definition~\ref{def_tasked_word} holds. 
Since $\cT$ is a $\cD$-task word, for every $d\in T_1$ we have $\cD \models \xi^{\omega(\ts)}(d)$, 
implying Condition~\ref{def_tasked_word_a} in Definition~\ref{def_tasked_word} holds. 

It remains to show that $\cT_1$ is unique.	
Assume for contradiction that there exists another $\trim{\cD}{1}$-task word $\tilde{\cT}$ satisfying the statement of the lemma.
Then there is some $d \in D_1$ and distinct $\ts_1,\tilde{\ts}\in\compatTasks$ 
such that 
$\cT_1 \models \ts_1(d)$ and $\tilde{\cT} \models \tilde{\ts}(d)$. 
We have $\omega(\tilde{\ts})=\omega(\ts)=\omega(\ts_1)$, and hence there is $\theta \in \omega(\ts)$ such that
either $P_\theta \in \ts_1 - \tilde{\ts}$
or $C_\theta \in \ts_1 - \tilde{\ts}$. 
In either case, since $\cT_1$ satisfies Condition~\ref{def_tasked_word_b} in Definition~\ref{def_tasked_word}, 
$\tilde{\cT}$ does not satisfy Condition~\ref{def_tasked_word_b}, in contradiction to the assumption that $\tilde{\cT}$ 
is a $\trim{\cD}{1}$-task word. 
\end{proof}
}

\newcommand{\proofcltwo}{
\begin{proof}
\label{proof_cl_2}
It is given that $\cD \models \perfectalpha_{w(\ell_1),w(\ell_2),\leq_1}(d_1,d_2)$. 
Since $\embeddingOp{\dataAbs}{\cT}$ is order-preserving and $\cD \models d_1 \lneq_1 d_2$, we get $\ell_1\lneq \ell_2$.
Since $\cT$ is a task word, there exist $\ts_1,\ts_2 \in \compatTasks$ such that $\cT \models \ts_1(d_1)$ and $\cT \models \ts_2(d_2)$. 
By definition of a task word, we have that $\cD \models \xi^{\omega(\ts_1)}(d_1)$ and $\cD \models \xi^{\omega(\ts_2)}(d_2)$,
and therefore 
$
\cD \models \perfectalpha_{w(\ell_1)}(d_1) \wedge \perfectalpha_{w(\ell_2)}(d_2)
$.

\sloppypar{
We consider one of the cases for $\perfectalpha_{w(\ell_1),w(\ell_2),\lesssim_2}(x,y)$. The other cases can be treated analogously. 
If 
$\perfectalpha_{w(\ell_1),w(\ell_2),\lesssim_2}(x,y) = x \lnsim_2 y$ 
then $w(\ell_1)\notin \Gammatop$ while $w(\ell_2)\in \Gammatop$.
By definition of $w=\dataAbs(\cT)$,
we have that  $\val{\cD}{\embeddingOp{\dataAbs}{\cT}(\ell_1)}$ $<$ $\maxdv_\cD$
and $\val{\cD}{\embeddingOp{\dataAbs}{\cT}(\ell_2)}$ $=$ $\maxdv_\cD$. 
Hence, 
$\cD \models \perfectalpha_{w(\ell_1),w(\ell_2),\lesssim_2}(d_1,d_2)$.

We consider one of the cases for $\perfectalpha_{w(\ell_1),w(\ell_2),S_2}(x,y)$. The other cases can be treated analogously. 
If 
$\perfectalpha_{w(\ell_1),w(\ell_2),S_2}(x,y) = \neg S_2 (y,x)$
then $w(\ell_1)\in \Gammatop$ while $w(\ell_2)\in \Gammarest$.
By definition of $w=\dataAbs(\cT)$,
we have that $\val{\cD}{\embeddingOp{\dataAbs}{\cT}(\ell_1)}$ $=$ $\maxdv_\cD$
and $\val{\cD}{\embeddingOp{\dataAbs}{\cT}(\ell_2)}$ $\leq$ $\maxdv_\cD-2$.  
Therefore, we have 
$ \cD \models \perfectalpha_{w(\ell_1),w(\ell_2),S_2}(d_1,d_2) $. 
}
\end{proof}

}

\newcommand{\proofperfectwotype}{
\begin{proof}
Since  $\perfect_{\alpha,\beta}(x,y) \not\models_{\dat{\Xi}} \chi(x,y) \wedge \chi(y,x)$, there exists 
$\cD$ and elements $d_1,d_2 \in D$ such that $\cD \models \perfect_{\alpha,\beta}(d_1,d_2)$ and $\cD \not\models \chi(d_1,d_2) \wedge \chi(d_2,d_1)$.
By Observation~\ref{lem_perfect_2type} we know that  
for every data word $\cD'$ and every two elements $d_1',d_2' \in D'$ such that
$\cD' \models \perfect_{\alpha,\beta}(d_1',d_2')$, 
the 2-type of $(d_1',d_2')$ is the same as the 2-type of $(d_1,d_2)$. 
Let us denote this 2-type by $\theta(x,y)$. 
Since $\cD \not\models \chi(d_1,d_2) \wedge \chi(d_2,d_1)$, we have that either
$\theta(x,y) \notin \Theta_\forall$ or $\theta(y,x) \notin \Theta_\forall$ 
and therefore also $\cD' \not\models \chi(d_1',d_2') \wedge \chi(d_2',d_1')$.
\end{proof}
}

\newcommand{\proofunitask}{
\label{proof_uni_task}
\begin{proof}
\begin{claim}
\label{lem_perfect_ext}
Let $w \in \Gamma^+$. Then $w$ is perfect if and only if $\ext(w)$ is perfect.
\end{claim}
\begin{proof}
\label{proof_perfect_ext}
Let $s = \ext(w)$.

Assume that $w$ is perfect.
Let $\ell_1 \lneq \ell_2$ be positions in $s$ 
such that at least one of $s(\ell_1)$, $s(\ell_2)$ is in $\Gammatop$. 
For $i=1,2$, let $\ell_i' = \embeddingOp{\ext}{w}(\ell_i)$. We have $w(\ell_i')=s(\ell_i)$ and $\ell_1' \lneq \ell_2'$.
Therefore $\perfect_{s(\ell_1),s(\ell_2)}(x,y) = \perfect_{w(\ell_1'),w(\ell_2')}(x,y)$.
Since $w$ is perfect,
$
\perfect_{s(\ell_1),s(\ell_2)}(x,y) \models_{\dat{\Xi}} \chi(x,y) \wedge \chi(y,x)
$.

Now assume that $s$ is perfect. 
Let $\ell_1' \lneq \ell_2'$ be positions in $w$ 
such that at least one of $w(\ell_1')$, $w(\ell_2')$ is in $\Gammatop$. 
Denote $\ell_1'' = \min\{\ell \mid w(\ell) = w(\ell_1')\}$ and $\ell_2'' = \max\{\ell \mid w(\ell) = w(\ell_2')\}$
and note that $\ell_1'',\ell_2'' \in \extPos(w)$. 
For $i=1,2$, let $\ell_i$ be such that $\ell_i'' = \embeddingOp{\ext}{w}(\ell_i)$.
We have $w(\ell_i') = w(\ell_i'')=s(\ell_i)$ for $i=1,2$ and $\ell_1 \lneq \ell_2$. 
Hence, $\perfect_{w(\ell_1'),w(\ell_2')}(x,y) = \perfect_{s(\ell_1),s(\ell_2)}(x,y)$, and 
since  $s$ is perfect,
$
\perfect_{w(\ell_1'),w(\ell_2')}(x,y) \models_{\dat{\Xi}} \chi(x,y) \wedge \chi(y,x)
$.

\end{proof}

Assume $\cT$ is perfect. Let $d_1$ and $d_2$ be distinct elements of $\cD$. 
We show $\cD \models \chi(d_1,d_2)$.
Let 
$
e = \maxdv_\cD - \max\{\val{\cD}{d_1},\val{\cD}{d_2}\}
$.
Let the universe of $\trim{\cD}{e}$ be $D'$. 
Note that $d_1,d_2\in D'$ and that $\maxdv_{\trim{\cD}{e}} \in \{\val{\trim{\cD}{e}}{d_1},\val{\trim{\cD}{e}}{d_2}\}$. 
Since $\cT$ is a perfect task word, we have that
$\ext(\dataAbs(\ttrim{\cT}{e}))$ is a perfect string, and by Claim~\ref{lem_perfect_ext}, so is $w_e = \dataAbs(\ttrim{\cT}{e})$.
By Lemma~\ref{cl_2},
either $\trim{\cD}{e} \models \perfect_{w_e(\ell_1),w_e(\ell_2)}(d_1,d_2)$ 
or $\trim{\cD}{e} \models \perfect_{w_e(\ell_2),w_e(\ell_1)}(d_2,d_1)$. 
In either case, since $\dataAbs(\ttrim{\cT}{e})$ is perfect, we have that
$\trim{\cD}{e} \models \chi(d_1,d_2)$.
Since $\chi(x,y)$ is quantifier-free and $\trim{\cD}{e}$ is a substructure of $\cD$,
we also have that $\cD \models \chi(d_1,d_2)$.

For the other direction, assume $\cD \models \varphi_\forall$. 
Assume for contradiction that there is $e$ such that
$\ext(\dataAbs(\ttrim{\cT}{e}))$ is not perfect. Then by Claim~\ref{lem_perfect_ext}, 
$w_e = \dataAbs(\ttrim{\cT}{e})$ is also not perfect. 
That is,
there exist positions $\ell_1 \lneq \ell_2$ in $w_e$
such that at least one of $w_e(\ell_1),w_e(\ell_2)$ is in $\Gammatop$, and such that
$
\perfect_{w_e(\ell_1),w_e(\ell_2)}(x,y) \not\models_{\dat{\Xi}} \chi(x,y) \wedge \chi(y,x)
$.
For $i=1,2$, let $d_i = \embeddingOp{\dataAbs}{\ttrim{\cT}{e}}(\ell_i)$. We have $d_1\lneq d_2$. 
By the definition of $\dataAbs$, $\maxdv_{\trim{\cD}{e}} \in \{\val{\trim{\cD}{e}}{d_1},\val{\trim{\cD}{e}}{d_2}\}$.
By Lemma~\ref{cl_2}, we have $\trim{\cD}{e} \models \perfect_{w_e(\ell_1),w_e(\ell_2)}(d_1,d_2)$
and by applying Lemma~\ref{lem_cor_perfect_2type} with $\alpha = w_e(\ell_1)$ and $\beta = w_e(\ell_2)$, we have that $\trim{\cD}{e} \not\models \varphi_\forall$. 
Since $\trim{\cD}{e}$ is a substructure of $\cD$ and $\varphi_{\forall}$ is universal, we also have that $\cD \not\models \varphi_\forall$ in contradiction to our assumption.
\end{proof}
}

\newcommand{\proofcheckconsec}{
\begin{proof}
Given $s_e$ and $s_{e+1}$, 
checking if $(s_e,s_{e+1})$ are consecutive in $\EXPSPACE$ is done as follows. 
We iterate over all $r \in \Gammatop^\star \cap \CGamma^\star$ such that $|r| \leq 7 |\Theta_\exists|$
and over all strictly monotone functions 
$g:[n_r] \to [n_r + n_{s_e}]$.
We search for such $r$ and $g$ for which $s_{e+1} = \rcon{r}{g}{s_e}$ and answer to whether such $r$ and $g$ are found. 
Lemmas~\ref{lem:consec_to_syntactic} and~\ref{lem:syntactic_to_consec} guarantee the correctness 
of a semi-decision procedure behaving as above without restricting the length of $r$, and we show that if there is an $r$ such that $s_{e+1} = \rcon{r}{g}{s_e}$, 
then there is such an $r$ with $|r| \leq 7 |\Theta_\exists|$ (see App.~\ref{app_proof_lem:delete_nonext}).
\end{proof}
}

\newcommand{\proofconsectosyntactic}{
\begin{proof}
We first need some additional notation.
Let
 \[
 \begin{array}{l}
   \begin{array}{llllllllllllllll}
 		w &=& \dataAbs(\cT) \\ 
		X' &=& \extElem(\ttrim{\cT}{1})\\
		X_{\onetoprm} &=& \embeddingOp{\dataAbs}{\cT}(\positions_\onetoprm(w)) \\ 
		X &=& X_{\onetoprm} \cup X' \\
		w_{X} &=& \dataAbs(\cT|_X)\\
		w_{X_\onetoprm} &=& \dataAbs(\cT|_{X_{\onetoprm}})
 	\end{array}
 \end{array}
  \]
 The string $r \in \Gammatop^\star \cap \PGamma^\star$ is obtained from $w_{X_\onetoprm}$ by substituting every letter $(\otop,\ts)$ with 
 $(\otop,\ts_P)$ such that $\ts_P = \{P_\theta \mid \theta \in  \omega(\ts)\}$.  
The function $g$ is given by
\[
g = \embeddingOp{\dataAbs}{\cT|_{X_\onetoprm}}\circ\embeddingOp{\dataAbs}{\cT|_X}^{-1}
\]
Finally, let $w' = \dataAbs(\ttrim{\cT}{1})$, and $w_{X'} = \dataAbs(\cT|_{X'})$. 

By Lemma~\ref{cl_extPos_cardinality},
${\extElem(\cT)} \subseteq X$. 
Hence,
$s = \ext(w_X)$. We will prove that $w_X = \rcon{r}{g}{s^0}$, and the lemma will follow.

Notice that $w_X = w_{X_{\onetoprm}} \shuffle_g w_{X'}$, $|w_{X_\onetoprm}| = |r|$, $|w_{X'}| = |s^0|$,
and $g$ is strictly monotone as the composition of two order-preserving functions.
Hence, $s^0$, $r$, and $g$ are as required in the definition of $\rcon{r}{g}{s^0}$ (Definition~\ref{def:rcon}). 
Let $\ell \in [n_{w_X}]$ and $d=\embeddingOp{\dataAbs}{\cT|_{X}}(\ell)$. 
Let $(\rcon{r}{g}{s^0})(\ell) = (\hlayerp{a},\ts_a)$ and
$w_X(\ell) = (\hlayerp{b},\ts_b)$.
By the construction of $s$ and $\rcon{r}{g}{s^0}$, 
$h_a = h_b$ and $\omega(\ts_a)=\omega(\ts_b)$. 
We need to prove that $\ts_a=\ts_b$. 
Let $s^1$ be as in Definition~\ref{def:rcon}, $s^1(\ell) = (\hlayerp{a^1},\ts_a^1)$,
and $\theta \in \omega(\ts_a)$. 

Assume $C_\theta \in \ts_a$. If $C_\theta \in \ts_a^1$, then 
 $\ell$ is not in the image of $g$, $d$ is an element of $\trim{\cD}{1}$, 
 and there is no element $d'$ of $\trim{\cD}{1}$ such that $\trim{\cD}{1} \models \theta(d,d')$, hence $C_\theta \in \ts_b$. 
Otherwise, there are 
$\ell_1<\ell_2 \in [n_{s^1}]$ such that
$
\perfect_{s^1(\ell_1),s^1(\ell_2)}(x,y) \equiv_{\dat{\Xi}} \theta(x,y)
$
and either $s^1(\ell_1) \in \Gammatop$
or $s^1(\ell_2) \in \Gammatop$, and
either $\theta \models_{\DW(\Xi)} x \leq_1 y$ and $\ell = \ell_1$, or
$\theta \models_{\DW(\Xi)} x >_1 y$ and $\ell = \ell_2$. 
By Lemma~\ref{cl_2}, this implies that 
$\cD \models \exists y\, \theta(d,y)$, and hence 
$C_\theta \in \ts_b$. 
 
 Conversely, assume $C_\theta \in \ts_b$. There is $d'$ in the universe of $\cT$ such that
 $\cD \models \theta(d,d')$. If both $d$ and $d'$ are elements of $\trim{\cD}{1}$, then 
$C_\theta \in \ts_a^1$ and hence $C_\theta\in \ts_a$. 
 Otherwise, we have $\maxdv_\cD \in \{\val{\cD}{d}, \val{\cD}{d'}\}$.
By Lemma~\ref{lem:extemal-witnesses}, we may assume w.l.o.g. that $d' \in \extElem(\cT)$, and hence $d'\in X$. 
Let $\ell' \in  [n_{w_X}]$ such that $d'=\embeddingOp{\dataAbs}{\cT|_{X}}(\ell')$. 
Let $\ell_1<\ell_2 \in [n_{s^1}]$ be such that
$\{\ell,\ell'\}=\{\ell_1,\ell_2\}$. 
Let $d_i = \embeddingOp{\dataAbs}{\cT|_X}(\ell_i)$.
By Lemma~\ref{cl_2}, 
$
\cD \models \perfect_{s^1(\ell_1),s^1(\ell_2)}(d_1,d_2)
$.
By Observation~\ref{lem_perfect_2type}, 
there is a 2-type $\theta'(x,y)$ such that 
$
\perfect_{s^1(\ell_1),s^1(\ell_2)}(x,y) \equiv \theta'(x,y)
$.
The 2-type $\theta'(x,y)$ is the 2-type of $(d_1,d_2)$.
We have $\theta \models x \leq_1 y$ if and only if $\cD\models d_1 \leq_1 d_2$. 
Hence
$\theta(x,y)= \theta'(x,y)$ if $\theta \models x \leq_1 y$,
and 
$\theta(x,y)= \theta'(y,x)$ if $\theta \models x >_1 y$. 
Consequently, $C_\theta \in \ts_b$. 
We get $\ts_a=\ts_b$.

\end{proof}

}

\newcommand{\proofsyntactictoconsec}{

\begin{proof}
Let $s=\ext(\rcon{r}{g}{s^0})$.
Without loss of generality, we may assume the universe $D_0$ of $\cD_0$ is disjoint from $\mathbb{N}$. 
Let $\bar{g}$ be as in the definition of $r \shuffle_g s^0$. 
Let $\cD$ be the data word over $\Xi$ with universe $D_0 \cup [n_r]$
such that:
\begin{enumerate}
 \item $\cD_0$ is the substructure of $\cD$ induced by $D_0$.
 
 \item For every $\ell \in [n_r]$ and $r(\ell)=(\hlayer,\ts)$, $\cD \models \xi^{\omega(\ts)}(\ell)$. 
 
 \item For every $\ell_1,\ell_2 \in [n_r]$, $\cD \models \ell_1 \sim_2 \ell_2$.
 
 \item For every $d\in D_0$ and $\ell \in [n_r]$, $\cD \models d <_2  \ell$.
 
 \item For every $\ell_1, \ell_2 \in [n_r]$, $\cD \models \ell_1 \leq_1 \ell_2$ if and only if $\ell_1 \leq \ell_2$.
 
 \item For every $d\in D_0$, let
  $d'\in D_0$ be the maximal element of $\extElem(\cT_0)$ 
 with respect to $\leq_1$ such that $d' \leq_1 d$, and let $\ell'\in [n_{s^0}]$ be such that
 $d' =\embeddingOp{\ext \, \circ \, \dataAbs}{\cT_0}(\ell')$. 
 For every $\ell_1 \in [n_r]$, 
 $\cD \models d \leq_1 \ell_1$ if and only if $\bar{g}(\ell') < g(\ell_1)$;
 if no such $d'$ exists for $d$ then 
 $\cD \models d \leq_1 \ell_1$.

\end{enumerate}
Let $\cT$ be a $\cD$-task word such that, for every $d \in D_0$,
there are $\ts,\ts_0 \in \compatTasks$ such that $\omega(\ts)=\omega(\ts_0)$,
$\cT\models \ts(d)$, and $\cT_0\models \ts_0(d)$.
Clearly $\ext(\dataAbs(\ttrim{\cT}{1})) = s^0$. 
Let $r$, $w_X$, $w_{X'}$, and $w_{X_\onetoprm}$
be as in Lemma~\ref{lem:consec_to_syntactic}. By the construction of $\cD$, 
\[
g = \embeddingOp{\dataAbs}{\cT|_{X_\onetoprm}}\circ\embeddingOp{\dataAbs}{\cT|_X}^{-1}.
\] 
By Lemma~\ref{lem:consec_to_syntactic}, 
$\ext(\rcon{r}{g}{s^0}) = \ext(\dataAbs(\cT))$, i.e., $s^0$ and $s$ are consecutive extremal strings.

\end{proof}
}

\begin{document}

\title{Pebble-Intervals Automata and \\ \texorpdfstring{FO$^2$}{FO2} with Two Orders\thanks{This work was supported by the Austrian Science Fund (FWF) projects P30360, P30873, and W1255.} \\ (Extended Version)
}

\author{Nadia Labai}
\author{Tomer Kotek}
\author{Magdalena Ortiz}
\author{Helmut Veith\thanks{This article is dedicated to the memory of Helmut Veith, who passed away tragically
while this manuscript was being prepared.}}
\affil{TU Wien, Vienna, Austria}

\date{}

\maketitle        

\begin{abstract}

We introduce a novel automata model, which we call \emph{pebble-intervals automata
(PIA)}, and study its power and closure properties. 
PIAs are tailored for  a decidable fragment of FO that is important 
for reasoning
about structures that use data values from infinite domains:
the two-variable fragment  with one total preorder and its
induced successor relation, one linear order, and an arbitrary number of unary
relations. 
We prove that the string projection of every language of
data words definable in the logic is accepted by a
 pebble-intervals automaton $\pA$, and obtain as a corollary 
 an automata-theoretic proof of the $\EXPSPACE$ upper bound for 
finite satisfiability  due to Schwentick and Zeume.

\end{abstract}

\section{Introduction}

Finding decidable fragments of First Order Logic (FO) 
that are expressive enough for reasoning in different applications is a major
line of research. 
A prominent such fragment is \emph{the two-variable fragment \FOtwo\
of FO}, which has a decidable finite satisfiability problem \cite{Mor75,GKV97}
and is well-suited for handling graph-like structures. 
It captures many \emph{description logics}, which are prominent formalisms for
knowledge representation, and  
several authors have recently applied fragments based on \FOtwo\ 
to verification of programs 
\cite{DBLP:conf/icdt/ItzhakyKRSTVZ17,Ahmetaj17,DBLP:conf/icdt/CalvaneseOS16,calvanese2014shape,Rensink}.
Unfortunately, \FOtwo\ has severe limitations, e.g., it cannot 
express transitivity, and in the applications to verification above, it cannot
reason about programs whose variables range over data values from infinite
domains.
This has motivated the exploration of decidable extensions of \FOtwo\ with special relations which are not axiomatizable in 
\FOtwo. 
For example, finite satisfiability of \FOtwo\ with a linear order was shown to be \NEXPTIME-complete in \cite{otto2001two}, even in the presence 
of the induced successor relation \cite{etessami2002first}, 
and equivalence relations have been used to model data values which can be tested
for equality 
\cite{bojanczyk2011two,DBLP:conf/pods/BojanczykDMSS06,DBLP:conf/lpar/DavidLT10,DBLP:conf/icdt/NiewerthS11}.
However, related extensions of \FOtwo\ with preorders easily become undecidable
 \cite{DBLP:conf/pods/BojanczykDMSS06,DBLP:conf/csl/ManuelZ13}.
Recently 
the logic $\mathrm{FO}^2(\leq_1,\lesssim_2,S_2)$, 
that is \FOtwo\ with a linear order $\leq_1$, a total preorder $\lesssim_2$ and its induced successor $S_2$, 
and any number of unary relations from a finite alphabet,  
was shown to have an $\EXPSPACE$-complete satisfiability problem \cite{ar:zeume12}. 
This logic can compare data values in terms of which is smaller than which and whether they are consecutive in $\lesssim_2$,
making it suitable to model linearly ordered data domains, and 
a good candidate for 
extending existing verification methods which use two-variable logics.
We  continue the study of
$\mathrm{FO}^2(\leq_1,\lesssim_2,S_2)$, and in particular, focus on
a suitable automata model for it.
Establishing a connection to suitable automata 
for fragments of FO that can talk about values from infinite domains is an
active area of research. Automata are also important  in automated verification, where they are used, for
example, to reason about temporal properties of program traces
\cite{DBLP:conf/lics/VardiW86,DBLP:books/handbookMC}.
We make the following contributions:
\begin{compactitem}
\item As an  automata model for $\mathrm{FO}^2(\leq_1,\lesssim_2,S_2)$ we
  propose \emph{pebble-intervals automata} (PIAs).  
Similarly to  
classical 
finite-state automata, PIAs are read-once
automata for strings.
However, they read the input in varying  order.
Using  
a fixed set of pebbles $[m]=\{1,\ldots,m\}$, a PIA reads a position $p$ 
by choosing three pebbles $i,j,k \in  [m]$ and non-deterministically 
moving $k$ to position $p$ between the positions of $i$
and $j$.

\item We study the computational power and closure
      properties of PIAs. We describe a restricted  
class of  PIAs
that accept exactly the regular languages, and show that
 some context-free languages, and 
even languages which are not context-free, are accepted by PIAs.
We prove that PIAs are effectively 
closed under union, concatenation, Kleene star, shuffle, and iterated
shuffle, \new{but not effectively 
closed under intersection, even with regular
  languages, nor under complement.}
\item
\new{We show that the emptiness problem for PIA is $\NL$-complete if the number of pebbles is logarithmic in the size of the
automaton, and is $\PSPACE$ in general.} 
\item 
We 
show that PIAs contain $\mathrm{FO}^2(\vocDataWords)$
in the following sense: 
for each 
sentence $\psi$, there is 
a PIA 
whose language coincides with the \emph{projection language} of $\psi$,
obtained by omitting $\lesssim_2$ and $S_2$ from the structures satisfying $\psi$. 
\item As a corollary, we get an automata-theoretic proof for $\EXPSPACE$ membership of finite satisfiability for
 $\mathrm{FO}^2(\vocDataWords)$ that was established in \cite{ar:zeume12}. 
\end{compactitem}

\section{Pebble-Intervals Automata}

In this section, we introduce pebble-intervals automata (PIA). 
We study their emptiness problem, their expressive power, 
and closure properties of the languages they accept.

Let $[n] = \{1, \ldots, n \}$. A \emph{string} of length $n\geq
0$ over alphabet $\Sigma$ is a
mapping $w: [n] \rightarrow \Sigma$, written also $w = w(1) \cdots w(n)$.
Note that $[0] = \emptyset$ and $w : [0] \to \Sigma$ is the empty string
$\varepsilon$. 
We often use $s$, $u$, $v$, and $w$ for strings, and  $|w|$ for the length of
$w$. 

\medskip

A PIA is equipped with a finite number $m$ of pebbles.
It begins its computation with no pebbles on the input $w$, and
uses  \textsc{move} transitions to place and replace pebbles. 
In a $\Move{k}{i}{j}$ transition, 
the pebble $k$ (which may or may not have been previously placed on $w$) is
non-deterministically placed on a previously unread position 
in the interval between pebbles $i$ and $j$. 
The input boundaries can be used as interval  boundaries, e.g.,
a $\Move{k}{i}{\lhd}$ transition places pebble $k$ to the right of
pebble $i$.
For convenience we allow \emph{silent} transitions that go to a new
state without moving any pebbles. 
As pebbles can only be placed on unread
positions, each position of $w$ is read at most once.  
In an accepting run all positions must be read, and the run must end at
an accepting state.

\begin{definition}[Pebble-intervals automata]
\label{def_pebble_intervals_automata}
A \emph{PIA} $\pA$ is a tuple $(\Sigma, m, Q, q_{\mathrm{init}}, F, \delta)$, 
where $\Sigma$ is the (finite) alphabet, $m \in \bN$,
 $Q$ is the finite set of  states, $q_{\mathrm{init}} \in Q$ is the initial state, $F
 \subseteq Q$ are the accepting states, and $\delta \subseteq (Q \times Q)
 \cup (Q \times \MoveS_m \times \Sigma \times Q)$ is the transition relation 
with 
$
\MoveS_m = \{\Move{k}{i}{j} \mid i \in [m] \cup \{\rhd\}, j \in [m] \cup \{\lhd\} ,k \in [m], i \neq j 	\}.
$
We may omit $m$ when it is clear from the context.
Transitions in $Q \times \MoveS \times \Sigma \times Q$ are \textsc{move} transitions, 
and transitions in $Q \times Q$ are \emph{silent} transitions. 
The \emph{size} of $\pA$ is $|\delta| + |\Sigma| + |Q|$. 
\end{definition}

The positions of $m$ pebbles on a string of length $n$ during a run of a PIA are described by an 
\emph{$(m,n)$-pebble assignment}, which 
is a function $\newtau:[m]\to [n] \cup
\{\bot\}$ with either 
 $\newtau(i) \not= \newtau(j)$ or $\newtau(i) = \newtau(j) = \bot$
for each $1\leq i < j \leq m$; 
the pebbles $j$ with $\newtau(j) = \bot$ are
unassigned. We define $\rho_\bot$ as $\rho_{\bot}(i) = \bot$ for every $i \in [m]$.
By $\hat{\newtau}:[m] \cup \{\rhd, \lhd\} \to \{0\} \cup [n+1]$ we denote 
the extension of 
$\tau$ with 
$\hat{\newtau}(\rhd) = 0$ and $\hat{\newtau}(\lhd) = n+1$.

\begin{definition}[Semantics of PIAs]
\label{def_semantics}
Consider a PIA $\pA=(\Sigma, m, Q, q_{\mathrm{init}}, F, \delta)$. 
A \emph{configuration} of $\pA$ on string $u \in \Sigma^\star$ is a triple $(q,\rho,N)$
where 
	$q\in Q$ is the current state,  
	$\rho:[m]\to [\length{u}]\cup \{\bot\}$ is the current pebble assignment,  
	and 
	$N \subseteq [\length{u}]$ is the set of already-read positions. 
The \emph{initial} configuration $\pi_{\mathrm{init}}$
is $(q_{\mathrm{init}}, \rho_{\bot}, \emptyset)$.
A configuration $(q,\rho,N)$ is \emph{accepting} if $q\in F$ and $N = [\length{u}]$.
Let $\pi = (q,\rho,N)$ and
$\pi' = (q',\rho',N')$ be configurations  
on $u$. 
We call them 
\emph{consecutive} and write $\pi \overset{t}{\rightsquigarrow} \pi'$
if there exists a transition $t$ in $\delta$ such that either: 
\begin{compactenum}
 \item $t$ is a silent transition of the form $(q,q')$, $N = N'$, and  $\rho = \rho'$; or 
 \item $t$ is a \textsc{move} transition $(q,\Move{k}{i}{j},u(\ell),q')$  
   with $\hat{\rho}(i) < \ell <
   \hat{\rho}(j)$ and $\ell \in [\length{u}]-N$, and additionally  
   $\rho' = \rho[k\mapsto \ell]$ and $N' = N \cup \{\ell\}$. That is, 
pebble $k$ is placed on position $\ell$  
in the open interval between $i$ and $j$, reading the letter
                              $u(\ell)$.  
\end{compactenum}
Let 
$\bar{t}=(t_1,\ldots,t_r)$ 
and $\bar{\pi} =
(\pi_0,\ldots,\pi_r)$ be sequences of transitions and configurations. 
We call $(\bar{t},\bar{\pi})$ a \emph{computation of $\pA$ on $u$} if
 $\pi_0 = \pi_{\mathrm{init}}$ and
$\pi_{i-1} \overset{t_i}{\rightsquigarrow} \pi_i$ for every $i\in [r]$, and write
$\pi_0 \overset{\bar{t}}{\rightsquigarrow} \pi_r$. 
\new{We call $(\bar{t},\bar{\pi})$ \emph{accepting} if $\pi_r$ is accepting.}
We write $\pi \overset{\star}{\rightsquigarrow} \pi'$ if $\pi
\overset{\bar{t}}{\rightsquigarrow} \pi'$ for some  $\bar{t}$. 
The automaton $\pA$ \emph{accepts} $u$ if there is an accepting  computation of $\pA$ on
$u$. 
The set of all $u$ accepted by $\pA$ is denoted 
$L(\pA)$, and called a \emph{PI language}.  
\end{definition}

\smallskip\noindent{\bf Computational power.} 
PIAs generalize standard  non-deterministic finite-state automata. 
A PIA  $\pA=(\Sigma, 1, Q, q_{\mathrm{init}}, F, \delta)$ 
with one pebble is \emph{unidirectional} 
if $q_{\mathrm{init}}$ has no incoming transitions, 
and the $\MoveNoArgs$ transitions from other states 
use $\Move{1}{1}{\lhd}$ only.

\begin{proposition}
\label{prop:regular-languages}
A  language $L$ is accepted by a standard 
non-deterministic
finite-state automaton 
iff 
$L = L(\pA)$ for a unidirectional PIA $\pA$  
with the same number of states. 
\end{proposition}

\noindent
PI languages also contain  non-regular languages, and even some
non-context-free ones. 

\begin{examples}
\label{ex_languages}
The following are examples of PI languages: 
\begin{compactenum}
 \item 
 There is a PIA $\pA_{\mathit{Dyck}}$ with one pebble that accepts
the Dyck language $L_{\mathit{Dyck}}$ of well-nested brackets, which  is context-free but not regular. 
  The alphabet has two letters $\lbbrck$ and $\rbbrck$, and the
 states are $q_{\lbbrck}$ and $q_{\rbbrck}$.
 The initial and only accepting state is $q_{\rbbrck}$.
 The transition relation contains
 $(q_{\rbbrck},\Move{1}{\rhd}{\lhd},\lbbrck,q_{\lbbrck})$
 and $(q_{\lbbrck},\Move{1}{1}{\lhd},\rbbrck,q_{\rbbrck})$. 
 $\pA_{\mathit{Dyck}}$ accepts a string iff there are as many left as right brackets,  and no prefix has more right than left brackets. 
 \item
A similar one pebble PIA 
accepts the language $L_{\mathit{two}}$  of all strings  of two types of parentheses, where 
	each type is well-nested with respect to itself, but not necessarily 
 to the other type. 
	E.g., $(\, [\, )\, ] \in L_{\mathit{two}}$, but $(\, ] \notin L_{\mathit{two}}$.
	$L_{\mathit{two}}$ is not context-free.
 \item
	$\{ a^n \$ b^n \# c^n \mid n \geq 0 \}$, which is not context-free, is accepted by a PIA with $3$ pebbles.
	Pebbles $1$ and $2$ read the $\$$ and the $\#$, and then the PIA
	keeps doing the following: 
	pebble $3$ reads an $a$ to the left of pebble $1$, a $b$ between pebbles $1,2$, and a $c$ to the right of pebble $2$.
 \item
	$\{ w\$ w \mid w \in \{0,1\}^+ \}$ is not context-free, and is accepted by a PIA with $3$ pebbles.
	Pebble $1$ reads the $\$$, pebble $2$ reads a letter $\sigma$ to the left of pebble $1$, and pebble $3$ also reads $\sigma$ to the right of pebble $1$.
	Then the PIA repeats: 
	\begin{enumerate*}[label=(\roman*)]
		\item a letter $\sigma$ is non-deterministically chosen,
		\item pebble $2$ reads $\sigma$ between its current position and pebble $1$, and
		\item pebble $3$ reads $\sigma$ to the right of its current position. 
	\end{enumerate*}
	Similar languages are PI languages, e.g., $\{w w^R w w \mid w \in
        \Sigma^\star\}$, where $w^R$ is $w$ in reverse.
\end{compactenum}
\end{examples}

We conjecture that not all context-free language are PI
languages; e.g, 
the  Dyck language of two types of  well-nested parentheses seems not to be PI.

\smallskip\noindent{\bf Closure properties.} We provide a construction of suitable PIAs
in the appendix to show the following.
\begin{theorem}
The class of PI  languages is  effectively 
closed under union, concatenation, Kleene-$\star$, shuffle, and iterated shuffle.
It is not effectively closed under
intersection, even with regular languages, nor under 
complement.
\end{theorem}

From the construction used in the proof of the above theorem, we also obtain:
\begin{corollary}
\label{cor_universality_inclusion}
The universality and inclusion problems for PIAs are undecidable.
\end{corollary}

\smallskip\noindent{\bf Emptiness.}
For deciding whether $L(\pA)\! \neq \!\emptyset$ for a given PIA, we use 
\emph{feasible sequences of transitions}, which are 
those that  correspond to an actual computation of a PIA. 
One can show that for a given PIA with $m$ pebbles,  $L(\pA) \neq \emptyset$ iff there is a feasible sequence of transitions $\bar{t}$ 
of length at most $|\pA|\cdot 2^{O(m \log m)}$, and that the existence of the latter  can be guessed and verified using a bounded amount of information (roughly a counter,  two transitions, and two pebble assignments). 
This gives us the upper bounds below, which hold also if $\pA$ is not given explicitly, as long as
$\delta$ can be computed  non-deterministically in $\log(|\pA|)$ space.  
For the case where $\pA$ has $O(\log |\pA|)$ pebbles, 
$\NL$-hardness follows from the same result for standard finite state automata and Prop.~\ref{prop:regular-languages}. 

\begin{theorem}
\label{th_pia_emptiness}
If a PIA $\pA$ has $O(\log |\pA|)$ pebbles,
its emptiness problem is $\NL$-complete. In general, the emptiness problem for PIA 
is in $\PSPACE$. 
\end{theorem}

\subsubsection*{Related automata models.} 
\emph{Jumping finite automata} \cite{DBLP:journals/ijfcs/MedunaZ12}
are probably the closest to PIAs: they  are   
essentially PIAs with one pebble,  which is placed on  an arbitrary unvisited position
without specifying an interval. 
In the context of languages with infinite alphabets, various automata models have been proposed that run on \emph{data words}: string words where values from an infinite domain are attached to each position.
Register automata are finite-state machines on data words which use registers to compare whether data values are equal  
 \cite{kaminski1994finite,neven2004finite,DBLP:journals/iandc/BouyerPT03}; their string  projection languages are regular.
 Pebble automata \cite{neven2004finite} use pebbles in a stack discipline to test for equality of data values.
Data automata \cite{DBLP:conf/pods/BojanczykDMSS06,bojanczyk2011two,DBLP:conf/lics/BojanczykMSSD06} are an extension of register automata introduced to prove the decidability of satisfiability of \FOtwo\ on words with a linear order, a successor relation, and an equivalence relation. 
Their projection languages are accepted by multicounter automata, which are finite automata on strings extended with counters, that are equivalent to Vector Addition Systems or Petri Nets \cite{DBLP:conf/ac/Esparza96}.
Class Memory Automata \cite{DBLP:journals/tcs/BjorklundS10} have the same expressive power as data automata.
Variable Finite Automata \cite{DBLP:conf/lata/GrumbergKS10} extend finite
state automata with variables  from 
an infinite alphabet.
Many works have studied these automata models and their 
variations, see \cite{segoufin2006automata} and \cite[Chapter 4]{KaraPhD} for surveys.

\section{PIAs and \texorpdfstring{$\mathrm{FO}^2(\leq_1,\lesssim_2,S_2)$}{FO2}}

To establish the relation between $\mathrm{FO}^2(\leq_1,\lesssim_2,S_2)$
and PIAs, we need some preliminaries. 
Recall that a total preorder $\lesssim$ is a transitive total relation which can be seen as an equivalence relation
whose equivalence classes are linearly ordered. 
We use $x \sim_2 y$ as shorthand for $(x \lesssim_2 y) \land (y \lesssim_2 x)$.
The induced successor relation $S$ of a total preorder $\lesssim$ 
is such  that $S(x,y)$ if $x \lesssim y$ and there is no $z$ such that $x \lnsim z \lnsim y$.

Two-variable logic (\FOtwo) is the restriction of FO to formulas that only use
two variables $x$ and $y$, and $\mathrm{FO}^2(\leq_1,\lesssim_2,S_2)$ is 
\FOtwo\ with a linear order $\leq_1$, a total preorder $\lesssim_2$ and its
induced successor $S_2$, and any number of unary relations from a finite alphabet.

All structures and strings in this paper are finite.
For a structure $\cA$, we denote its universe by $A$ and its size by $|A|$. 
The \emph{empty structure} has $A = \emptyset$ and is denoted $\emptystructure{\voc}$.

\smallskip
\noindent{\sffamily \bf Data words.}
Let   $\Sigma$ a finite alphabet. Its \emph{extension for data words} is
$\dvoc{\Sigma} = \langle \leq_1, 
\lesssim_2, S_2, \sigma: \sigma \in \Sigma \rangle$. 
A \emph{data word} over $\Sigma$ is a finite $\dvoc{\Sigma}$-structure $\cD$ with universe $D$
where $\sigma: \sigma \in \Sigma$ are interpreted as unary relations that partition $D$.
We use $\cD$, $\cD'$, etc.~to denote data words.  
The empty word is denoted by $\emptyset_{\dat{\Sigma}}$, and 
the class of all data words over $\Sigma$ by $\dat{\Sigma}$.
A set of data words is called a \emph{data language}.  

Let $\varphi_1,\varphi_2$ be $\FO^2(\dvoc{\Sigma})$ formulas.
We write $\varphi_1 \models_{\dat{\Sigma}} \varphi_2$ 
if $\cD \models \varphi_1$ implies $\cD \models \varphi_2$ for every 
$\cD \in \dat{\Sigma}$, and define  equivalence $\equiv_{\dat{\Sigma}}$ analogously. 
We may omit $\Sigma$ 
if clear from context.
The \emph{data value} $\val{\cD}{d}$ of an element $d \in D$ is the number  of equivalence classes $E$
of $\sim_2$ whose elements $d'\in E$ satisfy $d' \lesssim_2 d$, and 
$\maxdv_\cD = \max_{d\in D}\val{\cD}{d}$.
The \emph{string projection} of $\cD$, denoted $\project(\cD)$, is the string $w$ of length $|w| = |D|$ where for all $\ell \in [|w|]$,
$w(\ell) = \sigma$ if and only if $\cD \models \sigma(d)$
where $d$ is the unique element of $\cD$ such that $\ell =  |\{d' \in D \mid \cD \models d' \leq_1 d\}|$.
The projection of the empty structure $\emptystructure{\dvoc{\Sigma}}$, and
only of $\emptystructure{\dvoc{\Sigma}}$,   
is $\varepsilon$.
The \emph{projection language} of a data language $\Delta$ is the string language
$L(\Delta) = \{ w \mid w = \project(\cD) \mbox{ for some } \cD \in \Delta \}$.
If a formula $\psi$ defines $\Delta$, we write $L(\psi)$ for $L(\Delta)$.

\begin{example}
To avoid ambiguity, in our running examples
we use underlined symbols.  
Let $\ul{\Xi} = \{\ul{\xi}_1, \ul{\xi}_2\}$ be a set of unary relations
and
let $\ul{\cD}$ be the data word with universe $\ul{D} = \{\ul{a},\ul{b},\ul{c},\ul{d},\ul{e},\ul{f}\}$ 
where $\leq_1$ is the lexicographic order, the interpretation of $\ul{\xi}_1$
is $\{\ul{a},\ul{b},\ul{c},\ul{e}\}$, the interpretation of  $\ul{\xi}_2$ is $\{\ul{d},\ul{f}\}$, and 
$\ul{b} \lnsim_2 \ul{a} \lnsim_2 \ul{e} \lnsim_2 \ul{c} \lnsim_2 \ul{d} \sim_2 \ul{f}$.
Note e.g. that $\ul{\cD} \models S_2(\ul{a},\ul{e})$ and $\ul{\cD} \models \neg S_2(\ul{b},\ul{e}) \wedge (\ul{b} \lesssim_2 \ul{e})$.
The string projection of $\ul{\cD}$ is $\project(\ul{\cD}) = \ul{\xi}_1 \ul{\xi}_1 \ul{\xi}_1 \ul{\xi}_2 \ul{\xi}_1 \ul{\xi}_2$.
\end{example}

The goal of this section is to prove the following theorem.
\begin{theorem}
\label{th:containment}
If $\psi$ is a 
$\mathrm{FO}^2(\leq_1,\lesssim_2,S_2)$
sentence, there is a PIA $\pA$ 
with $L(\psi)=L(\pA)$.
\end{theorem} 

To prove this, we rely on the normal form defined next.  
A \emph{1-type} $\nu(x)$ over $\dvoc{\Sigma}$ is a maximal consistent conjunction of atomic and 
negated atomic formulas 
with the free variable $x$. A 
\emph{2-type} $\theta(x,y)$ is defined similarly.
Given a $\FO^2(\dvoc{\Sigma})$ formula $\psi$, we obtain a $\varphi$ in normal form  by
taking the Scott Normal Form \cite[Theorem 2.1]{gradel1999logics}  of $\psi$, 
and  
expanding the quantifier-free formulas to Disjunctive Normal Form, and in fact to disjunctions of $2$-types $\theta$.
The Scott Normal Form of $\psi$ introduces linearly many new symbols,
resulting in an extended $\Sigma'$. 
We let $\Xi = \{\xi_a \mid a\in [A]\}$ be an alphabet containing a symbol for
every 1-type over 
$\Sigma'$.

\begin{theorem}[Normal Form]
\label{th_our_snf}
\label{lem_th_our_snf}
Let $\psi \in \FO^2(\dvoc{\Sigma})$. 
Then 
there exist $A,B,C \in \mathbb{N}$,
an alphabet $\Xi = \{\xi_a \mid a \in [A]\}$, 
a formula  $\varphi \in \FO^2(\dvoc{\Xi})$ of the form  
$\varphi = \varphi_\forall \land \varphi_\exists$
and a letter-to-letter substitution $h:\Xi\to\Sigma$
such that  $L(\psi) = h(L(\varphi))$,  
\vspace{-0.3\baselineskip}
{\small
\[
\varphi_\forall = \forall x \forall y \bigvee_{\theta \in \Theta_\forall} \theta(x,y)
\qquad \qquad 
\varphi_\exists  = \varphi_\varepsilon \land \forall x \, \bigwedge_{a \in [A]}\xi_a(x) \rightarrow   
\bigwedge_{b \in [B]}{\exists y \, \bigvee_{c \in [C]}{\theta_{a b c}(x,y)}} \vspace{-0.3\baselineskip}
 \] 
}
with $\theta$ and $\theta_{abc}$ 2-types over $\dvoc{\Xi}$, and 
$\varphi_\epsilon =  \mathit{True}$ if $\emptystructure{\dvoc{\Sigma}} \models_{\dat{\Sigma}} \psi$ and 
$\varphi_\epsilon = \exists x\, (\mathit{True})$ if $\emptystructure{\dvoc{\Sigma}} \not\models_{\dat{\Sigma}} \psi$. 
Moreover, 
$\varphi$ is computable in $\EXPSPACE$ and is of length exponential in $|\psi|$. 
\end{theorem}

\noindent
We let $\Theta_\exists = \{ \theta_{abc} \mid a \in [A], b\in [B], c \in [C]\}$ and 
$\Theta = \Theta_\forall \cup \Theta_\exists$.
Given $a \in~[A]$, a \emph{witness type set for $a$}  is a choice of 2-types satisfying the right-hand side of the implication for $\xi_a$. That is,
a set of 2-types $\omega \subseteq \Theta_\exists$ that contains one $\theta_{abc}$ for every $b\in [B]$, representing a choice of the existential constraints an element needs to fulfill.
Denote by $\Omega_a$ the set of witness type sets for $a$ and let $\Omega = \bigcup_{a \in [A]}\Omega_a$. 
For a witness type set $\omega \in \Omega$, let $\omega(x) = \bigwedge_{\theta
  \in \omega}{\exists y \, \theta(x,y)}$ be its \emph{existential constraints}.
\new{Note that $\omega(x)$ is always satisfiable and that there is a unique letter $\xi^{\omega} \in \Xi$ such that $\omega(x) \models_{\dat{\Xi}} \xi^\omega(x)$.}

\begin{example}
Consider the following formula $\ul{\varphi}$ given in normal form \\[1mm]
{\small $\forall x \forall y \, \ul{\chi}(x,y) \wedge  
\forall x \, \big(\ul{\xi_1}(x) \rightarrow \exists y \left( \ul{\theta}_1(x,y) \vee \ul{\theta}_3(x,y) \right) \wedge 
\ul{\xi}_2(x) \rightarrow \exists y \left( \ul{\theta}_2(x,y) \vee
  \ul{\theta}_4(x,y) \right)\big) $ } \\[1mm]
where
$\ul{\chi}(x,y)$ is the disjunction of 2-types equivalent to $(\ul{\xi}_2(x) \wedge \ul{\xi}_2(y)) \rightarrow x \sim_2 y$, and the 
 $\ul{\theta}_i$  are given as the following 2-types (omitted
clauses are negated): \vspace{-0.5\baselineskip}
{\small
\[
\begin{array}{r@{\,=\,}l@{\quad}r@{\,=\,}l}
\ul{\theta}_1 & x <_1 y \wedge S_2(x,y) \wedge \ul{\xi}_1(x) \wedge \ul{\xi}_2(y) & 
\ul{\theta}_3 & x <_1 y \wedge \neg S_2(x,y) \wedge x \lesssim_2 y \wedge \ul{\xi}_1(x) \wedge \ul{\xi}_2(y) \\
\ul{\theta}_2 & y <_1 x \wedge S_2(y,x) \wedge \ul{\xi}_2(x) \wedge \ul{\xi}_1(y) &
\ul{\theta}_4 & y <_1 x \wedge \neg S_2(y,x) \wedge y \lesssim_2 x \wedge \ul{\xi}_2(x) \wedge \ul{\xi}_1(y)  
\end{array}   \vspace{-0.2\baselineskip}
\] 
}
\noindent 
A data word satisfies $\ul{\varphi}$ iff it is the empty structure, or
{\it (a)}
the largest element of $\leq_1$ has letter $\ul{\xi}_2$,
	{\it (b)}
	the smallest element of $\leq_1$ has letter $\ul{\xi}_1$,
	{\it (c)}
	all elements with  $\ul{\xi}_2$ have maximal value,  and
	{\it (d)}
	no element with $\ul{\xi}_1$ has maximal value.

Note that $\ul{\cD} \models \ul{\varphi}$.
The projection language $L(\ul{\varphi})$ is the regular language with regular expression $\ul{\xi}_1 (\ul{\xi}_1 + \ul{\xi}_2)^\star \ul{\xi}_2 + \varepsilon$.
We have $\ul{\Theta}_\exists = \{ \ul{\theta}_1, \ul{\theta}_2, \ul{\theta}_3, \ul{\theta}_4 \}$.
For $\ul{\varphi}$, we have $A = 2$, $B=1$, and $C=2$.
The witness type sets of $\ul{\varphi}$ are $\{\theta_{111}\}$, $\{\theta_{112}\}$, $\{\theta_{211}\}$, and $\{\theta_{212}\}$, where $\theta_{111} = \ul{\theta}_1$, 
$\theta_{112} = \ul{\theta}_3$, $\theta_{211} = \ul{\theta}_2$, and $\theta_{212} = \ul{\theta}_4$.  
Hence, we have $\ul{\Omega} = \{\{\ul{\theta}_{1}\}, \{\ul{\theta}_{2}\}, \{\ul{\theta}_{3}\}, \{\ul{\theta}_{4}\}\}$, 
and $\xi^{\{\ul{\theta}_1\}} = \xi^{\{\ul{\theta}_3\}} = \ul{\xi}_1$, and $\xi^{\{\ul{\theta}_2\}} = \xi^{\{\ul{\theta}_4\}} = \ul{\xi}_2$.
\end{example}

We construct a PIA $\pA^\varphi$ 
that   
accepts a string $w$  
iff it can be extended into a data word $\cD$ that satisfies the normal form
$\varphi$ of a given sentence $\psi$. 
Note that $\psi$ and $\varphi$ have different alphabets, 
but since there is a letter-to-letter substitution $h$ such that
$L(\psi)=h(L(\pA^\varphi))$, and  
PIAs are closed under  letter-to-letter substitutions, this proves
Theorem~\ref{th:containment}. 

For constructing our PIA, we 
first focus on 
the existential part, i.e.,  
whether $w$ can be extended into a $\cD$ that satisfies $\varphi_\exists$. 
This is achieved in two steps: 
\noindent{\bf (S1)}  We reduce the existence of $\cD$ to the existence of a
sequence of \emph{consecutive task words},  data words
that store additional information of already satisfied vs.\,`promised'  
 subformulas; 
the sequence should lead to a \emph{completed} task
  word where all promises are  fulfilled.
\noindent{\bf (S2)}
We do not have a bound on the length of task words and 
their data values, so we use \emph{extremal strings} to decide the existence of the  desired sequence with the limited memory of PIAs.  
After these two steps,  
we introduce \emph{perfect} extremal strings to  
guarantee the satisfaction of $\varphi_\forall$. 
Our PIA will then decide if a sequence of perfect extremal
strings exists. 

\smallskip\noindent
{\bf Task words for \texorpdfstring{$\varphi_{\exists}$}{varphi-exists}}
\label{sec_fo_data}
We start by defining \emph{task words}, which are like data words but do more book-keeping. Additionally to data values, elements in task words 
are assigned \emph{tasks}, which are witness type sets where each 2-type may be
marked as
\emph{completed} if its satisfaction has already been established, or as
\emph{promised} otherwise. 
We reduce the satisfaction of  $\varphi_{\exists}$ to the existence of a
sequence of $\cT_1, \ldots, \cT_n$ of
\emph{consecutive} task words,  where we keep  
assigning new data values 
and  updating promised into completed tasks, until we reached a \emph{completed} task word  
$\cT_n$. 

\begin{definition}[Tasks]  
For $\theta \in \Theta_{\exists}$, we call $C_\theta$ a \emph{completed task} 
and $P_\theta$ a \emph{promised task}. 
Let 
$\CTasks  \,{=}\, \{C_\theta \mid \theta \in \Theta_{\exists}\}$, 
$\PTasks \,{=}\, \{P_\theta \mid \theta \in \Theta_{\exists}\}$ and 
$\Tasks \,{=}\, \CTasks \cup \PTasks$.
\end{definition}

For each task set  $\ts \subseteq \Tasks$, there is at most one witness type set
$\omega \in \Omega$ that $\ts$ \emph{realizes}, which means that 
 for every $\theta \in \Theta_{\exists}$, 
 {\bf (1)} $|\{P_\theta,C_\theta\} \cap \ts| \leq 1$, and 
 {\bf (2)} $|\{P_\theta,C_\theta\} \cap \ts| = 1$ if and only if $\theta \in \omega$. 
If there is such an $\omega$, we denote it $\omega(\ts)$, and call $\ts$ 
an \emph{$\Omega$-realization}. The set of all 
$\Omega$-realizations is 
 $\compatTasks$, and 
$\compatCompletedTasks = \compatTasks \cap 2^\CTasks$  and 
$\compatPromisedTasks = \compatTasks \cap 2^\PTasks$. 

\begin{example}
Since the witness type sets in $\ul{\Omega}$ are singletons, so are 
the 
$\ts \in 2^{\Tasks}_{\ul{\Omega}}$. 
Let $\ul{ts}_i^C = \{ C_{\ul{\theta}_i} \}$ and $\ul{ts}_i^P = \{ P_{\ul{\theta}_i} \}$ for $i \in [4]$. Then
we have $2^{\Tasks}_{\ul{\Omega}} = \{ \ul{ts}_i^C \mid i \in [4]\} \cup \{ \ul{ts}_i^P \mid i \in [4]\}$, 
and $\{C_{\ul{\theta}_i}\}$ and $\{P_{\ul{\theta}_i}\}$ are $\{\theta_i\}$-realizations for $i \in [4]$.
\end{example}

\emph{$\cD$-task words} are data words that assign tasks to the elements of $\cD$.
More precisely,
each $d \in D$ is assigned, instead of a letter $\xi_a$, a task set $\ts$ that realizes a witness type set $\omega$
which  contains $C_\theta$ for each $\theta \in \omega$ that $d$ satisfies,
and $P_\theta$ for the remaining $\theta \in \omega$.

\begin{definition}[Task word]
\label{def_tasked_word}
\label{def_completed_tasked_word}
Let $\cD$ be a data word over $\Xi$. 
A \emph{$\cD$-task word} is a data word $\cT$ over $\compatTasks$
that has the same universe and order relations as $\cD$,
and for every $d \in D$ with $\cT \models \ts(d)$,
	{\bf {(1)}} $\cD \models \xi^{\omega(\ts)}(d)$, and
	{\bf {(2)}} for every $\theta \in \omega(\ts)$, $C_\theta \in \ts$ iff $\cD \models \exists y \, \theta(d,y)$.
A \emph{task word} $\cT$ is a $\cD$-task word for some $\cD$, and it is \emph{completed} if 
$\cT \models \varphi_\varepsilon \land \forall x \, \bigvee_{\ts \in \compatCompletedTasks} \ts(x)$. 
\end{definition}

\begin{example}
We define a $\ul{\cD}$-task word $\ul{\cT}$; its
 vocabulary is 
 $2^{\Tasks}_{\ul{\Omega}}$,
its universe is $\{\ul{a},$ $\ul{b},$ $\ul{c},$ $\ul{d},$ $\ul{e},$ $\ul{f}\}$, 
and
 $\leq_1$, $\lesssim_2$, and $S_2$ are the same as in $\ul{\cD}$.
The interpretation of the letter $\ul{ts}_1^C$ is $\{ \ul{c} \}$, 
that of $\ul{ts}_2^C$ is $\{ \ul{d} \}$,
that of $\ul{ts}_3^C$ is $\{ \ul{a},\ul{b},\ul{e} \}$,
and that of $\ul{ts}_4^C$ is $\{ \ul{f} \}$; the other  letters are empty.
As $\ul{\cD} \models \ul{\varphi}$, all existential constraints are satisfied and
$\ul{\cT}$ is completed.
\end{example}

\noindent
The satisfaction of $\varphi_\exists$ 
coincides with the existence of a completed task word.

\begin{lemma}
\label{lem_good_decoration}
Let $\cD \in \dat{\Xi}$. There exists a completed $\cD$-task word
iff $\cD \models \varphi_\exists$.
\end{lemma}

We now characterize the notion of \emph{consecutive} task words using 
\emph{trimmings}.
\begin{definition}[Trimming, consecutiveness]
\label{def_trimmed_tasked_word}
\label{def_trimming}
The \emph{trimming} of a data word $\cD$, 
denoted $\trim{\cD}{1}$, 
is the substructure of $\cD$ induced by removing the elements with the maximal data value.
For task words, trimmings are obtained by removing the elements with 
the largest data value and updating the tasks of the remaining elements correctly. That is, a \emph{trimming of a $\cD$-task word $\cT$} 
 is a
 $\trim{\cD}{1}$-task word $\cT_1$ such that $\omega(\ts)=\omega(\ts_1)$ 
for every $d$ and every $\ts,\ts_1 \in \compatTasks$ with
$\cT \models \ts(d)$ and $\cT_1 \models \ts_1(d)$.
We say that $\cT_{1},\cT$ are \emph{consecutive} if $\cT_{1}$ is a trimming of $\cT$.
\end{definition}

The trimming of a task word is unique, and we denote it 
$\ttrim{\cT}{1}$. 

\begin{example}
$\trim{\ul{\cD}}{1}$ is obtained from  $\ul{\cD}$
by removing 
$\ul{d}$ and $\ul{f}$. 
The $\trim{\ul{\cD}}{1}$-task word $\ttrim{\ul{\cT}}{1}$ 
has universe 
$\{\ul{a},\ul{b},\ul{c},\ul{e}\}$ and order relations 
as in $\trim{\ul{\cD}}{1}$.
Note that  $\ul{d}$ and $\ul{f}$ 
contributed in $\ul{\cD}$ to the satisfaction of $\ul{\varphi}_\exists$, 
so  $\ttrim{\ul{\cT}}{1}$ has promised tasks and is no longer completed,  with
 interpretations $\ul{ts}_1^P =\{ \ul{c} \}$, 
 $\ul{ts}_3^P =\{ \ul{a},\ul{b},\ul{e} \}$,
and the  remaining letters empty.
Note that the tasks for the shared elements 
of $\ul{\cT}$ and $\ttrim{\ul{\cT}}{1}$ realize the same witness type sets.
\end{example}

\noindent
We have achieved {\bf(S1)}: reducing satisfaction of $\varphi_\exists$ to finding a sequence of task words. 

\label{subsec_towards_automaton}

\begin{proposition}
There is a data word $\cD \models \varphi_\exists$ if and only if there is a sequence $\cT_1
 \ldots, \cT_n$ of consecutive task words, where $\cT_n$ is a completed $\cD$-task word.
\end{proposition}

Now to {\bf (S2)}: as the limited memory of PIAs hinders the manipulation of task words with unbounded  length and data values, we operate on their \emph{extremal strings} instead. 

First, in \emph{data abstractions} of task words, 
we do not distinguish all data values, but only 
the \emph{top layer}  elements with maximal value,
the \emph{second to top layer},  
and the rest.  
We let $\Layers = \{\otop,\ttop, \rest\}$, and define
the alphabet 
$\Gamma=  \Layers \times \compatTasks$. 
We also define its restrictions
to completed and promised tasks as
$\CGamma =  \Layers \times \compatCompletedTasks$ 
and 
$\PGamma =  \Layers \times \compatPromisedTasks$, 
while $\Gamma_h = \{\hlayer\} \times \compatTasks$ is the restriction of $\Gamma$ to some specific $h\in\Layers$.
For a symbol $\gamma=(\hlayer,\ts)$ in $\Gamma$, we denote $\ts(\gamma)=\ts$ and $\omega(\gamma)=\omega(\ts)$.

\begin{definition}[Data abstraction]
\label{def_data_abstraction}
Let $\cT$ be a $\cD$-task word.
For every $d \in D$, let $\ts_d$ be such that $\cT \models \ts_d(d)$,
and let $\cA$ be the data word over $\Gamma$ with  same universe and order relations as $\cT$, and with $\cA \models \gamma_h (d)$ 
where $\gamma_h = (\hlayer, \ts_d)$ iff
(a) $h = \onetoprm$ and  $\val{\cD}{d} = \maxdv_\cD$, 
(b) $h = \twotoprm$ and $\val{\cD}{d} = \maxdv_\cD-1$, or 
(c) $h = \restrm$ and $\val{\cD}{d} \in [\maxdv_\cD-2]$.
The \emph{data abstraction $\dataAbs(\cT)$} of $\cT$ 
is the string projection
$\project(\cA)$.
\end{definition}

Extremal strings are obtained from  data abstractions by 
keeping only the maximal and minimal positions in each layer with respect to the tasks. 
We extend to them the notions of \emph{consecutive} and \emph{completed}. 

\begin{definition}[extremal strings]
\label{def_function_ext}
\label{def_extremal_string}\label{def_consecutive_extremal}
Let $w \in \Gamma^\star$.
We define its extremal positions $\extPos(w)$:

 \renewcommand{\arraystretch}{1.25}
{\footnotesize
    \begin{tabular}{rl} 
$\positions_{h, \theta}(w)=$  & $\{ \ell \in [\length{w}] \mid w(\ell)  =
(\hlayer, \ts)$, $\theta\in \omega(\ts)\}$ \\
$\positions_{\restrm, P_\theta}(w)=$ &  $\{ \ell \in [\length{w}] \mid w(\ell)
= (\rest, \ts)$, $ P_\theta \in \ts\}$ \\
$\extremal_{h,\theta}(w)=$ &  $\{\ell \mid \ell = \max(\positions_{h,
  \theta}(w))$ \mbox{ or } { $\ell = \min(\positions_{h, \theta}(w))\}$} \\
{If $\theta \models x \leq_1 y$}, \quad 
$\extremal_{\theta}(w)=$ & $\left\{\ell\mid \ell=\max(\positions_{\restrm, P_\theta}(w))\right\}$ \\ 
{If $\theta \models y <_1 x$}, \quad
$\extremal_{\theta}(w)=$  & $\left\{\ell\mid \ell = \min(\positions_{\restrm, P_\theta}(w))\right\}$\\ 
$\extPos(w) =$ & $\displaystyle{\bigcup_{\theta \in \Theta_{\exists}} \big (
  \extremal_{\theta}(w) \!\!\!\! \bigcup_{h \in \Layers}
  \!\!\!\!\!\extremal_{h,\theta}(w) \big)  }$  \vspace{-3\baselineskip}
     \end{tabular}}
 
\smallskip\noindent
If $\extPos(w) = \{\ell_1, \ldots, \ell_r\}$ and  $\ell_1 < \cdots < \ell_r$, then the \emph{extremal string} of $w$ is $\ext(w) = w(\ell_1) \cdots w(\ell_r)$. 
 $\EXT(\Gamma) = \{\ext(w) \mid w \in \Gamma^\star\}$ denotes the set of extremal strings.
Note that 
$s=\ext(w)$ implies $\ext(s) = s$ and 
$\ext(\varepsilon)=\varepsilon$.
An extremal string $s$ is \emph{completed} if $s \in \CGamma^+$,
or if $s= \varepsilon$ and 
$\emptystructure{\dat{\Sigma}} \models_{\dat{\Sigma}} \psi$.
We occasionally  write $\ext(\cT)$ to mean 
$\ext(\dataAbs(\cT))$.

A pair $s',s$ of extremal strings is \emph{consecutive} if 
$s' = \ext(\ttrim{\cT}{1})$ and $s = \ext(\cT)$ for some task word $\cT$. 
\end{definition}

For an extremal string $s$ and  ${\ell} \in [\length{s}]$, 
the set of letters that can augment $s$ at position $\ell$ without being extremal is 
 $\Gamma^{\nonextsf}(s,\ell) =  
\{\gamma \in \Gamma \mid  \ext(s) = \ext(s(1)\cdots s({\ell}-1) \gamma s({\ell})\cdots s(\length{s})) \},$ and we define 
$ \Gamma_{h}^{\nonextsf}(s,\ell) = \Gamma_{h} \cap \Gamma^{\nonextsf}(s,\ell)$  for $h \in \Layers$.

\begin{example}
\label{ex_emb_ext_}
Let $\ul{w}$ be the following 6-letter string over $\ul{\Gamma} = \Layers \times 2^{\Tasks}_{\ul{\Omega}}$: \\
{\small\centering
$(\rest, \ul{ts}_3^C)(\rest, \ul{ts}_3^C)(\ttop, \ul{ts}_1^C)(\otop, \ul{ts}_2^C)(\rest, \ul{ts}_3^C)(\otop, \ul{ts}_4^C)$
}\\
Then $\extPos(\ul{w}) = \{1,3,4,5,6\}$, since the letter at position $2$ appears both to the left, at position $1$, and to the right, at position $5$, and  $\ul{s} = \ext(\ul{w})$  
is the substring obtained from $\ul{w}$ by removing the non-extremal position $2$. 
\end{example}

\noindent
The concludes {\bf (S2)}, reducing the existence of $\cD$  to a sequence of extremal strings. 

\begin{corollary}
There is a data word $\cD \models \varphi_\exists$ if and only if there is a sequence of consecutive extremal strings where the last one is completed.
\end{corollary}

\smallskip\noindent
{\bf Perfect extremal strings for \texorpdfstring{$\varphi_\forall$}{varphi-forall}}
We define in the appendix a formula  
which intuitively `extracts' the 2-type of elements in a data word. 
Let $\alpha = (\hlayerp{\alpha},\ts_\alpha)$ and $\beta = (\hlayerp{\beta},\ts_\beta)$
in $\Gamma$ with 
at least one of them 
in $\Gammatop$.
The formula $\perfect_{\alpha,\beta}(x,y)$  
implies for every atomic formula either itself or its negation. 
For example, if $h_\alpha = h_\beta = \onetoprm$, then 
$\perfectalpha_{\alpha,\beta,\lesssim_2}(x,y)$ implies $x \precsim_2 y$, $y \precsim_2 x$, $\neg S_2 (x,y)$, and $\neg S_2 (y,x)$.
Hence for all $\alpha, \beta \in \Gamma$ with at least one of them in $\Gammatop$,  
there exists a 2-type $\theta(x,y)$ such that 
$\perfect_{\alpha,\beta}(x,y) \equiv_{\dat{\Xi}} \theta(x,y)$. This allows us to describe the 2-type of elements in task words via $\perfect_{\alpha,\beta}$ formulas.
For any two elements of the data word, there is a (possibly iterated) trimming in which both appear and one of them has the maximal data value, and their perfect formula, 
which is equivalent to their 2-type, determines whether they satisfy the universal constraint $\chi$. 
Thus we can ensure satisfaction of $\chi$ using $\perfect_{\alpha,\beta}(x,y)$ formulas from all the trimmings. 

\begin{definition}[Perfect string, perfect task word] 
\label{def_perfect_string}
\label{def_perfect_tasked_word}
Let $w \in \Gamma^\star$. We say $w$ is a \emph{perfect string} if for every two positions $\ell_1 \lneq \ell_2$ in $w$ 
such that 
$\{w(\ell_1),w(\ell_2)\} \cap \Gammatop \not= \emptyset$ we have
$ \perfect_{w(\ell_1),w(\ell_2)}(x,y) \models_{\dat{\Xi}} \chi(x,y) \wedge \chi(y,x)$. 
Note that the empty string $\varepsilon$ is perfect. 
A task word $\cT$ is \emph{perfect} if it is empty, or if $\ext(\cT)$ and $\ext(\ttrim{\cT}{1})$ are perfect.
\end{definition}

\begin{example}
Let $\ul{\alpha} = (\ttop, \ul{ts}_1^C)$ and $\ul{\beta} = (\otop, \ul{ts}_2^C)$.
Then $\perfect_{\ul{\alpha},\ul{\beta}}(x,y)$ is given by:
{\small
{$\perfect_{\ul{\alpha},\ul{\beta}}(x,y) = \ul{\xi}_1(x) \wedge \ul{\xi}_2(y) \wedge (x <_1 y) \wedge (x \lnsim_2 y) \wedge S_2(x,y). $ }\\
}
The 2-type $\theta$ to which $\perfect_{\ul{\alpha},\ul{\beta}}(x,y)$ is equivalent over $\dat{\Xi}$  
is given by the conjunction of $\perfect_{\ul{\alpha},\ul{\beta}}(x,y)$ with
$\neg \ul{\xi}_2(x) \wedge \neg\ul{\xi}_1(y) \wedge (y \not<_1 x) \wedge (y \not\lesssim_2 x) \wedge \neg S_2(y,x)$.
We have that $\ul{w}$ is a perfect string, and 
$\perfect_{\ul{\alpha},\ul{\beta}}(x,y) \models_{\dat{\Xi}} \ul{\chi}(x,y) \wedge \ul{\chi}(y,x)$. 
\end{example}

We   
characterize satisfiability in terms of perfect completed task words.
\begin{lemma}
\label{lem_uni_task}
Let $\cT$ be a $\cD$-task word. 
$\cT$ is perfect if and only if $\cD \models \varphi_\forall$.
\end{lemma}

As a corollary of Lemma~\ref{lem_uni_task} and Lemma~\ref{lem_good_decoration}, we get:
\begin{proposition}
\label{prop_cor_perfect_completed_tasked_word}
For every data word $\cD \in \dat{\Xi}$, 
$\cD \models \varphi$ if and only if there exists a perfect completed $\cD$-task word.
There is $\cD \models \varphi$ if and only if there is a sequence of consecutive perfect extremal strings where the last one is completed.
\end{proposition}

We are almost ready to define $\pA^\vphi$. Intuitively, it will 
guess a sequence of extremal strings as in
Prop.~\ref{prop_cor_perfect_completed_tasked_word}, placing pebbles from an
  extremal string to a consecutive one.  
This requires the automaton to verify consecutiveness, and to know which
positions in consecutive extremal strings 
correspond to the same position in the input. 
This is easy if we
have the underlying task word;
indeed, given a task word $\cT$ and an extremal string
 $s' = \ext(\cT)$, 
there is a bijective mapping from the \emph{extremal elements} of $\cT$ that
$s'$ stores, to their positions in $s'$. The same holds for $\ttrim{\cT}{1}$
and $s = \ext(\ttrim{\cT}{1})$. 
By composing these mappings after inverting the latter, and restricting its
domain  to positions that remain extremal after updating the abstracted data values
(that is, shifting the top layer to second top, and the second top into
the remaining layer), we obtain a \emph{partial embedding from $s$ to $s'$ via
  $\cT$} that keeps track of the matching positions; the precise definition is in the appendix.   
But one major hurdle remains: these notions are defined in terms of a task
word $\cT$, and our PIA cannot store task words, only their extremal
strings.  
We overcome this through a merely syntactic characterization of
consecutiveness, which can be verified without a concrete task word.  
This rather technical step 
relies on the fact that if $s,s'$ are consecutive, then $s'$ can
be obtained by guessing a substring $r$ that will get new data values,
interleaving it into the proper positions $g$ of $s$, which can also be guessed,
and updating the abstracted data values.   
Also the partial embedding that keeps track of matching the
positions  
can be obtained without a concrete $\cT$, using $r$ and $g$.  

\begin{lemma}
\label{lem_check_consec_expspace}
We can decide whether two given extremal strings $s,s'$  are consecutive 
in $\EXPSPACE$.  
If they are, then we can also obtain in $\EXPSPACE$ a partial embedding
$\notOneTopEmbeddingNoArg{s}{s'}$  from
positions in $s$ to positions in $s'$ that coincides with the 
partial embedding from $s$ to $s'$ via
  $\cT$ for every task word $\cT$ such that 
 $s' = \ext(\cT)$ and $s = \ext(\ttrim{\cT}{1})$. 
\end{lemma}

\smallskip\noindent
{\bf 
The automaton.}
We give a high-level description of  $\pA^{\varphi} = (\Xi, m+1, Q, q_{\mathrm{init}}, F,
\delta)$, and refer to App.~\ref{sec_aut_def} for a full
definition. We have $m = 7 \cdot |\Theta_\exists|$: there is one pebble for each existential
constraint in  $\Theta_\exists$ and each layer in $\Gamma$, plus an additional pebble per
constraint, and one designated pebble $m+1$ to read non-extremal
positions.   $Q = Q_e \cup Q_p$  has two types of states: 
\begin{compactitem} 
\item $Q_e$ contains states  $(s,\tau)$ with  
$s$ a  perfect extremal string and 
$\newtau$ an $(m+1,\length{s})$-pebble assignment, which  intuitively describes
the assignment  after reading $s$.  
\item 
$Q_p$  contains states of the forms $(s,\tilde{s},\tau,0)$ and $(s,\tilde{s},\tau,1)$ for every
perfect extremal string $s$, 
non-empty prefix $\tilde{s}$ of $s$, and
$(m+1,\length{s})$-pebble assignment
$\newtau$  
that satisfies certain conditions that  hold when only the prefix $\tilde{s}$ has been read.
\end{compactitem}
The initial state is $q_{\mathrm{init}} \! = \! (\varepsilon, \rho_\bot)\in Q_e$
and the final states are 
$F \! = \! \{(s,\tau) \in Q_e \mid s\mbox{ is completed}\}$. 
The transition $\delta$ is roughly as follows. 
$\pA^{\varphi}$  
should transition from $(s,\tau) \in Q_e$ 
 to $(s', \tau') \in Q_e$ for consecutive $s,s'$, but 
since it can only move one pebble at a time, we have intermediate states in $Q_p$ 
which allow it to read $s'$ from left to right  
by iterating over all its prefixes.  
We start reading $s'$ by moving to $(s',s'(1),\tau'_{{q'}},0) \in Q_p$, 
where $\tau'_{{q'}}$ stores the pebble assignment induced by 
$\notOneTopEmbeddingNoArg{s'}{s}$. 
Once the whole extremal string $s'$ has been read, we move to the next
  extremal state. 

This finishes the construction of the automaton $\pA^\varphi$ with $L(\pA^\varphi) = L(\varphi)$, and thus the proof of Theorem~\ref{th:containment}.
\label{sec_expspace}
Concerning the upper bound on finite satisfiability,  
by Theorem~\ref{lem_th_our_snf} and $\EXT(\Gamma) \subseteq \Gamma^{7 |\Theta_\exists|}$, we get
that $\pA^\varphi$ has size at most double exponential in $|\psi|$. 
For the $\EXPSPACE$ upper bound, we need to show that the transition relation of $\pA^\varphi$ is $\EXPSPACE$-computable (Lemma~\ref{lem_check_consec_expspace} in the appendix).  
This with Theorem~\ref{th_pia_emptiness} 
gives an alternative proof of 
the upper bound in \cite{ar:zeume12}:
\begin{corollary}
\label{cor:expspace}
Finite satisfiability of $\mathrm{FO}^2(\vocDataWords)$ is in $\EXPSPACE$.
\end{corollary}

\noindent{\bf Relation to the proof  of Schwentick and Zeume
  \cite{ar:zeume12}}  Naturally, there are  similarities between the techniques;
our extremal strings and tasks are similar to their profiles and directional constraints. However, a key difference is that in  their
`geometric' view, elements of the data word are assigned points $(a,b)$ in the plane with $a$ a position
in $\leq_1$, and $b$ a data value. 
Existential constraints are indicated by marking the witnesses
with the letters they should have, and many profiles
in a  consistent sequence 
can contain points with the
same $a$ value. 
In contrast, our `temporal' view arises from the
computation of the PIA. We mark elements with existential
constraints they need to satisfy and that they have satisfied, which  is compatible with the read-once nature of PIA. 
 It does not seem possible to  use their proof
techniques without modifying PIA to allow multiple
readings of the input. The modified  model would work  for 
the logic-to-automata relation established here, but
we suspect it would be too strong for the other
direction.

\section{Discussion and Conclusion}
We introduced pebble-intervals automata (PIA) and
studied their computational power. 
We proved that the projections of data languages definable in
 $\mathrm{FO}^2(\leq_1,\lesssim_2,S_2)$
are
PI languages, and as a by-product, 
obtained an alternative proof that finite satisfiability  is in $\EXPSPACE$.
The main question that remains is the converse of our main result:
whether every PI language is the projection of an  $\mathrm{FO}^2(\leq_1,\lesssim_2,S_2)$
definable data language. 
We believe this is the case. 
Our work also gives rise to 
other questions. 
We suspect that our results  
can be extended to $\omega$-languages, and we would like to adapt them
to $C^2$, which extends \FOtwo\ with counting quantifiers 
\cite{DBLP:conf/csl/Pratt-Hartmann14,DBLP:journals/tocl/Tan14}. 
We also plan  to explore further the computational power of our
automata model, for instance,  to establish a
pumping lemma that allows us to prove that some context-free 
languages are not PI languages. 

 \bibliographystyle{plain}
 \bibliography{ref}

\newpage

\appendix

\section{Embeddings}
\label{app_embeddings}
Throughout the paper, we introduce definitions in which 
a data word or string $X$ is obtained from a data word or string $Y$ by applying an operation $Op$.
The operation $Op$ induces an injective order-preserving function $F_{Op}$ relating each position or universe element in $X$ 
to a position or universe element in $Y$. 
$F_{Op}$ is called an \emph{embedding of $X$ into $Y$}. 
By \emph{order-preserving} we mean that whenever a pair $a \leq_X b$ of positions or universe elements 
is mapped by $F_{Op}$ to a pair $a',b'$ of positions or universe elements, 
we have $a' \leq_Y b'$, where $\leq_Z$ for $Z\in \{X,Y\}$
is as follows:
\begin{enumerate*}[label=(\roman*)]
 \item if $Z$ is a string, $\leq_Z$ is the order on natural numbers;
 \item if $Z$ is a data word, $\leq_Z$ is the order $\leq_1$ of $Z$. 
\end{enumerate*}
For instance, given a data word $\cD$ with universe $D$ of size $n$, 
the embedding of the string projection $w_\cD = \project(\cD)$ of $\cD$ into $\cD$
assigns to every position $\ell \in [n]$ the element $d$
which satisfies $\ell =  |\{d' \in D \mid \cD \models d' \leq_1 d\}|$.
We denote the embedding of $X$ into $Y$ by $\embeddingOp{Op}{Y}$. In the case of the string projection,
$\embeddingOp{\project}{\cD}$ is an injective function (indeed a bijection)
from $[n]$ to $D$ which preserves the order $\leq_1$ of $\cD$ in terms of the order of positions in the projection. 
The inverse $\embeddingOp{Op}{Y}^{-1}$ of an embedding is a partial function. 
Given two operations $Op_1$ and $Op_2$, 
we denote the embedding resulting from their composition $X = Op_1(Op_2(Z))$,
$\embeddingOp{Op_2}{Z} \circ \embeddingOp{Op_1}{Op_2(Z)}$,
by  $\embeddingOp{Op_1 \circ Op_2}{Z}$. 

\begin{runningexample}
$\embeddingOp{\project}{\ul{\cD}}: [6] \rightarrow \ul{D}$ is given by:
\[
\begin{array}{llllllll}
\embeddingOp{\project}{\ul{\cD}}(1) = \ul{a} &&&	\embeddingOp{\project}{\ul{\cD}}(4)  = \ul{d} \\
\embeddingOp{\project}{\ul{\cD}}(2) = \ul{b} &&& \embeddingOp{\project}{\ul{\cD}}(5)  = \ul{e} \\
\embeddingOp{\project}{\ul{\cD}}(3) = \ul{c} &&& \embeddingOp{\project}{\ul{\cD}}(6)  = \ul{f}
\end{array}
\]
\end{runningexample}

\begin{runningexample}
\label{ex_emb_abst}
Let $\ul{w} = \dataAbs(\ul{\cT})$ be the data abstraction of the task word $\ul{\cT}$. Then $\ul{w}$ is the 6-letter string over $\ul{\Gamma} = \Layers \times 2^{\Tasks}_{\ul{\Omega}}$ given by:
{\small
\[
(\rest, \ul{ts}_3^C)(\rest, \ul{ts}_3^C)(\ttop, \ul{ts}_1^C)(\otop, \ul{ts}_2^C)(\rest, \ul{ts}_3^C)(\otop, \ul{ts}_4^C).
\]
}
Let
$\ul{w}' = \dataAbs(\ttrim{\ul{\cT}}{1})$. Then $\ul{w}'$ is given by the string
\[
(\rest, \ul{ts}_3^P)(\rest, \ul{ts}_3^P)(\otop, \ul{ts}_1^P) (\ttop, \ul{ts}_3^P)
\]
Note that the length discrepancy between $\ul{w}$ and $\ul{w}'$ matches the number of elements removed in the trimming. 
The embedding $\embeddingOp{\dataAbs}{\ul{\cT}}$ is given by:
\[
\begin{array}{llllllllllllllllll}
\embeddingOp{\dataAbs}{\ul{\cT}}(1) = \ul{a} &&&	\embeddingOp{\dataAbs}{\ul{\cT}}(4) = \ul{d}  \\
\embeddingOp{\dataAbs}{\ul{\cT}}(2) = \ul{b} &&&	\embeddingOp{\dataAbs}{\ul{\cT}}(5) = \ul{e}  \\
\embeddingOp{\dataAbs}{\ul{\cT}}(3) = \ul{c} &&&	\embeddingOp{\dataAbs}{\ul{\cT}}(6) = \ul{f} 
\end{array}
\]
The embedding $\embeddingOp{\dataAbs}{\ttrim{\ul{\cT}}{1}}$ is given by:
\[
\begin{array}{llllllllllllllllll}
\embeddingOp{\dataAbs}{\ttrim{\ul{\cT}}{1}}(1) = \ul{a} \\
\embeddingOp{\dataAbs}{\ttrim{\ul{\cT}}{1}}(2) = \ul{b} \\
\embeddingOp{\dataAbs}{\ttrim{\ul{\cT}}{1}}(3) = \ul{c}
\end{array}
\]\end{runningexample}

\begin{runningexample}
\label{ex_emb_ext}
The embedding $\embeddingOp{\ext}{\ul{w}}$ is given by: 
{
\[
\begin{array}{lllllllllllllllll}
\embeddingOp{\ext}{\ul{w}}(1) = 1 &	\embeddingOp{\ext}{\ul{w}}(2) = 3 & \embeddingOp{\ext}{\ul{w}}(3) = 4  \\ 
\embeddingOp{\ext}{\ul{w}}(4) = 5  & \embeddingOp{\ext}{\ul{w}}(5) = 6
\end{array}
\]
}
The embedding $\embeddingOp{\ext}{\ul{w}'}$ is the identity function $\embeddingOp{\ext}{\ul{w}'}(i) = i$.
\end{runningexample}

\section{Proof of closure properties}

We briefly  recall 
the definitions  of \emph{shuffle} and \emph{iterated
  shuffle}. 

\begin{definition}[(Iterated) shuffle] 
\label{shuffle-def}
Let $u,v \in \Sigma^\star$. 
A string  $w \in \Sigma^\star$ (of length $\length{u} + \length{v}$) 
is called a \emph{shuffle} of 
$u,v \in \Sigma^\star$ 
 if there is a set $S = \{i_1, \ldots i_{\length{u}}\} \subseteq
[\length{w}]$ such that $u$ is the string $w(i_1)\cdots w(i_{\length{u}})$,
and $v$ is the string of the remaining letters of $w$.
We denote $u \shuffle v$ the set of all strings that are a shuffle of $u$ 
  and $v$, and define
the shuffle of two languages $L,L' \subseteq \Sigma^\star$ 
as  $L \shuffle L' = \bigcup_{u \in L, v \in L'} u \shuffle v$.
The \emph{iterated shuffle} of a language $L$ is defined as $L^{\shuffle} = \bigcup_{i \geq 0} L_i$ where $L_0 = L$ and $L_{i+1} = L_i \shuffle L$.
\end{definition}

\begin{theorem}
\label{th_closure}
The class of PI  languages is  effectively 
closed under union, concatenation, Kleene-$\star$, shuffle, and iterated shuffle.
\end{theorem}

\proofclosure

\subsection{Proofs of non-closure}

\begin{theorem}
\label{th_intersection}
The class of PI  languages is not effectively closed under
 \begin{enumerate}
 	\item 
	intersection, even with regular languages, nor under 
	\item
	complement.
 \end{enumerate} 
\end{theorem}

\proofnonclosure

\subsection{Proof of Corollary~\ref{cor_universality_inclusion}}
\proofcorunivinc

\section{Proof of upper bound for emptiness} 
\label{app_emptiness}

We now prove the upper bounds of Theorem~\ref{th_pia_emptiness}.
For the rest of this section, assume a  PIA $\pA= (\Sigma, m, Q,
q_{\mathrm{init}}, F, \delta)$. 

\begin{definition}[Feasible sequence of transitions]
\label{def_feasible_seq}
An \emph{arrangement} of  pebbles $[m]$ is  a linear order $\Ord = \lng O, R_{\leq}\rng$ with $O \subseteq [m]$. 
Each $(m,n)$-pebble assignment $\rho$ induces an arrangement
$o(\rho)$ with  $O = [m] - \rho^{-1}(\bot)$ and  $R_{\leq}(i,j)$ if and only
if $\rho(i) \leq \rho(j)$  for every $i,j \in [m]$. 

Let $r\geq 0$ and 
let 
$\bar{t} = (t_1 ,\ldots, t_r)$  be a sequence of transitions. 
We say $\bar{t}$ is \emph{feasible} if 
there is a sequence $\bar{\Ord} = (\Ord_0 ,\ldots, \Ord_r)$ 
of arrangements 
and a sequence $\bar{q}=(q_0,\ldots,q_r)$ of states with $q_{\mathrm{init}}=q_0$, 
such that $\Ord_0$ is empty, and 
for every $1 \leq \ell \leq r$:
\begin{enumerate}
	\item 
	if $t_\ell$ is a silent transition, then $t_\ell = (q_{\ell-1},q_\ell)$ and $\Ord_\ell = \Ord_{\ell-1}$, and
	\item
	if
$t_\ell$ is a $\MoveNoArgs$ transition, then there are $k \in [m]$, $i, j \in ([m]-\{k\})\cup \{\rhd,\lhd\}$, and $\sigma\in\Sigma$ 
with  $t_\ell = (q_{\ell-1},\Move{k}{i}{j},\sigma,q_\ell)$ and the
arrangement   $\Ord_\ell = \lng O_\ell, R_{\leq,\ell} \rng$ is such that
 	$O_\ell = O_{\ell-1} \cup \{k\}$, 
 $R_{\leq,\ell}(i',j')$  if and only if $R_{\leq,\ell-1}(i',j')$ for all for
 $i',j' \in [m] - \{k\}$, and additionally, 
	if $i \in [m]$ then  $R_{\leq,\ell}(i,k)$, and 
	if $j \in [m]$ then $R_{\leq,\ell}(k,j)$. 
\end{enumerate}
\end{definition}
The following lemma shows that the feasible sequences of transitions are exactly the ones corresponding to actual computations of the automaton.
\begin{lemma}
\label{lem:feasible}
\label{lem_feasible}
A sequence $\bar{t}$ of transitions 
is feasible if and only if	
there is a string $u \in \Sigma^\star$ and 
a sequence of configurations $\bar{\pi}$ such that $(\bar{t},\bar{\pi})$ is a computation of $\pA$ on $u$. 
Moreover, 
$\bar{t}$ ends in an accepting state
if and only if
$(\bar{t},\bar{\pi})$ is an accepting computation. 
\end{lemma}

\begin{lemma}
\label{lem_short_string} 
There is a feasible sequence of transitions $\bar{t}$ 
of length at most $|\pA|\cdot 2^{O(m \log m)}$
ending in an accepting state 
iff $L(\pA) \neq \emptyset$. 
\end{lemma}
\begin{proof}
\label{proof_lem_short_string}
By Lemma~\ref{lem:feasible}, if $\bar{t}$ is feasible and ends in an accepting state, then there is a sequence of configurations $\bar{\pi}$
such that $(\bar{t},\bar{\pi})$  is an accepting computation of $\pA$ on some string $u$. Hence $u \in L(\pA)$. 

Conversely, assume that $L(\pA) \neq \emptyset$, and let
$(\bar{t},\bar{\pi})$ be an accepting computation of minimal length on any string $u$. 
Denote $(\bar{t},\bar{\pi}) = ((t_1,\ldots,t_r),\pi_0,\ldots,\pi_r))$. 
By Lemma~\ref{lem:feasible}, $\bar{t}$ is feasible and ends in an accepting state. Let $\bar{\Ord}$ be as guaranteed
for $\bar{t}$ in Definition~\ref{def_feasible_seq}. 
Now assume for contradiction that there are two distinct $h_1,h_2 \in [r]$ such that 
$t_{h_1} = t_{h_2}$ and $\Ord_{h_1}=\Ord_{h_2}$. 
Then $(t_1,\ldots,t_{h_1},t_{h_2+1},\ldots,t_r)$ is a feasible sequence of transitions ending in an accepting state.
Hence, there is a word $u_{12}$ which is accepted by a computation of length $r-(h_2-h_1)$, in contradiction to the minimality of 
$(\bar{t},\bar{\pi})$. Hence, the length of $\bar{t}$ is at most $|\delta|\cdot M$, where $M$ is the number of arrangements $\Ord$ of $m$ pebbles.
We have that $M \leq 2^m \cdot m! $ since there are $2^m$ ways of choosing a subset $O \subseteq [m]$ as the universe of $\Ord$, 
and $|O|! \leq m!$ ways to linearly order the set $O$.
\end{proof} 

\subsection{Proof of Theorem~\ref{th_pia_emptiness}}

\begin{proof}
\label{proof_th_pia_emptiness}
Let $\pA= (\Sigma, m, Q, q_{\mathrm{init}}, F, \delta)$.
We non-deterministically attempt to guess a feasible sequence of transitions $\bar{t}$ which ends at an accepting state.
Due to the number of non-repeating sequences of arrangements,
$L(\pA) \neq \emptyset$ if and only if there is a feasible sequence of transitions 
ending in an accepting state 
and sequences $\bar{\Ord}$ as in Definition~\ref{def_feasible_seq} 
whose lengths are at most $|\pA|\cdot 2^{O(m \log m)}$.
We only need to keep simultaneously 
a counter of the sequence length $r$, two transitions $t_{h}, t_{h+1}: h\in [r-1]$ and two arrangements $\Ord_h, \Ord_{h+1}$.
The size of the representation of a transition is logarithmic 
in $|\pA|$.  
The size of the representation of an arrangement is $O(m^2)$.
Completeness for the case where the automaton has $O(\log |\pA|)$ pebbles follows from the $\NL$-completeness of the emptiness problem of standard finite state automata
and Prop.~\ref{prop:regular-languages}. 
\end{proof}

\section{Proof of the normal form}

\begin{proof}
Here we prove Theorem~\ref{th_our_snf}.
We denote the class of all finite structures over a vocabulary $\voc$ by $\str{\voc}$.
To define $h$, we first need to introduce two functions (called translations), $\trans_1$ and $\trans_2$.
For simplicity we assume that the empty word satisfies $\psi$, and therefore
$\varphi_\varepsilon = \mathit{True}$. 
For the other case, we need to change the following by conjoining each of
 $\varphi^0$, 
 $\varphi^1_\exists$, and
 $\varphi^2_\exists$ with $\varphi_\varepsilon = \exists x \, (\mathit{True})$.

\paragraph*{The translation $\trans_1$}
Let 
\[
\varphi^0 = \forall x \forall y \, \chi^0(x,y) \wedge \bigwedge_{b=1}^{B}{\forall x \exists y \, \chi_b^0(x,y)}
\]
be the Scott Normal Form of $\psi$ (see e.g. \cite[Theorem 2.1]{gradel1999logics}). 
The formula $\varphi^0$ is over the vocabulary 
$\dvoc{\Sigma} \cup \vocSNF$  where $\vocSNF$ is a set of fresh  unary relations.
The formulas 
$\chi^0$, and $\chi_b^0$ are quantifier-free. 
The length of $\varphi^0$ is linear in that of $\psi$. 

For every model $\cD \in \dat{\Sigma}$ of $\psi$, 
there is a unique expansion of $\cD$ which satisfies $\varphi^0$. 
Let 
$\trans_1: \str{\dvoc{\Sigma} \cup \vocSNF} \rightarrow \dat{\Sigma}$ 
be the function
which takes a $\str{\dvoc{\Sigma} \cup \vocSNF}$-structure to its reduct to $\dvoc{\Sigma}$. 
The following hold:
	\begin{enumerate}[label=(\roman*${}_1^0$),ref=(\roman*${}_1^0$)]
		\item \label{proof_transa_i_dp}
		For every $\cE \in \str{\dvoc{\Sigma} \cup \vocSNF}$,
		if $\cE \models {\varphi^0}$ then $\trans_1(\cE) \models_{\dat{\Sigma}} \psi$, and
		\item \label{proof_transa_ii_dp}
		For every $\cD \in \dat{\Sigma}$,
		if $\cD \models_{\dat{\Sigma}} \psi$ then there exists $\cE \in \str{\dvoc{\Sigma} \cup \vocSNF}$ 
		such that $\trans_1(\cE) = \cD$ and $\cE \models {\varphi^0}$.
	\end{enumerate}

\paragraph*{Making the types explicit}
Next we define a formula $\varphi^1$ which is equivalent to $\varphi^0$.
Let $A $ be a finite set such that
$\{\nu_a \mid a\in A\}$ is the set of 1-types over $\dvoc{\Sigma} \cup \vocSNF$. 
Every quantifier-free formula is equivalent to a disjunction of 2-types. 
Hence, there a set $C$ whose size is at most the number of 2-types over $\dvoc{\Sigma} \cup \vocSNF$
such that
every conjunct ${\forall x \exists y \, \chi_b^0(x,y)}$ is equivalent to
\[
 \forall x \, \bigwedge_{a=1}^{A}{\nu_a(x) \rightarrow {\exists y \, \bigvee_{c=1}^{C}{\beta_{abc}(x,y)}}}
\]
where $\beta_{abc}(x,y)$ is 2-type  over $\dvoc{\Sigma} \cup \vocSNF$
for every $a$, $b$, and $c$. 
There is a set $\Theta^\beta_\forall(x,y)$ 
of 2-types  over $\dvoc{\Sigma} \cup \vocSNF$ such that 
\[
\begin{array}{llllll}
\displaystyle \chi^0(x,y) &\equiv&
\displaystyle \bigvee_{\beta \in \Theta^\beta_\forall} \beta(x,y),
\\
\displaystyle \varphi^0 &\equiv&
\displaystyle  \forall x \forall y \bigvee_{\beta \in \Theta^\beta_\forall} \beta(x,y)
\wedge \bigwedge_{b=1}^{B}\forall x \, \bigwedge_{a=1}^{A}\nu_a(x) \rightarrow  \\ & & {\exists y \, \bigvee_{c=1}^{C}{\beta_{abc}(x,y)}}.
\end{array}
\]
Let
$
\varphi^1 = \varphi^1_\forall \wedge \varphi^1_\exists
$, 
where
\[
\begin{array}{lll}
\varphi^1_\forall &=& \forall x \forall y \bigvee_{\beta \in \Theta^\beta_\forall} \beta(x,y), \\
\varphi^1_\exists &=& \forall x \, \bigwedge_{a=1}^{A}{\nu_a(x) \rightarrow \bigwedge_{b=1}^{B}{\exists y \, \bigvee_{c=1}^{C}{\beta_{abc}(x,y)}}}.	
\end{array}
\]

We have $\varphi^1 \equiv \varphi^0$ and from~\ref{proof_transa_i_dp} and~\ref{proof_transa_ii_dp}:
	\begin{enumerate}[label=(\roman*${}_1^1$),ref=(\roman*${}_1^1$)]
		\item \label{proof_transa_i}
		For every $\cE \in \str{\dvoc{\Sigma} \cup \vocSNF}$,
		if ${\cE} \models \varphi^1$ then $\trans_1({\cE}) \models \psi$, and
		\item \label{proof_transa_ii}
		For every $\cD \in \dat{\Sigma}$,
		if $\cD \models_{\dat{\Sigma}} \psi$ then there exists ${\cE} \in \str{\dvoc{\Sigma} \cup \vocSNF}$ 
		such that $\trans_1({\cE}) = \cD$ and ${\cE} \models \varphi^1$.
	\end{enumerate}

\paragraph*{The translation $\trans_2$}

Let $\Xi = \{\xi_a \mid a \in [A]\}$.
For every 2-type $\beta$ over 
$\dvoc{\Sigma} \cup \vocSNF$,
let $\beta^{\Xi}$ be the 2-type over $\dvoc{\Xi}$ such that:
\begin{itemize}
 \item For every $\alpha(x,y)$ which is one of $R(x,y)$ or $\neg R(y,x)$ for $R \in \{ \leq_1, \lesssim_2, S_2\}$, 
 $\beta(x,y) \models \alpha(x,y)$ if and only if $\beta^{\Xi}(x,y) \models \alpha(x,y)$. 
 \item For every $a \in [A]$ and $z \in \{x,y\}$, 
 $\beta(x,y) \models \nu_a(z)$ if and only if $\beta^{\Xi}(x,y) \models \xi_a(z)$. 
\end{itemize}

Let $\varphi^2 \in \FO^2(\dvoc{\Xi})$ be the formula obtained from $\varphi^1$ by
replacing the 1-types $\nu_a(x)$ with $\xi_a(x)$
and the 2-types $\beta(x,y)$  and $\beta_{abc}(x,y)$ with $\beta^\Xi(x,y)$ and $\beta_{abc}^\Xi(x,y)$ respectively: 
\[
\varphi^2= \varphi^2_\forall \wedge \varphi^2_\exists  
\]
where
\[
\begin{array}{llll}
\varphi^2_\forall & = & 	\forall x \forall y \bigvee_{\beta^\Xi \in \Theta_\forall^\Xi} \beta^\Xi(x,y), \\
\varphi^2_\exists & = & \forall x \, \bigwedge_{a = 1}^{A}{\xi_a(x) \rightarrow 
\bigwedge_{b=1}^{B}{\exists y \, \bigvee_{c=1}^{C}{\beta_{a b c}^\Xi(x,y)}}}.
\end{array}
\]
and where $\Theta_\forall^\Xi = \{\beta^\Xi \mid \beta \in \Theta^\beta_\forall\}$.

Finally, we define the translation 
$
\trans_2: \dat{\Xi} \rightarrow \str{\dvoc{\Sigma} \cup \vocSNF}
$. 
Let $\cD \in \dat{\Xi}$ with universe $D$.
We define $\trans_2(\cD)$ as follows:
\begin{itemize}
	\item 
	The universe and order relations of $\trans_2(\cD)$ are identical to those of $\cD$.
	\item
	For every $\xi_a \in \Xi$ and $d \in D$, if $\cD \models \xi_a(d)$ then 
	$d$ has 1-type $\nu_a(x)$ in $\trans_2(\cD)$. 
\end{itemize}

Observe that for a data word $\cD \in \dat{\Xi}$ with universe $D$, we have 
for all $d,d' \in D$ and every 2-type $\beta^\Xi$,
$\cD \models \beta^\Xi(d,d')$   if and only if $\trans_2(\cD) \models \beta(d,d')$.
Hence, from~\ref{proof_transa_i} and~\ref{proof_transa_i}:
	\begin{enumerate}[label=(\roman*${}_2$),ref=(\roman*${}_2$)]
		\item \label{proof_transb_i}
		For every $\cE \in \dat{\Xi}$,  
		if $\cE \models_{\dat{\Xi}} \varphi^2$ then $\trans_2(\cE) \models \varphi^1$, and
		\item \label{proof_transb_ii}
		For every $\cD \in \str{\dvoc{\Sigma} \cup \vocSNF}$, 
		if $\cD \models \varphi^1$ 
		then there exists $\cE \in \dat{\Xi}$ such that $\trans_2(\cE) = \cD$ and $\cE \models_{\dat{\Xi}} \varphi^2$.
	\end{enumerate}

Let $h$ be a letter-to-letter substitution given as follows: for every $\xi_a \in \Xi$, 
$h(\xi_a)=\sigma_a$, where $\sigma_a$ is the unique letter in $\Sigma$ such that 
$\nu_a (x) \models \sigma_a(x)$. 
Let $\overline{h}$ be the function 
  $\overline{h}:\dat{\Xi}\to\dat{\Sigma}$
  such that for every $\cD \in \dat{\Xi}$, $\overline{h}(\cD) = \cE$, where $\cE$
  has the same universe and order relations as $\cD$, and where the interpretation
  $\sigma^\cE$ of $\sigma\in\Sigma$ in $\cE$ is 
  $
   \bigcup_{\xi \in \overline{h}^{-1}(\sigma)} \xi^\cD
  $.
Note that $\overline{h}$ is the composition of $\trans_2$ and $\trans_1$. 
Using \ref{proof_transa_i}, \ref{proof_transb_i}, 
\ref{proof_transa_ii}, and \ref{proof_transb_ii}, we have:
\begin{enumerate}[label=(\roman*${}_h$)]
	\item \label{lemma:snf:i}
	For every $\cE \in \dat{\Xi}$, 
	$\cE \models_{\dat{\Xi}} \varphi^2$ implies $\overline{h}(\cE) \models_{\dat{\Sigma}} \psi$, and
	\item \label{lemma:snf:ii}
	For every $\cD \in \dat{\Sigma}$, 
	$\cD \models_{\dat{\Sigma}} \psi$ implies the existence of $\cE \in \dat{\Xi}$ 
	such that $\overline{h}(\cE) = \cD$ and $\cE \models_{\dat{\Xi}} \varphi^2$.
\end{enumerate}

Let $\hat{h}$ be the function 
 $\hat{h}:2^{\Xi^\star}\to 2^{\Sigma^\star}$ which transforms every word $u$ in the input language 
by substituting the letters according to $h$.
Now we can prove that 
$L(\psi) = \hat{h}(L(\varphi^2))$. 
Observe that for $\cE \in \dat{\Xi}$, 
$\hat{h}(\{\project(\cE)\}) =  \{\project(\overline{h}(\cE))\}$.

Let $u \in \hat{h}(L(\varphi^2))$. 
There is some $\cE \in \dat{\Xi}$ such that $\cE \models \varphi^2$ and $\{u\} = \hat{h}(\{\project(\cE)\})$
and hence $u = \project(\overline{h}(\cE))$. By~\ref{lemma:snf:i}, 
$\overline{h}(\cE) \models \psi$, and hence $u \in \hat{h}(L(\varphi^2))$. 

Conversely, 
let $u \in L(\psi)$. There is some $\cD \in \dat{\Sigma}$ such that $\cD \models \psi$ and $u = \project(\cD)$. 
By~\ref{lemma:snf:ii}, 
there is $\cE \in \dat{\Xi}$ 
	such that $\overline{h}(\cE) = \cD$ and $\cE \models_{\dat{\Xi}} \varphi^2$.
	Hence, $\project(\cE) \in L(\varphi^2)$. 
	We have $\hat{h}(\{\project(\cE)\}) = \{\project(\overline{h}(\cE))\} = \{ \project(\cD)\} = \{u\}$, and hence
	$u \in \hat{h}(L(\varphi^2))$. 

Given $\psi$, the formula $\varphi^0$ can be computed in polynomial time in the length of $\psi$. 
The size of $\vocSNF$ is linear in the length of $\psi$. 
W.l.o.g. we can assume that every symbol in $\dvoc{\Sigma} \cup \vocSNF$ occurs in $\psi$. 
Then the number of 1-types and 2-types over 
$\dvoc{\Sigma} \cup \vocSNF$ is at most exponential in the length of $\psi$, and 
the formulas $\varphi^1$ and $\varphi^2$ can be computed in exponential space. 
The lemma follows with the notation slightly simplified by replacing
$\beta^\Xi$ with $\theta$, 
$\beta_{abc}^\Xi$ with $\theta_{abc}$,
$\Theta_{\forall}^\Xi$ with $\Theta_{\forall}$,
$\varphi_\exists^2$ with $\varphi_\varepsilon \land \forall x \, \bigwedge_{a=1}^{A}{\xi_a(x) \rightarrow \bigwedge_{b=1}^{B}{\exists y \, \bigvee_{c=1}^{C}{\theta_{a b c}(x,y)}}}$, and
$\varphi_\forall^2$ with $\forall x \forall y \, \bigvee_{\theta \in \Theta_\forall} \theta(x,y)$.

\end{proof}

\section{Proofs for the Automata-logic Connection}
\subsection{Proof of Lemma~\ref{lem_good_decoration}}
\proofgooddecoration

\subsection{Proof of Uniqueness of \texorpdfstring{$\ttrim{\cT}{1}$}{T1}}
\label{app_unique}

\begin{lemma}
\label{lemma_unique_tasked_trimming1}
\label{lemma_trim_omega_compat}
Let $\cT$ be a $\cD$-task word. 
There is a unique $\trim{\cD}{1}$-task word $\cT_1$ for which $\omega(\ts)\,{=}\,\omega(\ts_1)$ 
for each $\ts,\ts_1\,{\in}\,\compatTasks$ and each $d$ in the universe of $\trim{\cD}{1}$ 
s.t.
$\cT\,{\models}\,\ts(d)$ and $\cT_1\,{\models}\,\ts_1(d)$.
\end{lemma}

\proofuniquetasked

\subsection{Definitions and Proofs for Perfect Extremal Strings} 
\label{app_perfect}

\begin{definition}[$\perfect_{\alpha,\beta}(x,y)$]
\label{def_perfect_formula}
Let $\alpha = (\hlayerp{\alpha},\ts_\alpha)$ and $\beta = (\hlayerp{\beta},\ts_\beta)$
in $\Gamma$ with
at least one of them  
in $\Gammatop$.
Using the subformulas in Table~\ref{tab_perfect}, 
we define
$\perfect_{\alpha,\beta}(x,y) = \perfectalpha_{\alpha}(x) \wedge \perfectalpha_{\beta}(y) \wedge  
\bigwedge_{\mathit{bin} \in \{\leq_1, \lesssim_2, S_2\}}{\perfectalpha_{\alpha,\beta,\mathit{bin}}(x,y)}.$
\end{definition}

 \begin{table}
\caption{Components of  $\perfect_{\alpha,\beta}(x,y)$}
	\centering
	\renewcommand{\arraystretch}{1.75}
	\begin{tabular}{l@{\,=\,}l}
	\hline
	$\perfectalpha_{\alpha}(x)$ & $\xi^{\omega(\ts_\alpha)}(x)$; 
	\\
	$\perfectalpha_{\beta}(y)$  & $\xi^{\omega(\ts_\beta)}(y)$;
	\\
	$\perfectalpha_{\alpha,\beta,\leq_1}(x,y)$  & $x \lneq_1 y$; 
	\\
	$\perfectalpha_{\alpha,\beta,\lesssim_2}(x,y)$ &
	$\begin{cases}
	x \lnsim_2 y, & h_\alpha \not= \onetoprm \\
	y \lnsim_2 x, & h_\beta \not=\onetoprm \\
	x \sim_2 y, & \text{Otherwise};
	\end{cases}
	$
	\\
	$\perfectalpha_{\alpha,\beta,S_2}(x,y)$  & 
	$\begin{cases}
	S_2(x,y), &  h_\alpha = \twotoprm \\
	S_2 (y,x), &  h_\beta = \twotoprm \\
	\neg S_2(x,y), &  h_\alpha = \restrm \\
	\neg S_2 (y,x), &  h_\beta = \restrm \\
	\mathit{True}, &  \text{Otherwise}
	\end{cases}
	$ 
	\end{tabular}
	\label{tab_perfect}
 \end{table}

\begin{lemma}
\label{cl_2}
Let $\cT$ be a $\cD$-task word with universe $D$ and let $w = \dataAbs(\cT)$. 
Let $d_1,d_2 \in D$ be such that
$\cD \models d_1 \lneq_1 d_2$
and 
$ \maxdv_\cD \in \{\val{\cD}{d_1},\val{\cD}{d_2}\} $. 
Let $\ell_1$ and $\ell_2$ be such that $d_1$ is mapped to position $\ell_1$ in $\dataAbs(\cT)$ 
and 
$d_2$ is mapped to position $\ell_2$ in $\dataAbs(\cT)$.
Then $\ell_1\lneq \ell_2$ and
$\cD \models \perfect_{w(\ell_1),w(\ell_2)}(d_1,d_2)$.
\end{lemma}
\proofcltwo

\begin{lemma}\label{lem_cor_perfect_2type}
Consider $\alpha, \beta \in \Gamma$  such that at least one of them is in  $\Gammatop$
and such that   $\perfect_{\alpha,\beta}(x,y) \not\models_{\dat{\Xi}} \chi(x,y) \wedge \chi(y,x)$.
Then it holds that $\exists x\, \exists y\, \perfect_{\alpha,\beta}(x,y) \models_{\dat{\Xi}} \neg \varphi_\forall$.
\end{lemma}

We first recall the following: 
\begin{obs}
\label{lem_perfect_2type}
Let $\alpha, \beta \in \Gamma$ such that at least one of them is in $\Gammatop$. 
There exists a 2-type $\theta(x,y)$ such that 
$\perfect_{\alpha,\beta}(x,y) \equiv_{\dat{\Xi}} \theta(x,y)$. 
\end{obs}

\proofperfectwotype

\subsection{Proof of Lemma~\ref{lem_uni_task}}
\proofunitask

\subsubsection{A syntactic representation of consecutive extremal strings}
\label{subsec_syntactic}

The extremal string $\rcon{r}{g}{s^0}$ simulates the extension of a task word $\cT^0$ whose extremal string is $s^0$ by adding 
elements with a new maximal data value. 
The letters of these elements are determined by $r$ and their placement in the linear order of $\cT^0$
is determined by $g$. 

We introduce some notation. 
\[
\begin{array}{llll}
\positions_{h}(w) &=& \{\ell \in [\length{w}] \mid w(\ell) \in \Gamma_h\}	\\
\positions_{<\onetoprm}(w) &=& \{\ell \in [\length{w}] \mid w(\ell) \in \Gammattop \cup \Gammarest\}
\end{array}
\]

\begin{definition}[$\rcon{r}{g}{s^0}$]
\label{def:rcon}
\label{def:downarrow}
Let $s \in \Gamma^\star$ . We denote by $\da{s}$ the string that is obtained from $s$ by substituting
 letters of the form $(\otop,\ts)$ with $(\ttop,\ts)$ and
letters of the form $(\ttop,\ts)$  or $(\rest,\ts)$ with $(\rest,\ts)$. 
Let $s^0 \in \EXT(\Gamma)$, $r \in \Gammatop^\star \cap \PGamma^\star$, and
let $g:[n_r] \to [n_r + n_{s^0}]$ be a strictly monotone function. 
Then let $s^1 = r \shuffle_g \da{s^0}$, 
where $u \shuffle_g v$ denotes the shuffle of $u = w(i_1)\cdots w(i_{\length{u}})$ and $v$ associated to the function $g$ that has $g(i) = i_j$ for all $1 \leq j \leq \length{u}$, see Definition~\ref{shuffle-def}.

We now define the string $\rcon{r}{g}{s^0}$, which is obtained from 
$s^1$ substituting every letter
$(\hlayer,\ts)$ at position $\ell \in [n_{s^1}]$
with a letter $(\hlayer,\ts')$. 
Essentially, this replacement updates the tasks to take into account the new elements
  described by $r$ (and whose position is given by $g$). 
The letter $(\hlayer,\ts')$ is 
 such that 
$\omega(\ts)=\omega(\ts')$ and,
for every $\theta \in \omega(\ts)$,
we have $C_\theta\in \ts'$ if $\theta$ is completed, that is, if any of the following hold
\begin{enumerate} 
 \item $\theta$ was already completed: $C_\theta \in \ts$. 
 \item The necessary witness $\ell_2$ for $\ell$ was found  to the right:
 $\theta \models_{\DW(\Xi)} x \leq_1 y$ and 
 there is
$\ell_2 \in [n_{s^1}]$ such that
$\ell\leq \ell_2$, 
either $\ell \in \positions_{\onetoprm}(s^1)$
or $\ell_2 \in \positions_{\onetoprm}(s^1)$, and 
$\perfect_{s^1(\ell),s^1(\ell_2)}(x,y) \equiv_{\dat{\Xi}} \theta(x,y)$.
\item The necessary witness $\ell_2$ for $\ell$ was found  to the left:
 $\theta \models_{\DW(\Xi)} y <_1 x$ and 
there is
$\ell_1 \in [n_{s^1}]$ such that $\ell_1 < \ell$, 
either $\ell_1 \in \positions_{\onetoprm}(s^1)$
or $\ell \in \positions_{\onetoprm}(s^1)$, and 
$\perfect_{s^1(\ell_1),s^1(\ell)}(y,x) \equiv_{\dat{\Xi}} \theta(x,y)$.
\end{enumerate}
Otherwise, since $\theta$ was not completed, it remains as a promised task and
we have $P_\theta \in \ts'$. 
\end{definition}

The pair $(s^0,\rcon{r}{g}{s^0})$ of extremal strings 
is consecutive.

\begin{lemma}
\label{lem:syntactic_to_consec}
Let $\cT_0$ be a $\cD_0$-task word,
$s^0 = \ext(\cT_0)$, 
$r \in \Gammatop^\star \cap \PGamma^\star$, and
let $g:[n_r] \to [n_r + n_{s^0}]$ be a strictly monotone function. 
Then $s^0$ and $\ext(\rcon{r}{g}{s^0})$ are consecutive. 
\end{lemma}
\proofsyntactictoconsec

\begin{example}
\label{ex_da}
Applying $\downarrow$ to $\ul{s}'$, 
we have \[ \da{\ul{s}'} = (\rest, \ul{ts}_3^P)(\rest, \ul{ts}_3^P)(\ttop,
\ul{ts}_1^P) (\rest, \ul{ts}_3^P) \] 
Let $\ul{r} = (\otop, \ul{ts}_2^P)(\otop, \ul{ts}_4^P)$, and let $\ul{g}: [2] \rightarrow [2+4]$ be $\ul{g}(1) = 4$ and $\ul{g}(2) = 6$.
Then  
\[
\begin{array}{lll}
\rcon{\ul{r}}{\ul{g}}{\ul{s}'}  &=& (\rest, \ul{ts}_3^C)(\rest, \ul{ts}_3^C)(\ttop, \ul{ts}_1^C) \\  && (\otop, \ul{ts}_2^C)(\rest, \ul{ts}_3^C)(\otop, \ul{ts}_4^C).	
\end{array}
\]
Observe that $\ext(\rcon{\ul{r}}{\ul{g}}{\ul{s}'}) = \ul{s}$, and recall that $(\ul{s}',\ul{s})$ are consecutive.
\end{example}

If we have the $\cD$-task word $\cT$ at hand, then we can obtain the 
string $r$ and the function $g$ for $s$ and $s^0$. 
In fact, obtaining $r$ is quite easy: it suffices to look at the elements with maximal
data value,  and substitute every letter $(\otop,\ts)$ with 
$(\otop,\ts_P)$ such that $\ts_P = \{P_\theta \mid \theta \in  \omega(\ts)\}$.  

\begin{lemma}
\label{lem:consec_to_syntactic}
Given a $\cD$-task word $\cT$ and two extremal strings $s$ and $s^0$ such that
$s = \ext(\cT)$ and $s^0 = \ext(\ttrim{\cT}{1})$, we can effectively obtain an 
$r \in \Gammatop^\star \cap \PGamma^\star$ and $g: [|s^0|] \rightarrow [|s^0| + |s|]$
such that $s=\ext(\rcon{r}{g}{s^0})$.
\end{lemma}
\proofconsectosyntactic

\new{Unfortunately, we cannot keep the task word. This is not a problem,
  because we can show that we do not need the task word to keep track of the relevant 
 information on how positions evolve along extremal strings. 
 Let $s,s'$ be consecutive extremal strings, and let $\cT$ be a task word such
 that  $s = \ext(\ttrim{\cT}{1})$ and $s' = \ext(\cT)$. Then we can define a  
 \emph{partial embedding from $s$ to $s'$ via $\cT$}, which maps each
 position in $s$ that corresponds to an \emph{extremal element} in $\cT$, to its
 position in $s'$, where the extremal elements of $\cT$ are the ones that
 correspond to positions in $\ext(\cT)$. The precise definition is given in
 App.~\ref{proof_prop_f_extremal}. 
 Crucially, this embedding is independent of the specific $\cT$, and can be
 effectively constructed.} 

\begin{lemma}
\label{lem_check_consec_expspace_copy}
Given two extremal strings $s,s'$, whether they are consecutive can be decided
in $\EXPSPACE$.  
If they are, then we can also obtain in $\EXPSPACE$ a partial embedding
$\notOneTopEmbeddingNoArg{s}{s'}$  from
positions in $s$ to positions in $s'$ that coincides with the 
partial embedding from $s$ to $s'$ via
  $\cT$ for every task word $\cT$ such that 
 $s' = \ext(\cT)$ and $s = \ext(\ttrim{\cT}{1})$. 
\end{lemma}

\proofcheckconsec

\subsection{Witnesses of task completion}
\label{subsec_def_tasked}

Here we prove a few lemmas which provide witnesses to the completion of tasks based on existing witnesses to the same or other tasks.

\begin{lemma}
\label{lemma_twotop_type}
Let $\cD$ be a data word and let $d,d'$ be elements with the same 1-type such that $\cD \models d \leq_1 d' \wedge d \sim_2 d'$.
Let $\theta \in \Theta_\exists$.
\begin{enumerate}
	\item 
	If $\theta(x,y) \models x \leq_1 y$ and $\cD \models \exists y \ \theta(d',y)$, then $\cD \models \exists y \ \theta(d,y)$.
	\item
	If $\theta(x,y) \models y \leq_1 x$ and $\cD \models \exists y \ \theta(d,y)$,\hphantom{'} then $\cD \models \exists y \ \theta(d',y)$.
\end{enumerate}
\end{lemma}
\begin{proof}
\begin{enumerate}
	\item 
	Since $\cD \models d \leq_1 d' \wedge d \sim_2 d'$, and $d$, $d'$ have the same 1-type,
	for any element $d''$ such that $\cD \models d' \leq_1 d''$, the 2-type of $(d',d'')$ is the same as
	the 2-type of $(d,d'')$. Since $\theta(x,y) \models x \leq_1 y$, if $\cD \models \exists y \ \theta(d',y)$ there exists some $d'' \geq_1 d'$ such that
	$\cD \models \theta(d',d'')$. Therefore also $\cD \models \theta(d,d'')$ and $\cD \models \exists y \ \theta(d,y)$.
	\item
	Analogous to the previous case.
\end{enumerate}
\end{proof}

\begin{lemma}
\label{lemma_cor_twotop_type}
Let $\cT$ be a $\cD$-task word and let $d,d'$ 
be elements such that $\cD \models d \leq_1 d' \wedge d \sim_2 d'$.
Let $\cT \models \ts_d(d)$
and $\cT \models \ts_{d'}(d')$
such that $\ts_{d}$ and $\ts_{d'}$ realize the same witness type set.
Let $\theta \in \Theta_\exists$.

\begin{enumerate}
	\item 
	If $\theta(x,y) \models x \leq_1 y$ and $C_\theta \in \ts_{d'}$, then $C_\theta \in \ts_{d}$.
	\item
	If $\theta(x,y) \models y \leq_1 x$ and $C_\theta \in \ts_{d}$, then $C_\theta \in \ts_{d'}$.
\end{enumerate}
\end{lemma}
\begin{proof}
If $\ts_{d}$ and $\ts_{d'}$ realize the same witness type set, then $d$ and $d'$ have the same 1-type in $\cD$. 
Then the claim follows from Lemma~\ref{lemma_twotop_type} and the definition of task words.
\end{proof}

\begin{lemma}
\label{lemma_rest_theta}
Let $\cD$ be a data word and let $d,d'$ be elements with the same 1-type such that $\cD \models d \leq_1 d'$,
$\val{\cD}{d} \leq \maxdv_{\cD}-2$, and 
$\val{\cD}{d'} \leq \maxdv_{\cD}-2$. 
Let $\theta \in \Theta_\exists$ be such that $\trim{\cD}{1} \not\models \exists y \ \theta(d,y)$ 
and $\trim{\cD}{1} \not\models \exists y \ \theta(d',y)$.

\begin{enumerate}
	\item 
	If $\theta(x,y) \models x \leq_1 y$ and $\cD \models \exists y \ \theta(d',y)$, then $\cD \models \exists y \ \theta(d,y)$.
	\item
	If $\theta(x,y) \models y \leq_1 x$ and $\cD \models \exists y \ \theta(d,y)$, \, then $\cD \models \exists y \ \theta(d',y)$.
\end{enumerate}
\end{lemma}
\begin{proof}
\begin{enumerate}
	\item 
	Let $\cD \models \exists y \ \theta(d',y)$, and denote by $d''$ the element such that $\cD \models \theta(d',d'')$.
	Since $\trim{\cD}{1} \not\models \exists y \ \theta(d',y)$, we conclude that $\val{\cD}{d''}=\maxdv_\cD$.
	Since both and $d$ and $d'$ have data value at most $\maxdv_\cD-2$, 
	it holds that $\cD \models d \lnsim_2 d'' \wedge \neg S_2(d,d'')$ and $\cD \models d' \lnsim_2 d'' \wedge \neg S_2(d',d'')$.
	Since $\cD \models d \leq_1 d'$ and $\theta(x,y) \models x \leq_1 y$, also $\cD \models d \leq_1 d''$, and all in all, 
	the 2-type of $(d',d'')$ is the same the 2-type of $(d,d'')$. Therefore, $\cD \models \theta(d,d'')$, implying $\cD \models \exists y \ \theta(d,y)$.
	\item
	Analogous to the previous case.
\end{enumerate}

\end{proof}

\begin{lemma}
\label{lemma_cor_rest_theta}
Let $\cT$ be a $\cD$-task word and let $d,d'$ be elements such that $\cD \models d \leq_1 d'$,
$\val{\cD}{d} \leq \maxdv_{\cD}-2$, and 
$\val{\cD}{d'} \leq \maxdv_{\cD}-2$. 
Let $\ttrim{\cT}{1} \models \ts_1(d)$
and $\ttrim{\cT}{1} \models \ts_1'(d')$ such that
$\omega(\ts_1)=\omega(\ts_1')$.
Let $\theta \in \Theta_\exists$ such that $P_\theta \in \ts_1\cap \ts_1'$.
Finally, let $\cT \models \ts(d)$ and $\cT \models \ts'(d')$.
\begin{enumerate}
	\item 
	If $\theta(x,y) \models x \leq_1 y$ and $C_\theta \in \ts'$, then $C_\theta \in \ts$.
	\item
	If $\theta(x,y) \models y \leq_1 x$ and $C_\theta \in \ts$, then $C_\theta \in \ts'$.
\end{enumerate}

\end{lemma}
\begin{proof}
Since $\ts_{1}$ and $\ts_1'$ realize the same witness type set, $d$ and $d'$ have the same 1-type in $\cD$.
Since $P_\theta \in \ts_1\cap \ts_1'$, we have 
$\trim{\cD}{1} \not\models \exists y \ \theta(d,y)$ 
and $\trim{\cD}{1} \not\models \exists y \ \theta(d',y)$. 
Then the claim follows from Lemma~\ref{lemma_rest_theta} and the definition of task words.

\end{proof}

\begin{lemma}[Extremal witnesses]
\label{lem:extemal-witnesses} 
 Let $\cT$ be a $\cD$-task word. Let $d,d_0$ be elements of $\cT$ and $\theta\in\Theta_\exists$. 
 If $\cD \models \theta(d,d_0)$
 and $\maxdv_\cD\in\{\val{\cD}{d},\val{\cD}{d_0}\}$, 
 then there is an element $d' \in \extElem(\cT)$
 such that $\cD \models \theta(d,d')$. 
\end{lemma}

\begin{proof}
Let $w=\dataAbs(\cT)$.  
 Let $w(\embeddingOp{\dataAbs}{\cT}^{-1}(d))=(\hlayer,\ts)$, and let
  $w(\embeddingOp{\dataAbs}{\cT}^{-1}(d_0))=(h_0,\ts_0)$. 
 Let $\theta \in \ts_0$. There are $d_{0,1} \leq_1 d_0 \leq_1 d_{0,2}$ such that
 \begin{enumerate} 
 \item $d_{0,1},d_{0,2} \in \extElem_{\theta,h}(\cT)$, 
  \item $\cD \models \xi^{\omega(\ts_0)}(d_{0,1})$ and $\cD \models \xi^{\omega(\ts_0)}(d_{0,2})$, 
  \item $\val{\cD}{d_0} = \val{\cD}{d_{0,1}} = \val{\cD}{d_{0,2}}$
 	if $\val{\cD}{d_0} \geq \maxdv_\cD-1$,  
  \item both	$\val{\cD}{d_{0,1}} \leq  \maxdv_\cD-2$ and
 	$\val{\cD}{d_{0,2}} \leq \maxdv_\cD-2$  
	if we have $\val{\cD}{d_0} \leq \maxdv_\cD-2$, and 
 \item $d_0,d_{0,1}$, $d_{0,2}$ have the same 1-type in $\cD$. 
 \end{enumerate}
 Since $\maxdv_\cD\in\{\val{\cD}{d},\val{\cD}{d_0}\}$, 
 the 2-types of $(d,d_0)$ and either $(d,d_{0,1})$ or $(d,d_{0,2})$ are equal  depending on whether $\theta(x,y) \models x \leq_1 y$;
 for the case that $\val{\cD}{d_0} \geq \maxdv_\cD-1$ we use that  $\cD \models d_0 \sim_2 d_{0,1}$
 and $\cD \models d_0 \sim_2 d_{0,2}$,
 while the case that $\val{\cD}{d_0} \leq \maxdv_\cD-2$ is 
 similar to the discussion in the proof of 
 Lemmas~\ref{lemma_rest_theta}. Hence the lemma follows by setting either 
 $d'=d_{0,1}$ or $d'=d_{0, 2}$.
\end{proof}

\subsection{Lemmas for \texorpdfstring{\EXPSPACE}{EXPSPACE}}
\label{app_lemmas}
\label{app_proof_lem:delete_nonext}

To guarantee we non-deterministically explore the entire transition relation, we use 
the following lemma to bound
the search space:
\begin{lemma}
\label{lem:delete_nonext}
Let $\cT$ be a $\cD$-task word. Let $d$ be an element of $\cT$ not in $\extElem(\cT)$
such that $\val{\cD}{d} = \maxdv_\cD$. 
Let $\cD_{-d}$ and $\cT_{-d}$ 
be the substructures of $\cD$ respectively $\cT$ obtained by removing $d$. 
Then $\cT_{-d}$ is a $\cD_{-d}$-task word and $\ext(\dataAbs(\cT)) = \ext(\dataAbs(\cT_{-d}))$. 
\end{lemma}

Assume that $\cT_{-d}$ is not a task word. 
There are
$d' \in D$, $\ts \in \compatTasks$, and $\theta \in \omega(\ts)$ such that $d \not= d'$, 
$\cT \models \ts(d)$, and $C_\theta \in \ts$,
but for every element $d'' \not=d$ of $\cD_{-d}$, $\cD_{-d} \not\models \theta(d',d'')$.
However, the non-existence of such $d''$ is in contradiction to Lemma~\ref{lem:extemal-witnesses}. 
Since $d \notin \extElem(\cT)$, $\extElem(\cT) = \extElem(\cT_{-d})$.
The string $\dataAbs(\cT_{-d})$ is the substring of $\dataAbs(\cT)$ obtained by deleting the letter corresponding to $d$,
and hence $\ext(\dataAbs(\cT)) = \ext(\dataAbs(\cT_{-d}))$.

\begin{lemma}
\label{lem:automaton_size}
Let $n = |\psi|$. 
The size of $\pA$ is at most double exponential in $n$, while $m$ is at most exponential in $n$. 
\end{lemma} 

\begin{proof}
By Lemma~\ref{lem_th_our_snf}, $|\Xi|$ and $|\Theta|$ and hence $m$ are exponential in $n$. 
The size of $\Gamma$ is therefore double exponential in $n$. 
We have $\EXT(\Gamma) \subseteq \Gamma^{7 \cdot |\Theta_\exists|}$, which is double exponential in $n$. 
Given an extremal string $s$ of length $n_s$, 
the number of $(m+1,n_s)$-pebble assignments is at most $(n_s+1)^{m+1}$,
which is double exponential in $n$. Hence $|Q|$ is double exponential. 
Since 
$\delta \subseteq (Q \times Q) \cup (Q \times \MoveS_{m+1} \times \Sigma \times Q)$,
$|\delta|$ is double exponential. 
\end{proof}

\begin{lemma}
\label{lem:delta_exptime}
$\delta$ is $\EXPSPACE(\log(|\pA|))$-computable. 
\end{lemma} 

\begin{proof}
Let $n = |\psi|$. 
Using Lemma~\ref{lem:automaton_size}, 
the sizes of the representation of a state $q\in Q$, a string $s \in \Gamma^{7 \cdot |\Theta_\exists|}$,
a pebble $k \in [m+1]$, and a letter $\gamma \in \Gammatop$ are all at most exponential in $n$. 
We  can verify that $s$ is an extremal string in exponential space
by computing 
$\extPos(s)$ and verifying that 
$|\extPos(s)|=|s|$. The set $\extPos(s)$ can be computed by going over the string $s$
and keeping track of the relevant minimum and maximum positions
(of which there are an exponential number). 
Similarly, for $\tilde{\ell} \in [|s|]$, one can verify
that $\gamma \in \nonext{s}{\tilde{\ell}}$. 

Verifying that an extremal string $s$ of length $n_s$
is perfect is done as follows: for every two positions $\ell_1<\ell_2$ of $s$ 
such that
$\{\ell_1,\ell_2\} \cap \positions_{\onetoprm}(s) \not= \emptyset$,
it is straightforward to compute the formula $\perfect_{s(\ell_1),s(\ell_2)}(x,y)$ in exponential space. 
Since the formula $\perfect_{s(\ell_1),s(\ell_2)}(x,y)$ is a conjunction of atoms and negations of atoms,
it is also easy to complete it to its  equivalent 2-type $\beta(x,y)$. 
We can then check that $\beta(x,y),\beta(y,x)\in \Theta_\forall$. 

We know from \hyperref[sec_expspace]{our previous discussion} and Lemma~\ref{lem:delete_nonext} that is it possible to non-deterministically compute the extremal strings 
which are consecutive to $s$, and hence the set of transitions leaving $q \in Q$, in $\EXPSPACE$. 
\end{proof}

\section{The automaton \texorpdfstring{$\pA^\vphi$}{A-varphi}}
\label{sec_aut_def}
\label{subse:def_automaton}

We define an automaton $\pA^{\varphi} = (\Xi, m+1, Q, q_{\mathrm{init}}, F, \delta)$ where
$Q = Q_e \cup Q_p$
and 
$m = 7 \cdot |\Theta_\exists|$.
$\pA^{\varphi}$ uses one pebble for each existential constraint in
$\Theta_\exists$ and each layer in $\Gamma$, plus an additional pebble per
constraint. It also uses the designated pebble $m+1$ to read non-extremal
positions. 
We describe its states and transitions next. 

\noindent
\textbf{Non-prefix states.}
$Q_e$ is the set of states of the form $(s,\tau)$ with 
$s$ a  perfect extremal string and
$\newtau$ an 
$(m+1,\length{s})$-pebble assignment
satisfying the following conditions, which hold when $s$ has just been
read: 
\begin{enumerate}[label={(\subscript{\bf c}{\arabic*}}),ref=(\subscript{c}{\arabic*}),series=conditions] 
	\item \label{tau_c_mPlusOneIsBotAgain}
	the pebble $m+1$ was not used to read $s$, i.e., we have $\newtau(m+1)=\bot$, 
	\item \label{tau_c4} 
        every position of $s$ has a pebble on it, i.e., 
	if $s \neq \varepsilon$, then $[\length{s}] \subseteq \newtau([m])$, and
	\item \label{tau_c5}
        if $s=\varepsilon$, there are no pebbles on $s$, i.e., 
         $\newtau = \rho_\bot$, in other words, $\newtau([m+1])
        = \{\bot\}$. 
\end{enumerate}

We would like the automaton to transition from $(s,\tau) \in Q_e$ 
 to $(s', \tau') \in Q_e$ if $s$ and $s'$ are consecutive. 
But since the automaton can only move on pebble at a time, we need some 
 intermediate steps. 
We therefore have another set of \emph{prefix states} $Q_p$
that the automaton uses to read  extremal strings from left to right, 
by iterating over all their prefixes. 

\noindent
\textbf{Prefix states.}
$Q_p$ is the set of states of the forms $(s,\tilde{s},\tau,0)$ and $(s,\tilde{s},\tau,1)$ for every
perfect extremal string $s$,
non-empty prefix $\tilde{s}$ of $s$, and 
$(m+1,\length{s})$-pebble assignment
$\newtau$  
satisfying similar conditions as before, but which now hold if only the prefix
$\tilde{s}$ has been read: 
\begin{enumerate}[resume*=conditions]
	\item \label{tau_c_mPlusOneIsBot}
	$\newtau(m+1)=\bot$, 
	\item \label{tau_c2}
	$[\length{{\tilde{s}}}-1]\subseteq\newtau([m])$, and
	\item \label{tau_c3}
	pebbles beyond the current prefix had been placed previously, i.e., for every $\length{{\tilde{s}}} \leq \ell \leq \length{s}$,
	$\ell \in \newtau([m])$  if and only if $s(\ell) \notin \Gammatop$.
\end{enumerate}

The $0$/$1$ flag in prefix states is used below for deciding where to place the
$m+1$ pebble. 
\vspace{0.1em}
\\
\textbf{Initial state.}
The initial state is $q_{\mathrm{init}} = (\varepsilon, \rho_\bot)\in Q_e$.
\vspace{0.1em}
\\
\textbf{Final states.}
The final states are 
$F = \{(s,\tau) \in Q_e \mid s\mbox{ is completed}\}$. 

For a state $q \in Q_p$ of the form $(s,\tilde{s},\tau,0)$ or $(s,\tilde{s},\tau,1)$, or $q \in Q_e$
of the form $(s,\tau)$, denote $\tau_q = \tau$. 
For a state 
$q \in Q_p \cup Q_e$, we say 
a pebble $k \in [m+1]$ is \emph{available in $q$} if 
$\newtau_q(k) = \bot$. 
Note that the pebble $m+1$ is available by~\ref{tau_c_mPlusOneIsBot} and~\ref{tau_c_mPlusOneIsBotAgain}. 
Let $q = (s,\tilde{s},\tau_q,b)$ or $q = (s,\tau_q)$ with $|s| = n_s$. 
Since $\newtau_q$ is an $(m+1,n_s)$-pebble assignment, there are at most $n_s$ pebbles $k$ with $\newtau_q(k)\not=\bot$. 
By the bound on the length of extremal strings, we have $n_s \leq  7 \cdot |\Theta_\exists| < m$, therefore there 
is at least one pebble $k\in [m]$ which is available in $q$. All in all, we have:
\begin{obs}
For every $q \in Q_p \cup Q_e$, there is a pebble $k \in [m]$ such that both pebble $k$ and pebble $m+1$ are available in $q$. 
\end{obs}

\noindent \textbf{Transitions from prefix states.}
We can now define $\delta$. 
Let $q = (s,\tilde{s}, \tau_q,b) \in Q_p$.
\begin{enumerate} 
	\item \label{trans_q1_non_ext} \emph{Non-extremal transitions}: while
          reading a prefix of a task word, pebble $m+1$ iterates over
          non-extremal positions. 
	For every letter $\gamma \in \nonext{s}{\length{\tilde{s}}}$, we have that
	$(q,\Move{(m+1)}{i}{\lhd},\xi^{\omega(\gamma)},q') \in \delta$
	where $q' = (s,\tilde{s}, \tau_q, 1)$ and $i$ is equal to $\argmax_{0 \leq \ell < \length{\tilde{s}}} \{ t \mid \hat{\newtau}_q(t) = \ell \}$ if $b=0$, and is $m+1$ if $b=1$.
	\item \label{trans_q1_not_1top}
	\label{trans_q1_1top}
	\emph{Extremal transitions}: if we are at a new 1-top
        position, we read it with an available pebble, and if the current position
        already has a pebble, a silent transition moves on. The automaton will
        now be on either the next prefix, or the next extremal state if $s =
        \tilde{s}$, that is, the whole $s$ has been read. 
	Let $\newtau_{q'}$ be $\newtau_q[k\mapsto \length{\tilde{s}}]$ if $\tilde{s}(\length{\tilde{s}})\in \Gammatop$, and be 
	$\newtau_q$ if $\tilde{s}(\length{\tilde{s}})\notin \Gammatop$.
	Let $q'$ be $(s,\tilde{s} s(\length{\tilde{s}}+1),\tau_{q},0)$ if $\tilde{s} \neq s$, and be $(s,\tau_{q})$ if $\tilde{s} = s$.
	Let $i = \argmax_{0 \leq \ell < \length{\tilde{s}}} \{ t \mid \hat{\newtau}_q(t) = \ell \}$ and
	$j = \argmin_{\length{\tilde{s}} \leq \ell \leq \length{s}+1} \{ t \mid \hat{\newtau}_q(t) = \ell \}$.
	We have $(q,q') \in \delta$ if $\tilde{s}(\length{{\tilde{s}}}) \notin \Gammatop$, and 
	$(q,\Move{k}{i}{j},\xi^{\omega(\tilde{s}(\length{{\tilde{s}}}))},q') \in \delta$ if 
	$\tilde{s}(\length{{\tilde{s}}}) \in \Gammatop$ and the pebble $k$ is available in $q$.
\end{enumerate}
\textbf{Transitions from non-prefix states.}
If $(s_0,s)$ are consecutive, we start reading $s$ by moving to ${q'} =
(s,s(1),\tau_{{q'}},0) \in Q_p$, 
where $\tau_{{q'}}$ stores the pebble assignment induced by 
$\notOneTopEmbeddingNoArg{s}{s_0}$. 
For every consecutive pair $(s_0,s)$ of extremal strings and states
$q = (s_0,\tau_{q}) \in Q_e$ and 
${q'} = (s,s(1),\tau_{{q'}},0) \in Q_p$, 
we have $(q,{q'}) \in \delta$ if $\notOneTopEmbeddingNoArg{s}{s_0}$ induces $\tau_{q'}$ as follows:
for every pebble $k \in [m]$, 
$\newtau_{{q'}}(k)$ is $(\notOneTopEmbeddingNoArg{s}{s_0})^{-1}(\newtau_{q}(k))$ if $\newtau_{q}(k) \mbox{ is in the image of }\notOneTopEmbeddingNoArg{s}{s_0}$, and is $\bot$
 otherwise.

\subsection{Proof of \texorpdfstring{$L(\varphi) \subseteq L(\pA^\varphi)$}{L(varphi) subseteq L(A)}}
\label{proof_lem_defined_accepted}

\begin{lemma}
\label{lem_defined_accepted}
Let $\cD$ be a data word with $\project(\cD) = u$ 
such that $\cD \models \varphi$. 
Then $u \in L(\pA^\varphi)$.
\end{lemma}

We recall the following: 
\begin{obs}
\label{prop_available_pebble}
For every $q \in Q_p \cup Q_e$, there is a pebble $k \in [m]$ such that both pebble $k$ and pebble $m+1$ are available in $q$. 
\end{obs}

\paragraph*{Coherent configurations}
\label{proof_prop_one_extremal}

Let $\cD$ be a data word with $\project(\cD) = u$ and $|u| = n$. 
Let $\cT$ be a perfect $\cD$-task word, $w=\dataAbs(\cT)$, and $s = \ext(w)$
with $|s|=n_s$. 
Let 
$\rho$ be an $(m+1,n)$-pebble assignment, 
$\tau$ an $(m+1,n_s)$-pebble assignment, 
$q = (s,\tilde{s},\tau,b)$ or $q=(s,\tau)$ be a state and
$\pi = (q,\rho,N)$ a configuration on $u$.

$\pi$ is $w$-\emph{coherent} if
for every $k \in [m]$ such that  $\tau(k) \neq \bot$, 
$\rho(k) = \embeddingOp{\ext}{w}(\tau(k))$.
\begin{lemma}
\label{lem_one_extremal}
Let $\cD$ be a data word with $\project(\cD)=u$, let $\cT$ be a perfect $\cD$-task word, $w = \dataAbs(\cT)$, $s = \ext(w)$, and 
let $\pi = ((s,s(1),\tau,0),\rho,\positions_{<\onetoprm}(w))$ be a $w$-coherent configuration on $u$. 
There
is a $w$-coherent configuration $\pi' = ((s,\tau'),\rho',[n_u])$ on $u$ such that $\pi \rightsquigarrow^\star_u \pi'$ .
\end{lemma}

\begin{proof}
Let $C \in \mathbb{N}$ and $1 \leq \ell_1 < \cdots < \ell_C \leq n_u$ such that 
\[
\{\ell_1, \ldots, \ell_C\} = \extPos(w) \cup \positions_{\onetoprm}(w).
\]
For every $c \in \{0\}\cup[C]$, let
$N_{c} = \positions_{<\onetoprm} \cup \{\ell_1,\ldots,\ell_{c}\}$.
We have $N_{c+1} = N_{c} \cup \{\ell_{c}\}$ for $c \not=0$. 
Let $a_0 = 0$. For every $c \in \{0\}\cup [C-1]$, let:
\[	
a_{c+1} = 
	\begin{cases}
	a_{c}+1, &\ell_{c+1} \in \extPos(w) \\
	a_{c}, & \ell_{c+1} \notin \extPos(w)
	\end{cases}
\]
Let $\ell_{\max} = \max(\positions_{\onetoprm}(w))$. 
Observe that $\ell_C \in \extPos(w)$, since it holds that $\ell_{\max} \in \extPos_{\onetoprm,\theta}(w)$
for every $\theta\in \omega(w(\ell_{\max}))$. 
Consequently, $s(1)\cdots s(a_c) = s$ if and only if $c=C$. 
Similarly, $\ell_1 \in \extPos(w)$. 
	
We give a construction of a computation $(\bar{t},\bar{\pi})$ with 
transitions 
$\bar{t} = (t_1, \ldots, t_{C})$
and $w$-coherent configurations
$\bar{\pi} = (\pi_0, \ldots, \pi_{C})$ on $u$
such that 
$\pi_0 = \pi$, and 
$\pi_c \rightsquigarrow^{t_{c+1}}_u \pi_{c+1}$ 
for all $c \in \{0\}\cup[C-1]$.
Let $\pi_c=(q_c, \rho_c, N_c)$ where 
$
q_c = (s,s(1) \cdots s(a_c + 1), \tau_c, b_c)
$ 
for $c < C$ and 
$q_C = (s, \tau_C)$. 
We construct $t_{c+1}$ and $\pi_{c+1}$ inductively for $c\in \{0\}\cup[C-1]$ by dividing into cases as follows.

Assume $\ell_{c+1} \in \extPos(w)$. 
	We have $\ell_{c+1} = \embeddingOp{\ext}{w}(a_{c+1})$ and hence $w(\ell_{c+1}) = s(a_{c+1})$. 
	If $c+1 < C$, let $b_{c+1} = 0$. 
	\begin{enumerate} 
		\item 
		Assume $s(a_{c+1}) \notin \Gammatop$.
		Let $t_{c+1} = (q_c,q_{c+1})$, $\rho_{c+1} = \rho_c$, and $\tau_{c+1} = \tau_c$.
		Then $t_{c+1} \in \delta$,  $\pi_c \rightsquigarrow^{t_{c+1}}_u \pi_{c+1}$, and $\pi_{c+1}$ is a $w$-coherent configuration on $u$. 
		
		\item
		Assume $s(a_{c+1}) \in \Gammatop$. 
		Let $\tau_{c+1} = \tau_c[k \mapsto a_{c+1}]$. 
		By Observation~\ref{prop_available_pebble}, 
		there is some available pebble $k \in [m]$ in $q_c$. 
		There exist $i,j \in [m] \cup \{\rhd,\lhd\}$ and 
		\[
		t_{c+1}=(q_c,\Move{k}{i}{j},\xi^{\omega(s(a_{c+1}))},q_{c+1})
		\]
		such that $t_{c+1} \in \delta$. 
		By the choice of $i$ and $j$ in the definition of 
		an extremal transition from a prefix state in \S~\ref{subse:def_automaton},
		$\hat{\tau_{c}}(i) < a_{c+1} \leq \hat{\tau_{c}}(j)$. 
		Since $\pi_c$ is $w$-coherent and $\embeddingOp{\ext}{w}$ is order-preserving, 
		$\hat{\rho_{c}}(i) < \ell_{c+1} \leq \hat{\rho_{c}}(j)$. 
		We have $\ell_{c+1} \notin N_{c}$, hence
		no pebble has been placed on $\ell_{c+1}$.
		In particular we have $\ell_{c+1} \not= \rho_c(j)$,
		implying that $\hat{\rho_{c}}(i) < \ell_{c+1} < \hat{\rho_{c}}(j)$. 
		Let $\rho_{c+1} = \rho_c[k \mapsto a_{c+1}]$.
		Since $s(a_{c+1}) = w(\ell_{c+1})$, $\xi^{\omega(s(a_{c+1}))} = \xi^{\omega(\ell_{c+1})} = u(\ell_{c+1})$. 
		Then $\pi_c \rightsquigarrow^{t_{c+1}}_u \pi_{c+1}$  and $\pi_{c+1}$ is $w$-coherent. 
		
	\end{enumerate}
	
	Now assume $\ell_{c+1} \notin \extPos(w)$. Then $\ell_{c+1} \in \positions_{\onetoprm}(w)$,
	$\ell_{c+1} \notin N_c$, 
	$a_{c+1} = a_{c}$, and
	for every $\theta\in\Theta_{\exists}$ there are
	$
	1\leq c_{\theta,l}<c<c_{\theta,r} \leq C
	$
	such that
	\[
	\begin{array}{llllll}
	c_{\theta,l} &=&  \max(\{\tilde{c} \in [C] \mid \ell_{\tilde{c}} \in \extPos_{\onetoprm,\theta}(w))\} \cap[c-1]),\\
	c_{\theta,r} &=&  \min\,(\{\tilde{c}  \in [C] \mid \ell_{\tilde{c}} \in \extPos_{\onetoprm,\theta}(w))\} \\ && \hphantom{\min\,(}\cap \{c+1,\ldots,C\}). 
	 \end{array}
	\]
	Let $s_{\mathit{pr}}$ and $s_{\mathit{su}}$ be respectively the
	prefix and suffix of $s$ given by $s_{\mathit{pr}} = s(1)\cdots s(a_{c+1})$
	and $s_{\mathit{su}} = s(a_{c+1}+1)\cdots s(n_s)$. 
	We have $s = s_{\mathit{pr}} s_{\mathit{su}}$ and 
	$\ext(s)=\ext(s_{\mathit{pr}} w(\ell_{c+1}) s_{\mathit{su}})$, and hence
	$w(\ell_{c+1}) \in \nonext{s}{a_{c}+1}$.
	Let $\tau_{c+1} = \tau_{c}$ and $b_{c+1} = 1$. 
	There exists $i$ such that
	\[
	(q_c,\Move{(m+1)}{i}{\lhd},\xi^{\omega(w(\ell_{c+1}))},q_{c+1}) \in \delta.
	\]

	If $b_c = 0$, then by the choice of $i$ in the definition of a non-extremal transition 
	from a prefix state in \S~\ref{subse:def_automaton},
	$\hat{\tau}_c(i) \leq  a_c$. 
	Since $\embeddingOp{\ext}{w}$ is order-preserving, 
	$\hat{\rho}_c(i) \leq \ell_{c}$ and hence $\hat{\rho}_c(i) < \ell_{c+1}$.	
	If $b_c = 1$, then the computation $((t_1,\ldots,t_c),(\pi_0,\ldots,\pi_c))$
	has at least one non-extremal transition. Hence the pebble $m+1$ was moved during this computation,
	i.e.~$\rho_c(m+1) \in \{\ell_1,\ldots,\ell_{c}\}$.
	In either case, $\hat{\rho}_c(i) < \ell_{c+1} < \hat{\rho}_c(\lhd)=n_u+1$.
	Let $\rho_{c+1}=\rho_c$, and note that 
	$\xi^{\omega(\ell_{c+1})} = u(\ell_{c+1})$. 
	Then $\pi_c \rightsquigarrow^{t_{c+1}}_u \pi_{c+1}$,  
	and since $\tau_{c+1}=\tau_c$ and $\rho_{c+1}$ agrees with
	$\rho_c$ on all $k \in [m]$, we have that $\pi_{c+1}$ is $w$-coherent.

The lemma follows with 
$t' = t_C$ and
$\pi' = \pi_{C}$.
\end{proof}

\paragraph*{Induced sequences of extremal strings}
\label{proof_dataword_series}

Let $\cD$ be a data word with $h = \maxdv_\cD$, and $\cT$ a $\cD$-task word. 
For every $e \in \{0\} \cup [h]$, let 
$s_{e} = \ext(\dataAbs(\ttrim{\cT}{h-e}))$.
We say the sequence $s_0, \ldots, s_h$ is \emph{induced by~$\cT$}.
By Prop.~\ref{prop_cor_perfect_completed_tasked_word}:
\begin{lemma}
\label{lem_dataword_series}
Let $\cD$ be a data word with $h = \maxdv_\cD $ such that $\cD \models \varphi$.
There exists a perfect completed $\cD$-task word $\cT$ such that, if $s_0, \ldots, s_h$
is induced by $\cT$, 
then $s_0 = \varepsilon$, $s_h$ is completed, and
$(s_{e-1},s_e)$ is consecutive for $e \in [h]$. 
\end{lemma}

\begin{proof}
Since $\cD \models \varphi$, by Proposition~\ref{prop_cor_perfect_completed_tasked_word}, there exists a perfect completed $\cD$-task word $\cT$.
Let $s_0, \ldots, s_h$ be the sequence of extremal strings  induced by $\cT$. 
We have $s_0 = \ext(\dataAbs(\ttrim{\cT}{h})) = \varepsilon$.
Since $\cT$ is perfect, all the extremal strings in the sequence are perfect. Since $\cT$ is completed,  $s_h = \ext(\dataAbs(\cT))$
is completed.
Since for $e\in[h]$ we have $\ttrim{\cT}{h-(e-1)} = (\ttrim{\cT}{h-e})\trim{\,\!}{1}$, every pair $(s_{e-1},s_e)$ is consecutive. 
\end{proof}

\paragraph*{Proof of Lemma~\ref{lem_defined_accepted}}
Let $R = \maxdv_\cD$. 
By Lemma~\ref{lem_dataword_series}, 
there exists a perfect completed $\cD$-task word $\cT$ such that the sequence  $s_0, \ldots, s_R$
of extremal strings induced by $\cT$ satisfies
that $s_0 = \varepsilon$, $s_R$ is completed, and
$(s_{r-1},s_r)$ is consecutive for $r \in [R]$. 
For every $r \in  \{0\} \cup [R]$,
let $u_r$ be the string projection of $\trim{\cD}{R-r}$
and let
$w_r = \dataAbs(\ttrim{\cT}{R-r})$.
We have $s_r = \ext(w_r)$. Since $\cT$ is perfect, 
so are $\ttrim{\cT}{R-r}$ and 
$s_r$. 

Let $\pi_0=\pi_{\mathrm{init}}$. Observe that for every $r\in \{0\}\cup [R]$, $\pi_0$ is a $w_r$-coherent configuration on $u_r$. 
We construct
a sequence of transitions $(t_1,\ldots,t_R)$
and a sequence of configurations $(\pi_1,\ldots,\pi_R)$
such that,
for every $r\in [R]$,
$\pi_r = ((s_r,\tau_r), \rho_r, [n_{r}])$
is a $w_r$-coherent configuration on $u_r$
and 
$\pi_0 \rightsquigarrow^{\star}_{u_{r}} \pi_r$.
We construct the transitions and configurations inductively as follows. 
For every $r \in [R]$, assume there are $t_{r}$ and $\pi_{r}$ as described above. 

The universe of $\ttrim{\cT}{R-r}$ is a subset of the universe of $\ttrim{\cT}{R-(r+1)}$.
We use the notation
$\embeddingNoArg{w_r}{w_{r+1}}$ for the embedding 
obtained by the composition of 
$\embeddingOp{\dataAbs}{\ttrim{\cT}{R-r}}$
and 
$\embeddingOp{\dataAbs}{\ttrim{\cT}{R-(r+1)}}^{-1}$.
The string $u_r$ is a substring of $u_{r+1}$.
We denote by $\pos_{r,r+1}$ the mapping of positions of $u_r$ to positions of $u_{r+1}$.
Since the universe of
$\ttrim{\cT}{R-e}$ is equal to the universe of $\trim{\cD}{R-e}$ for $e\in \{r,r+1\}$,
$\embeddingNoArg{w_r}{w_{r+1}} = \pos_{r,r+1}$.

Let $\tilde{\pi}_r = ((s_r,\tau_r), \tilde{\rho}_{r}, \positions_{<\onetoprm}(w_{r+1}))$
be a configuration on $u_{r+1}$ with $\tilde{\rho}_{r}(k) = \bot$ if $\rho_{r}(k) = \bot$,
and $\tilde{\rho}_{r}(k)  = \embedding{w_r}{w_{r+1}}{\rho_{r}(k)}$ if $\rho_r(k) \not= \bot$.
The semantics of PIA allow us to lift a computation from a substring to a string, thus 
$\pi_0 \rightsquigarrow^{\star}_{u_{r+1}} \tilde{\pi}_r$.

Let
$\pi_{r+1}^\starttext$
be the configuration on $u_{r+1}$ given by
\[
 \begin{array}{lll}
\pi_{r+1}^\starttext &=& ((s_{r+1}, s_{r+1}(1),\tau_{r+1}^\starttext,0), \tilde{\rho}_{r}, \positions_{<\onetoprm}(w_{r+1}))\\
 \end{array}
\]
where for every pebble $k \in [m]$, 
$\tau_{r+1}^\starttext(k)$ is $\bot$ if $\tau_{r}(k)$ is not in the image of $\notOneTopEmbeddingNoArg{s_{r+1}}{s_r}$,
and is $(\notOneTopEmbeddingNoArg{s_{r+1}}{s_r})^{-1}(\tau_{r}(k))$ otherwise.
Let 
\[
t_{r+1}=((s_r,\tau_r),(s_{r+1}, s_{r+1}(1),\tau_{r+1}^\starttext,0)).
\]
Since
the pair $(s_{r},s_{r+1})$ is consecutive,
we get that $t_{r+1}\in \delta$ and 
$\tilde{\pi_{r}} \rightsquigarrow^{t_{r+1}}_{u_{r+1}} \pi_{r+1}^\starttext$,
implying that
$\pi_{0} \rightsquigarrow^{\star}_{u_{r+1}} \pi_{r+1}^\starttext$. 
We prove that $\pi_{r+1}^\starttext(k)$ is $w_{r+1}$-coherent.
$\pi_{r}$ is $w_r$-coherent by the assumption, hence
$\rho_{r}(k) = \embeddingOp{\ext}{w_{r}}(\tau_{r}(k))$.

Let $k \in [m]$ be such that $\tau_{r+1}^\starttext(k) \not= \bot$. 
Then 
$\tau_{r}(k)$ is in the image of $\notOneTopEmbeddingNoArg{s_{r+1}}{s_r}$
and in particular $\tau_{r}(k)\not=\bot$, and $\tau_{r+1}^\starttext(k)$ is given by:
\[
\begin{array}{llll}
\embeddingOp{\ext \, \circ \, \dataAbs}{\ttrim{\cT}{R-(r+1)}}^{-1}
\left(
\embeddingOp{\ext \, \circ \, \dataAbs}{\ttrim{\cT}{R-r}}
(\tau_r(k))
\right)  = & \\
\embeddingOp{\ext}{w_{r+1}}^{-1}
\left(
\embeddingNoArg{w_{r}}{w_{r+1}}
\left(
\rho_r(k)
\right)
\right).
\end{array}
\]
We have
\[
\begin{array}{lll}
\embeddingOp{\ext}{w_{r+1}}(\tau_{r+1}^\starttext(k)) = \embeddingNoArg{w_{r}}{w_{r+1}}
\left(
\rho_r(k)
\right)
 =  \tilde{\rho}_{r}(k)
	
\end{array}
\]

Hence, $\pi_{r+1}^\starttext(k)$ is $w_{r+1}$-coherent. 

Now we apply Lemma~\ref{lem_one_extremal} with the data word 
$\trim{\cD}{R-(r+1)}$ of size $n_{r+1}$, 
the string projection $u_{r+1}$, the perfect $\trim{\cD}{R-(r+1)}$-task word $\ttrim{\cT}{R-(r+1)}$,
the abstraction $w_{r+1}$, the extremal string $s_{r+1}$, and the configuration $\pi_{r+1}^\starttext(k)$;
we get 
that there are $\tau_{r+1}$ and $\rho_{r+1}$  
such that $\pi_{r+1}$
is a $w_{r+1}$-coherent configuration on $u_{r+1}$
and
$\pi_{r+1}^\starttext(k) \rightsquigarrow^{\star}_{u_{r+1}} \pi_{r+1}$, and
therefore
$\pi_0 \rightsquigarrow^{\star}_{u_{r+1}} \pi_{r+1}$.
The lemma follows for the computation 
$\pi_0 \rightsquigarrow^{\star}_{u_{R}} \pi_{R}$
since  $u=u_{r+1}$, $n_R = n$, and $(s_R,\tau_R) \in F$.

\section{Proof of \texorpdfstring{$L(\pA^\varphi) \subseteq L(\varphi)$}{L(A) subseteq L(varphi)}}

\begin{proof}
Given $w \in L(\pA^\varphi)$, we build a data word $\cD$ for it
based on an accepting computation of $\pA^\varphi$ on $w$.
To show that $\cD \models \varphi$, we prove the existence of a perfect completed $\cD$-task word based on the syntactic representation of consecutive extremal strings.

Let $(\bar{t},\bar{\pi})$ be an accepting computation.
Let $z$ be the number of transitions from a state in $\{q_{\mathrm{init}}\} \cup Q_e$ to a state in $Q_p$ in $\bar{t}$.
The computation can be broken down into parts as follows:
\[
\pi_0 
\rightsquigarrow_{w}^{t_0} \pi_1^{(1)} \rightsquigarrow_{w}^{\bar{t}_1} \pi_{1}^{(2)} \rightsquigarrow_{w}^{\bar{t}_2} 
\cdots 
\rightsquigarrow_{w}^{\bar{t}_{h-1}} \pi_{1}^{(z)} \rightsquigarrow_{w}^{\bar{t}_z} \pi_f
\]
where 
the target state of any transition is in $Q_e$ if and only if there is $e \in [z]$ such that the transition is the last one in $\bar{t}_e$. 
The sequence $\bar{t}$ is equal to the concatenation of $t_0$ and the sequences $\bar{t}_1, \ldots, \bar{t}_z$. 
For every $e$ and $a$, 
let the state of $\pi^{(e)}_a$ be $(s_e, s_e(1)\cdots s_e(a_e), \tau_{e,{a_e}}, b_e)$. 

For every transition $t = t^{(e)}_a$, let $\gamma_t \in \Gammatop \cup \{ \varepsilon\}$ be:
\begin{enumerate} 
 \item if $t\in Q\times Q$, $\gamma_t = \varepsilon$.
 \item Otherwise, 
 if $t$ is a non-extremal transition, let $\gamma_t \in \nonext{s_e}{a_e}$
 such that $t \in Q \times \MoveS_{m+1} \times \{\xi^{\gamma_t}\} \times Q$. 
 If $t$ is an extremal transition, 
 let $\gamma_t = s(a_e)$ (and note that we again have that  $t \in Q \times \MoveS_{m+1} \times \{\xi^{\gamma_t}\} \times Q$). 
\end{enumerate}

Let $\gamma_{\bar{t}_{e}} = \gamma_{t_1^{(e)}} \cdots \gamma_{t_{\mathit{len}(e)}^{(e)}}$, where $\mathit{len}(e)= |\bar{t}_{e}|$. 
For every $e \in [z]$,
let $v_e$, $u_e$ be the substrings of $w$ which the automaton reads during the transitions $\bar{t}_e$ respectively 
$(t_0,\bar{t}_1,\ldots,\bar{t}_e)$. 
Let $g_e'$ be such that $u_e$ is the shuffle $v_e \shuffle_{g_e'} u_{e-1}$ relative to the positions in $w$. 

We prove the following by induction on $e$: 
there is a data word $\cD_e$ and a $\cD_e$-task word $\cT_e$ such that for every $p \leq e$ we have:
	(i) 
	$\project(\trim{\cD_e}{p}) = u_{e-p}$,
	(ii)
	$\ext(\dataAbs(\ttrim{\cT_e}{p})) = s_{e-p}$,
	(iii)
	$\ttrim{\cT_e}{p}=\cT_{e-p}$ and $\ttrim{\cT_e}{p}$ is a $\trim{\cD_e}{p}$-task word,
	(iv)
	the universe of $\cD_e$ is contained in $[e] \times \mathbb{N}$.

Since for every $e \in [z]$, $s_e$ is an extremal string appearing in a state in $Q_p$, $s_e$ is perfect. 
Therefore $\cT_z$ is a perfect $\cD$-task word.
Since the computation is accepting, the extremal string $s_z$ is complete and therefore $\cT_z$ is a completed $\cD$-task word.
By Prop.~\ref{prop_cor_perfect_completed_tasked_word}, we have $\cD \models \varphi$.
\end{proof}

We now prove inductively the following Lemma: 
\begin{lemma}
\label{lem_accepted_satisfies}
Let $w \in \Xi^\star$ be accepted by $\pA^\varphi$. Then  $w \in L(\varphi)$.
\end{lemma} 

We assume the induction hypothesis for $e-1$ and prove for $e$.
Let $r_e \in \Gammatop^\star \cap \PGamma^\star$ be obtained from $\gamma_{\bar{t}_e}$ by setting all tasks to $P$. 
Notice $v_e = \xi^{\gamma_{\bar{t}_{e}}} = \xi^{\gamma_{r_{e}}}$. 
Let $g_e:[n_{r_e}] \to [n_{r_e}+n_{s_{e-1}}]$ be given by $g_e(\ell) = \ell + | \extPos(\dataAbs(\cT_{e-1})) \cap [g_e'(\ell) - \ell] |$.
Recall that $\bar{g}_e: [n_{s_{e-1}}] \rightarrow [n_{r_e}+n_{s_{e-1}}]$.
Note:
\[
\begin{array}{llll}
\forall \ell \in [n_{r_e}], & \xi^{(\rcon{r_e}{g_e}{s_{e-1}})(g_e(\ell))} = (\rcon{v_e}{g_e'}{u_{e-1}})(g_e'(\ell)) \\
\forall \ell \in [n_{s_{e-1}}],& \xi^{(\rcon{r_e}{g_e}{s_{e-1}})(\bar{g}_e(\ell))} = (\rcon{v_e}{g_e'}{u_{e-1}})(\bar{g}_e'(\ell)) 
\end{array}
\]

\begin{claim}
\label{lem:extend_D0}
There is a data word $\cD_e$ such that $\project(\cD_e)= u_e $
and the universe of $\cD_e$ is contained in $[e] \times \mathbb{N}$, 
and there exists a $\cD_e$-task word $\cT_e$ such that 
$\ext(\dataAbs(\cT_e))= \ext(\rcon{r_e}{g_e}{s_{e-1}})$ and for every $p \leq e$, $\ttrim{\cT_e}{p} = \cT_{e-p}$. 
\end{claim}
\begin{proof}
Let $D_{e-1}$ be the universe of $\cD_{e-1}$. 
Let $\cD_e$ be the data word over $\Xi$ with universe $D_{e-1} \cup (\{e\} \times [n_{r_e}])$
such that:
\begin{enumerate}
 \item $\cD_{e-1}$ is the substructure of $\cD_e$ induced by $D_{e-1}$.
 
 \item For every $\ell \in [n_{r_e}]$ and $r_e(\ell)=(\hlayer,\ts)$, $\cD_e \models \xi^{\omega(\ts)}(e,\ell)$. 
 
 \item For every $\ell_1,\ell_2 \in [n_{r_e}]$, $\cD_e \models (e,\ell_1) \sim_2 (e,\ell_2)$.
 
 \item For every $d\in D_{e-1}$ and $\ell \in [n_{r_e}]$, $\cD_e \models d <_2  (e,\ell)$.
 
 \item For every $\ell_1, \ell_2 \in [n_{r_e}]$, $\cD_e \models (e,\ell_1) \leq_1 (e,\ell_2)$ \\ if and only if $\ell_1 \leq \ell_2$.
 
 \item For every $d\in D_{e-1}$ and $\ell \in [n_{r_e}]$, 
 $\cD \models d \leq_1 (e,\ell)$ if and only if
 $g_e'(\ell) \geq |\{ d' \in D_{e-1} \mid \cD_{e-1} \models d' \leq_1 d \}| + \ell$.
 
\end{enumerate}

Let $\cT_e$ be the $\cD_e$-task word such that
\begin{itemize}
	\item 
	for every $d \in D_{e-1}$,
	there are $\ts,\ts' \in \compatTasks$ such that $\omega(\ts)=\omega(\ts')$,
	$\cT_e \models \ts(d)$, and $\cT_{e-1} \models \ts'(d)$, and
	\item
	for every $(e,\ell) \in D_e$, there are $\ts,\ts' \in \compatTasks$ such that $\omega(\ts)=\omega(\ts')$,
	$r_e(\ell) = (\otop, \ts')$, and $\cT_{e} \models \ts(e,\ell)$.
\end{itemize} 
By our construction, $\ttrim{\cT_e}{p} = \cT_{e-p}$ for $p=1$. For $p>1$, this equality follows from the induction hypothesis.

Clearly $\ext(\dataAbs(\ttrim{\cT_e}{1})) = \ext(\dataAbs(\cT_{e-1})) = s_{e-1}$, by the induction hypothesis. 
We apply Lemma~\ref{lem:consec_to_syntactic} with $\cD_e$ for $\cD$, $\cT_e$ for $\cT$ and $r_e$ for $r$.
Note that $g$ in the lemma is $g_e$, hence we get $\ext(\dataAbs(\cT_e)) = \ext(\rcon{r_e}{g_e}{s_{e-1}})$.

Let $\cD_{\max}$ be the substructure of $\cD_e$ which consists of the elements of $\cD_e$ with maximal data value.
By the definition of $\cD_e$, we have $\project(\cD_{\max}) = \xi^{r_e} = v_e$.
By induction, $\project(\cD_{e-1}) = u_{e-1}$.
By the definition $g_e'$, $u_e = v_e \shuffle_{g_e'} u_{e-1}$, and hence $\project(\cD_e) = u_e$.

\end{proof}

From Lemma~\ref{cl_extPos_cardinality}, it follows that there exists 
$
f_e: [n_{r_e}] \to [n_{r_e}+n_{s_{e-1}}]
$ 
such that 
$s_e$ is a substring of $\rcon{{r_e}}{{{f_e}}}{{s_{e-1}}}$, hence 
$s_e=\ext(\rcon{{r_e}}{{{f_e}}}{{s_{e-1}}})$. 

Let $r_e' = \ext(r_e)$ and $s_{e-1}' = \ext(\da{s_{e-1}})$.
There exist 
$
G_e, F_e:[n_{r_e'}] \to [n_{r_e'} + n_{s_{e-1}'}]
$ 
such that 
$\ext(\rcon{{r_e}}{f_e}{s_{e-1}}) = r_e' \shuffle_{F_e} s_{e-1}'$
and
$\ext(\rcon{{r_e}}{g_e}{s_{e-1}}) = r_e' \shuffle_{G_e} s_{e-1}'$.
Assume that $F_e  \not= G_e$. 
Recall that $\bar{F}_e,\bar{G}_e: [n_{s_{e-1}'}] \rightarrow [n_{r_e'}+n_{s_{e-1}'}]$.
Let $\tilde{\ell} \in [n_{r_{e}'} + n_{s_{e-1}'}]$ be the length of the maximal common prefix of 
$r_e' \shuffle_{F_e} s_{e-1}'$ and $r_e' \shuffle_{G_e} s_{e-1}'$.
We divide into two cases:
\begin{itemize}
\item 
	Assume the letter at position $\tilde{\ell}+1$ in 
	$r_e' \shuffle_{G_e} s_{e-1}'$ is in $\Gammatop$. 
	Then the letter at position $\tilde{\ell}+1$ in
  $r_e' \shuffle_{F_e} s_{e-1}'$ is not in $\Gammatop$.
	Let $\ell_1 \in [n_{r_e'}]$ and $\ell_2 \in [n_{s_{e-1}'}]$ be such that $G_e(\ell_1) = \bar{F}_e(\ell_2) = \tilde{\ell}+1$.
	Then $F_e(\ell_1) > \bar{F}_e(\ell_2)$ and $\bar{G}_e(\ell_2) > G_e(\ell_1)$.
	\item 
	Assume the letter at position $\tilde{\ell}+1$ in 
	$r_e' \shuffle_{F_e} s_{e-1}'$ is in $\Gammatop$. 
	Analogously to the previous case, we have $G_e(\ell_1) > \bar{G}_e(\ell_2)$ and $\bar{F}_e(\ell_2) > F_e(\ell_1)$.
	
\end{itemize}
In both cases,
$G_e(\ell_1) < \bar{G}_e(\ell_2)$
if and only if
$\bar{F}_e(\ell_2) < F_e(\ell_1)$. 

Let $t^{(e)}_c$ be the transition in which the automaton reads position $\embeddingOp{\ext}{r_e}(\ell_1)$ of $v_e$. 
Let $\ell_1^w \in [n_w]$ be the position of $w$ which $t^{(e)}_c$ reads. 
Let $\Move{k}{i}{j}$ be the move in $t^{(e)}_c$. 
There is a pebble 
$k' = \tau_{e,c}^{-1}(\bar{F}_e(\ell_2))$
on the position in $s_e$ corresponding to $\ell_2$ in $s_{e-1}'$. 
Let $\ell_2^w \in [n_w]$ be the position of $k'$ in $w$.
$G_e$ and $\bar{G}_e$ have disjoint images, hence $G_e(\ell_1) \neq \bar{G}_e(\ell_2)$, and similarly for $F_e,\bar{F}_e$.
We divide into cases:
\begin{enumerate}
 
 \item If $G_e(\ell_1) < \bar{G}_e(\ell_2)$, 
 then 
 $\bar{F}_e(\ell_2) < F_e(\ell_1)$. 
 By the definition of the automaton, the pebble $i$ is located to the left of $\ell_1^w$, and $i$
  is either $k'$ itself, or another pebble located to the right of $k'$. 
  Hence, $\ell_1^w > \ell_2^w$. 
  But since  $G_e(\ell_1) < \bar{G_e}(\ell_2)$, we have 
  $g_e(\embeddingOp{\ext}{r_e}(\ell_1)) < \bar{g_e}(\embeddingOp{\ext}{\,\da{s_{e-1}}}{\ell_2})$, 
  we have $g_e'(\embeddingOp{\ext}{r_e}(\ell_1)) < \bar{g_e'}(\embeddingOp{\ext}{\,\da{s_{e-1}}}(\ell_2))$, and hence $\ell_1^w < \ell_2^w$, in contradiction.
 
 \item The case of $g_e(\ell_1) > \bar{g_e}(\ell_2)$ is analogous. 
 
 \end{enumerate}

Hence $F_e = G_e$ and $s_e = \ext(\rcon{{r_e}}{{{g_e}}}{{s_{e-1}}})$,
and the induction hypothesis holds.

 To define partial embeddings, we use the following definition and lemma:
 \begin{definition}[$\extElem(\cT)$]
 Denote $\extElem(\cT) = \embeddingOp{\dataAbs}{\cT}(\extPos(\dataAbs(\cT)))$ for a $\cD$-task word $\cT$.
 That is, the elements corresponding to positions in the extremal string of $\cT$.
 For $h \in \Layers$, $\theta \in \Theta_\exists$, denote  $\extElem_{h,\theta}(\cT)$ and $\extElem_\theta(\cT)$ similarly. 
 \end{definition}

 \begin{lemma}
 \label{cl_extPos_cardinality}
 Let $\cT$ be a $\cD$-task word. 
 For every $\theta \in \Theta_\exists$:
 \begin{enumerate}
 	\item 
 	$\extElem_{\mathrm{2top},\theta}(\cT) = 
 	\extElem_{\mathrm{1top},\theta}(\ttrim{\cT}{1})$,
 	\item
 	$\extElem_{\theta}(\cT) 
 	\subseteq 
 	\extElem_{\mathrm{2top},\theta}({\ttrim{\cT}{1}}) 
 	\cup \extElem_{\theta}({\ttrim{\cT}{1}})$, and
 	\item
 	$\extElem_{\mathrm{rest},\theta}(\cT)
 	\subseteq  
 	\extElem_{\mathrm{2top},\theta}({\ttrim{\cT}{1}}) 
 	\cup \extElem_{\mathrm{rest},\theta}({\ttrim{\cT}{1}})$.
 \end{enumerate}	 

 \end{lemma}
\begin{proof}
We prove:
\begin{enumerate}
\item
	$\extElem_{\mathrm{2top},\theta}(\cT) = 
	\extElem_{\mathrm{1top},\theta}(\ttrim{\cT}{1})$ 
\item
	$\extElem_{\theta}(\cT) 
	\subseteq 
	\extElem_{\mathrm{2top},\theta}({\ttrim{\cT}{1}}) 
	\cup \extElem_{\theta}({\ttrim{\cT}{1}})$
\item
	$\extElem_{\mathrm{rest},\theta}(\cT)
	\subseteq$ \\  
	$\extElem_{\mathrm{2top},\theta}({\ttrim{\cT}{1}}) 
	\cup \extElem_{\mathrm{rest},\theta}({\ttrim{\cT}{1}})$
\end{enumerate}

Let $w = \dataAbs(\cT)$ and $w' = \dataAbs(\ttrim{\cT}{1})$.
\begin{enumerate}
	\item
	Let $d \in D$. Since $\val{\cT}{d} = \val{\ttrim{\cT}{1}}{d}+1$, 
	We have $w(\embeddingOp{\dataAbs}{\cT}^{-1})(d) \in \Gammattop$ iff $w'(\embeddingOp{\dataAbs}{\ttrim{\cT}{1}}^{-1})(d)) \in \Gammatop$. 
	Using Lemma~\ref{lemma_trim_omega_compat}, we have for every $\theta \in \Theta_\exists$ that
	\[\embeddingOp{\dataAbs}{\cT}
	(\positions_{\mathrm{2top},\theta}(w)) = 
	\embeddingOp{\dataAbs}{\ttrim{\cT}{1}}
	(\positions_{\mathrm{1top},\theta}(w')).\]
	Since $\embeddingOp{\dataAbs}{\cT}$ and $\embeddingOp{\dataAbs}{\ttrim{\cT}{1}}$ are order-preserving, 
	for both functions $\mathit{opt}=\max$ and $\mathit{opt}=\min$ we have that the positions
	$\ell_{\mathit{opt}} = \mathit{opt}(\positions_{\mathrm{2top},\theta}(w))$, 
	$\ell_{\mathit{opt}}' = \mathit{opt}(\positions_{\mathrm{1top},\theta}(w'))$
	are obtained under $\dataAbs$ from the same $\cT$ element
	\[
	d_{\mathit{opt}} =  \embeddingOp{\dataAbs}{\cT}(\ell_{\mathit{opt}}) = \embeddingOp{\dataAbs}{\ttrim{\cT}{1}}(\ell_{\mathit{opt}}').
	\]

	\item
	Let $\ell \in \extremal_{\theta}(w)$, $d = \embeddingOp{\dataAbs}{\cT}(\ell)$, and let $w(\ell) = (\rest,\ts_{d})$.
	We have that $P_\theta \in \ts_d$. 
	Let $\ell'$ be such that $d = \embeddingOp{\dataAbs}{\ttrim{\cT}{1}}(\ell')$, and $w'(\ell') = (\hlayer,\ts_{d}')$ with $h\in\{\twotoprm,\restrm\}$. 
	By Lemma~\ref{lemma_trim_omega_compat}, $\theta \in \omega(\ts_{d}')$. 
	Let $D'$ be the universe of $\trim{\cD}{1}$. 
	Since $\trim{\cD}{1}$ is a substructure of $\cD$ and $d\in D'$, if $\cD \not\models \exists y \ \theta(d,y)$ then also $\trim{\cD}{1} \not \models  \exists y \ \theta(d,y)$
	and therefore also $P_\theta \in \ts_d'$.
	Assume $\theta(x,y) \models x \leq_1 y$. The case of $\theta(x,y) \models y \leq_1 x$ is analogous. 
	
	Assume for contradiction that $\ell' \notin \extremal_{\mathrm{2top},\theta}(w')$ and $\ell' \notin \extremal_{\theta}(w')$.
	There exist $d_1 \in D'$, $\ell_1' \in [|D'|]$, and $h_1 \in \{\twotoprm,\restrm\}$
	such that $d_1 = \embeddingOp{\dataAbs}{\ttrim{\cT}{1}}(\ell_1')$, $w'(\ell_1')\in \Gamma_{h_1}$,
	$\ell' \lneq \ell_1'$, $\cD \models d \lneq_1 d_1$, and:		
	\begin{enumerate}
 	 \item if $w'(\ell') \in \Gammattop$, then 
 	 $h_1=\twotoprm$, 
 	 $\ell_1' \in \extremal_{\mathrm{2top},\theta}(w')$, and
 	 $\cD \models d \sim_2 d_1$;
 	 \item if $w'(\ell') \in \Gammarest$, then 
 	 $h_1=\restrm$, 
 	 $\ell_1' \in \extremal_{\theta}(w')$ and $w'(\ell_1')\in \Gammarest$.
	\end{enumerate}
	Let $\ell_1$ be such that $d_1 = \embeddingOp{\dataAbs}{\cT}(\ell_1)$.
	We have $\ell \lneq \ell_1$ and $w(\ell_1) \in \Gammarest$.
	Let $w(\ell_1) = (\rest,\ts_{d_1})$ and $w'(\ell_1') = (\hlayerp{1},\ts_{d_1}')$. 
	Since $\ell_1' \in \extremal_{\mathrm{2top},\theta}(w') \cup \extremal_{\theta}(w')$,
	we have $\theta\in\omega(\ts_{d_1}')$. 
	By Lemma~\ref{lemma_trim_omega_compat} we have $\theta \in \omega(\ts_{d_1})$, 
	and using that $\ell \in \extremal_{\theta}(w)$ we have $C_\theta \in \ts_{d_1}$
	and hence $\cD \models \exists y \ \theta(d_1,y)$. 
	Since $\theta \in \omega(\ts_d)\cap\omega(\ts_{d_1})$, we have $\xi^{\omega(\ts_d)} = \xi^{\omega(\ts_{d_1})}$, implying
	that $d$ and $d_1$ have the same 1-type in $\cD$. 
	
	\begin{enumerate}
		\item 
		If $w'(\ell') \in \Gammattop$, 
		by Lemma~\ref{lemma_twotop_type} we have $\cD \models \exists y \ \theta(d,y)$, 
	        in contradiction to $P_\theta \in \ts_{d}$. 
		\item
		If $w'(\ell') \in \Gammarest$, 
		since $\ell_1' \in \extremal_{\theta}(w')$ we have $P_\theta \in \ts_{d_1}'$ 
		and hence $\trim{\cD}{1} \not \models  \exists y \ \theta(d_1,y)$. 
		By Lemma~\ref{lemma_rest_theta} we have $\cD \models \exists y \ \theta(d,y)$, 
		in contradiction to $P_\theta \in \ts_{d}$. 
	\end{enumerate}
	\item
	For every $d \in D$, $\val{\cD}{d} = \val{\trim{\cD}{1}}{d} - 1$. 
	Hence, using 
	Lemma~\ref{lemma_trim_omega_compat}, we have that for every $\theta \in \Theta_\exists$, 
	the embedding $\embeddingOp{\dataAbs}{\cT}(\positions_{\mathrm{rest}, \theta}(w)))$ is given by
	\[
	\embeddingOp{\dataAbs}{\ttrim{\cT}{1}}(\positions_{\twotoprm, \theta}(w') \cup
	 \positions_{\restrm, \theta}(w')).
	\]
	Since $\embeddingOp{\dataAbs}{\cT}$ and $\embeddingOp{\dataAbs}{\ttrim{\cT}{1}}$ are order-preserving, 
	for both functions $\mathit{opt}=\max$ and $\mathit{opt}=\min$ we have that the positions
	$\ell_{\mathit{opt}}' = \mathit{opt}(\positions_{\twotoprm, \theta}(w') \cup \positions_{\restrm, \theta}(w'))$
	and 
	$\ell_{\mathit{opt}} = \mathit{opt}(\positions_{\mathrm{rest}, \theta}(w)))$
	are obtained under $\dataAbs$ from the same $\cT$ element
	$d_{\mathit{opt}} =  \embeddingOp{\dataAbs}{\cT}(\ell_{\mathit{opt}}) = \embeddingOp{\dataAbs}{\ttrim{\cT}{1}}(\ell_{\mathit{opt}}')$.

	\end{enumerate}

\end{proof}

 Let $(s',s)$ be a pair of consecutive extremal strings, and 
 let  $\cT$ be a task word such that
 $s = \ext(\cT)$
 and 
 $s' = \ext(\ttrim{\cT}{1})$. 
 We denote by  
 $\notOneTopEmbeddingNoArgT{s}{s'}{\cT}$
 the function 
 from the set 
 $s^{-1}(\Gamma - \Gammatop)$
 of positions $\ell$ 
 for which $s(\ell)\notin \Gammatop$
 to $[n_{s'}]$ 
 defined as
 $\notOneTopEmbeddingNoArgT{s}{s'}{\cT}(\ell) = 
 \embeddingOp{\ext \, \circ \, \dataAbs}{\ttrim{\cT}{1}}^{-1}\left(
 \embeddingOp{\ext \, \circ \, \dataAbs}{\cT}
 (\ell)\right)$.
 The function $\notOneTopEmbeddingNoArgT{s}{s'}{\cT}$ is well-defined: 
 since $s(\ell)$ is not in $\Gammatop$, we get from Lemma~\ref{cl_extPos_cardinality} that
 $\embeddingOp{\ext \, \circ \, \dataAbs}{\cT}
 	(\ell) \in
 	\embeddingOp{\ext \, \circ \, \dataAbs}{\ttrim{\cT}{1}}([n_{s'}])$.
 We call $\notOneTopEmbeddingNoArgT{s}{s'}{\cT}$ a \emph{partial embedding via $\cT$} since
 it is injective and order-preserving.

 \begin{runningexample}
 $\notOneTopEmbeddingNoArgT{\ul{s}}{\ul{s}'}{\ul{\cT}}$ 
 has the domain $\{1,2,4\}$, and  
 is equal to the composition:
 \[
 \embeddingOp{\ext}{\dataAbs(\ttrim{\ul{\cT}}{1})}^{-1} \circ \embeddingOp{\dataAbs}{\ttrim{\ul{\cT}}{1}}^{-1}
 \circ \embeddingOp{\dataAbs}{\ul{\cT}} \circ \embeddingOp{\ext}{\dataAbs(\ul{\cT})}.
 \]
 The embeddings $\embeddingOp{\dataAbs}{\ttrim{\ul{\cT}}{1}}$ and $\embeddingOp{\dataAbs}{\ul{\cT}}$ 
 were given in Example~\ref{ex_emb_abst},
 and the embeddings $\embeddingOp{\ext}{\dataAbs(\ttrim{\ul{\cT}}{1})}$ and $\embeddingOp{\ext}{\dataAbs(\ul{\cT})}$ 
 were given in Example~\ref{ex_emb_ext}. We have:
\[
\begin{array}{lll}
\notOneTopEmbeddingNoArgT{\ul{s}}{\ul{s}'}{\ul{\cT}}(1) &=& 1 \\
\notOneTopEmbeddingNoArgT{\ul{s}}{\ul{s}'}{\ul{\cT}}(2) &=& 3 \\ 
\notOneTopEmbeddingNoArgT{\ul{s}}{\ul{s}'}{\ul{\cT}}(4) &=& 4	
\end{array}
\]
 \end{runningexample}

 \begin{definition}[Partial embedding]
 Let $(s',s)$ be consecutive extremal strings. 
 We denote by  
 $\notOneTopEmbeddingNoArg{s}{s'}$
 the function 
 from
 $s^{-1}(\Gamma - \Gammatop)$
 to $[|s'|]$ 
 defined as follows.
 Let $\cT$  be a task word, and 
 $s = \ext(\cT)$ and 
 $s' = \ext(\ttrim{\cT}{1})$. 
 Then we have
 $\notOneTopEmbeddingNoArg{s}{s'} = 
 \notOneTopEmbeddingNoArgT{s}{s'}{\cT}$.
 \end{definition}

The following lemma shows that partial embeddings are well-defined:
\begin{lemma}
\label{lem_f_extremal}
Let $\cT_1$ and $\cT_2$ be task words, and let $s = \ext(\dataAbs(\cT_1)) = \ext(\dataAbs(\cT_2))$ and
$s' = \ext(\dataAbs(\ttrim{\cT_1}{1})) = \ext(\dataAbs(\ttrim{\cT_2}{1}))$.
Then $\notOneTopEmbeddingNoArgT{s}{s'}{\cT_1}= 
\notOneTopEmbeddingNoArgT{s}{s'}{\cT_2}$. 
\end{lemma}

\paragraph*{Proof of Lemma~\ref{lem_f_extremal}}
\label{proof_prop_f_extremal}
Let $n_s = |s|$.
Since $\ext(\dataAbs(\cT_1)) = \ext(\dataAbs(\cT_2))$, 
the domains of 
$\notOneTopEmbeddingNoArgT{s}{s'}{\cT_1}$ and 
$\notOneTopEmbeddingNoArgT{s}{s'}{\cT_2}$ 
are equal.  
Assume for contradiction that $\notOneTopEmbeddingNoArgT{s}{s'}{\cT_1} \not= \notOneTopEmbeddingNoArgT{s}{s'}{\cT_2}$.
Let $\ell_s$ be the minimal position in $s$ such that
$s(\ell_s)\notin \Gammatop$ and 
$\notOneTopEmbeddingNoArgT{s}{s'}{\cT_1}(\ell_s) \neq \notOneTopEmbeddingNoArgT{s}{s'}{\cT_2}(\ell_s)$.
For $i=1,2$, 
let $w_i = \dataAbs(\cT_i)$, $n_i = |w_i|$, and $w_i' = \dataAbs(\ttrim{\cT_i}{1})$.
Let $\ell_i, d_i$, $\ell_{s',i}$ be as follows:
\[
 \begin{array}{lllll}
 \ell_i &=& \embeddingOp{\ext}{w_i}(\ell_s) \\
  d_i & = & \embeddingOp{\dataAbs}{\cT_i}(\ell_i)\\
  \ell_{s',i} &=& \notOneTopEmbeddingNoArgT{s}{s'}{\cT_i}(\ell_s)\\
 \end{array}
\]
and note $d_i = \embeddingOp{\ext \, \circ \, \dataAbs}{\cT_i}(\ell_s) = \embeddingOp{\ext \, \circ \, \dataAbs}{\ttrim{\cT_i}{1}}(\ell_{s',i})$,
and since $s(\ell_s)\notin \Gammatop$, $d_i$ belongs to the universe of $\ttrim{\cT_i}{1}$. 
For distinct $i,j\in\{1,2\}$, 
let $d_{i,j}$, $\ell_{i,j}$, $\ell_{w_2',1}$ be as follows:
\[
 \begin{array}{lllll}
\ell_{w_j',i} &=& \embeddingOp{\ext}{w_j'}(\ell_{s',i}) \\
 d_{i,j} &=& \embeddingOp{\dataAbs}{\ttrim{\cT_j}{1}}(\ell_{w_j',i})\\
\ell_{i,j} &=& \left(\embeddingOp{\dataAbs}{\cT_j}\right)^{-1}(d_{i,j})\\
\end{array}
\]
and note $d_{i,j} =  \embeddingOp{\ext \, \circ \, \dataAbs}{\ttrim{\cT_j}{1}}(\ell_{s',i})$,
and that $d_{i,j}$ belongs to the universe of $\ttrim{\cT_j}{1}$ and hence to that of $\cT_j$.
W.l.o.g $\ell_{s',1} < \ell_{s',2}$, and therefore by the order-preservation property of embeddings,  
$\cT_1 \models d_1 \lneq_1 d_{2,1}$, 
$\ell_1 < \ell_{2,1}$, 
$\cT_2 \models d_{1,2} \lneq_1 d_2$, and
$\ell_{1,2} < \ell_2$ (see Table~\ref{table:prop_f_extremal}). 
\begin{table}
\caption{\label{table:prop_f_extremal}
 }
In each row there is one element $d$ of $\ttrim{\cT_i}{1}$ with $i\in \{1,2\}$. 
Each of the columns $u\in \{w_1,w_2,s\}$ indicates the position of $d$ in the string $u$ according to the embedding
$\embeddingOp{Op}{\cT_i}$ with the appropriate operation $Op = \dataAbs$ for $u \in \{w_1,w_2\}$ and $Op = \ext \circ \dataAbs$ for $u=s$. 
Each of the columns $u'\in \{s,w_2'\}$ indicates the position of $d$ in the string $u'$ according to the embedding
$\embeddingOp{Op}{\ttrim{\cT_i}{1}}$ with $Op = \dataAbs$ for $u' = w_2'$ and $Op = \ext \circ \dataAbs$ for $u'=s$.
Each of the elements and positions in row 1  (respectively row 3) 
are strictly smaller than the elements and positions in the same column in row 2 (respectively row 4). 
(The comparison of elements is with respect to $\leq_1$.) 
\\ \, \
\\
\begin{tabular}{c l l l l l l l}
\cline{2-8}
& ~ $\ttrim{\cT_1}{1}$ & $\ttrim{\cT_2}{1}$ & $w_1$ &  $w_2$  & $s$ & $s'$ & $w_2'$
\\
\cline{2-8}
\footnotesize smaller & & $d_{1,2}$  & & $\ell_{1,2}$ & & $\ell_{s',1}$ & $\ell_{w_2',1}$  \\
\footnotesize larger  & & $d_2$ &  & $\ell_2$ & $\ell_s$ & $\ell_{s',2}$ & \\
\cline{2-8}
\footnotesize smaller & ~ $d_1$ &  & $\ell_1$ &  & $\ell_s$ & $\ell_{s',1}$ &\\
\footnotesize larger & ~ $d_{2,1}$ &  & $\ell_{2,1}$ &  & & $\ell_{s',2}$ &  
\\
\cline{2-8}
\end{tabular}
\end{table}
Let 
\[
 \begin{array}{lll}
  \ts(s(\ell_s)) &=& \ts_s  \\
  \ts(w_i(\ell_{i})) &=& \ts_i\\
  \ts(s'(\ell_{s',i})) &=& \ts_{s',i}\\
  \ts(w_{j}(\ell_{i,j})) &=& \ts_{i,j}
  \end{array}
\]
From the definition of $\ext$, 
$\ts_s = \ts_1 = \ts_2$. 
By the definitions $\ext$ and $\dataAbs$ and from Lemma~\ref{lemma_trim_omega_compat}, 
$\ts_s$, $\ts_{s',1}$, $\ts_{s',2}$, $\ts_{1,2}$, and $\ts_{2,1}$
all realize the same set-type. 

Before continuing the proof of Lemma~\ref{lem_f_extremal}, we prove three claims. 

\begin{claim}
\label{lemma_PEmb_layers}
Let $(s',s)$ be a pair of consecutive extremal strings, and 
let  $\cT$ be a task word such that
$s = \ext(\dataAbs(\cT))$
and 
$s' = \ext(\dataAbs(\ttrim{\cT}{1}))$.
Let $\ell \in [n_s]$. Then:
\begin{enumerate}
 \item $s(\ell)\in \Gammattop$ if and only if $s'(\notOneTopEmbeddingNoArgT{s}{s'}{\cT}(\ell))\in \Gammatop$. \vspace{0.1em} 
 \item $s(\ell)\in \Gammarest$ if and only if $s'(\notOneTopEmbeddingNoArgT{s}{s'}{\cT}(\ell))\in \Gammattop \cup \Gammarest$. 
\end{enumerate}
\end{claim}
\begin{proof}
The claim follows from the definitions of  $\notOneTopEmbeddingNoArgT{s}{s'}{\cT}$, $\ext$, and $\dataAbs$, and from Claim~\ref{cl_extPos_cardinality}
\end{proof}

\begin{claim}
\label{cl_notextPos}
$\ell_{1,2} \notin \extPos(w_2)$. 
\end{claim}
\begin{proof}
	Assume for contradiction that $\ell_{1,2} \in \extPos(w_2)$. 
	Then there exists $\tilde{\ell}_{1,2} \in [n_s]$ such that
	$\tilde{\ell}_{1,2} = \embeddingOp{\ext}{w_2}(\ell_{1,2})$ and 
	$d_{1,2} = \embeddingOp{\ext \, \circ \, \dataAbs}{\cT_2}(\tilde{\ell}_{1,2})$ 
	and we have that $\notOneTopEmbeddingNoArgT{s}{s'}{\cT_2}(\tilde{\ell}_{1,2}) = \ell_{s',1}$.
	Since $d_{1,2} \lneq_1 d_2$ we have 
	$s(\tilde{\ell}_{1,2}) \notin \Gammatop$ and $\tilde{\ell}_{1,2} < \ell_s$. 
	Since $\notOneTopEmbeddingNoArgT{s}{s'}{\cT_1}$ is injective, 
	and $\notOneTopEmbeddingNoArgT{s}{s'}{\cT_1}(\ell_s) = \ell_{s',1}$, 
	we have $\notOneTopEmbeddingNoArgT{s}{s'}{\cT_1}(\tilde{\ell}_{1,2}) \neq \ell_{s',1}$, in contradiction to the minimality of $\ell_s$. 
\end{proof}

\begin{claim}
\label{cl_rest}
Let $P_{\bar{\theta}} \in \ts_s$, 
$s(\ell_s) \in \Gammarest$, 
and $\ell_2 \in \extremal_{\bar{\theta}}(w_2)$. Then 
$\bar{\theta}(x,y) \models x \leq_1 y$ and 
$P_{\bar{\theta}} \in \ts_{2,1}$.
\end{claim}
\begin{proof}
Since $s(\ell_s) \in \Gammarest$,
by Claim~\ref{lemma_PEmb_layers} we have $s'(\ell_{s',1}),s'(\ell_{s',2}) \in \Gammattop \cup \Gammarest$. 
Hence 
\[
\begin{array}{ll}
\val{\cD}{d_1},\val{\cD}{d_2},\val{\cD}{d_{2,1}},\val{\cD}{d_{1,2}} 
\leq \\
\maxdv_{\cD} -2 = \maxdv_{\trim{\cD}{1}} -1
\end{array}
\]
and $s(\ell_{1,2}),s(\ell_{2,1}) \in \Gammarest$. 
Since $\omega(\ts_s) = \omega(\ts_{1,2})=\omega(\ts_{2,1})$, we have $\bar{\theta} \in \omega(\ts_{1,2}) \cap \omega(\ts_{2,1})  $. 
Since $\trim{\cD}{1}$ is a substructure of $\cD$ and $P_{\bar{\theta}} \in \ts_s$, 
we also have $P_{\bar{\theta}} \in \ts_{s',1}\cap \ts_{s',2}$.

Assume for contradiction that $\bar{\theta}(x,y) \models y \lneq_1 x$. 
If $P_{\bar{\theta}} \in \ts_{1,2}$,
then since $\ell_{1,2} < \ell_2$, we have
$
\ell_2 \notin \extremal_{\bar{\theta}}(w_2) = 
\left\{\ell \mid \ell = \min(\positions_{\restrm, P_\theta}(w))\right\}
$
in contradiction. 
If $C_{\bar{\theta}} \in \ts_{1,2}$, then since $P_{\bar{\theta}} \in \ts_{s',1} \cap \ts_{s',2}$, 
and since $\bar{\theta}(x,y) \models y \lneq_1 x$, 
from Lemma~\ref{lemma_cor_rest_theta} (with $d = d_{1,2}$ and $d'=d_2$) it follows that $C_{\bar{\theta}} \in \ts_s$, in contradiction to $P_{\bar{\theta}} \in \ts_s$.
Hence $\bar{\theta}(x,y) \models x \leq_1 y$.

Since $\bar{\theta} \in \omega(\ts_{2,1})$, either $P_{\bar{\theta}}\in\ts_{2,1}$ or $C_{\bar{\theta}}\in\ts_{2,1}$. 
If $C_{\bar{\theta}} \in \ts_{2,1}$, then
from Lemma~\ref{lemma_cor_rest_theta} (with $d = d_1$ and $d'=d_{2,1}$) it follows that $C_{\bar{\theta}} \in  \ts_s$, in contradiction to $P_{\bar{\theta}} \in \ts_s$.
Hence $P_{\bar{\theta}} \in \ts_{2,1}$. 
\end{proof}

We are now ready to resume the proof of Lemma~\ref{lem_f_extremal}. 
We divide into cases depending on whether 
$s(\ell_s) \in \Gammattop$ or $s(\ell_s) \in \Gammarest$. 
First assume that $s(\ell_s) \in \Gammattop$. 
	Since $\ell_{s',1} = \notOneTopEmbeddingNoArgT{s}{s'}{\cT_1}(\ell_s)$,
	we have that $s'(\ell_{s',1}) \in \Gammatop$ by Claim~\ref{lemma_PEmb_layers}.
	Since $s' = \ext(w_2')$, we have 
	that
	$\ell_{w_2',i} \in \extPos(w_2')$, so
	there is some $\theta \in \omega(\ts_s)$ such that 
	$\ell_{w_2',i} \in \extremal_{\mathrm{1top},\theta}(w_2')$ 
	and $d_{1,2} \in \embeddingOp{\dataAbs}{\ttrim{\cT_2}{1}}(\extremal_{\mathrm{1top},\theta}(w_2'))$. 
	Therefore
	by Lemma~\ref{cl_extPos_cardinality}(1) we have that $d_{1,2} \in \embeddingOp{\dataAbs}{\cT_2}(\extremal_{\twotoprm,\theta}(w_2))$.
	Hence $\ell_{1,2} \in \extPos(w_2)$, in contradiction to Claim~\ref{cl_notextPos}. 
	
	Now assume
	$s(\ell_s) \in \Gammarest$. We have $s'(\ell_{s',1}),s'(\ell_{s',2}) \in \Gammattop \cup \Gammarest$ by Claim~\ref{lemma_PEmb_layers},
	i.e.~ 
	$
	d_{1,2},d_{2,1} \leq \maxdv_{\trim{\cD}{1}} -1 = \maxdv_{\cD} - 2
	$.
	Hence, $s(\ell_{1,2}),s(\ell_{2,1}) \in \Gammarest$. 
	Since  $\ell_2 = \embeddingOp{\ext}{w_2}(\ell_s)$ we have that $\ell_2 \in \extPos(w_2)$.
	Let $\theta\in \omega(\ts_s) = \omega(\ts_{1,2})$. 
	By Claim~\ref{cl_notextPos} and using that $s(\ell_{1,2})\in \Gammarest$, 
	there exists $\tilde{\ell}_{1,2} < \ell_{1,2}$ such that $\tilde{\ell}_{1,2} \in \extPos_{\restrm,\theta}(w_2)$. 
	Let $\ell_{0,\theta}$ be such that $\tilde{\ell}_{1,2} = \embeddingOp{\ext}{w_2}(\ell_{0,\theta})$.
	We have that $\tilde{\ell}_{1,2} < \ell_2$, and hence $\ell_{0,\theta} < \ell_s$. 
	Let $s({\ell}_{0,\theta}) = (\hlayerp{0,\theta},\ts_{0,\theta})$. 
	Since $\tilde{\ell}_{1,2} \in \extPos_{\restrm,\theta}(w_2)$ we have $\theta \in \omega(\ts_{0,\theta})$. 
	Let $\ell_{3,\theta}$ be such that $\ell_{2,1} = \embeddingOp{\ext}{w_1}(\ell_{3,\theta})$.
	Since $\ell_1 < \ell_{2,1}$ we have $\ell_s < \ell_{3,\theta}$. 
	We have $\theta \in \omega(\ell_{2,1}) = \omega(\ell_{3,\theta})$. 	
	Hence, $\ell_2 \notin \extremal_{\mathrm{rest},\theta}(w_2)$ for all $\theta \in \omega(\ts)$. 
	Consequently, there is $\bar{\theta}$ such that 
	$\ell_2 \in \extremal_{\bar{\theta}}(w_2)$ 
	and $P_{\bar{\theta}} \in \ts_2 = \ts_s$.

  By Claim~\ref{cl_rest} we have $\bar{\theta}(x,y) \models x \leq_1 y$
	and $P_{\bar{\theta}} \in \ts_{2,1}$, and hence 
	there is $\ell_{\bar{\theta},w_1}$ such that 
	$\extremal_{\bar{\theta}}(w_1)= \{\ell_{\bar{\theta},w_1}\}$ and $\ell_{\bar{\theta},w_1} = \max(\positions_{\restrm, P_{\bar{\theta}}}(w_1))$.
	We have $P_{\bar{\theta}} \in \ts(w_1(\ell_{\bar{\theta},w_1}))$ and 
	$\ell_1 < \ell_{2,1} \leq \ell_{\bar{\theta},w_1}$.
	Let $\ell_{\bar{\theta},s} \in [n_s]$ 
	and $\ell_{\bar{\theta},w_2} \in [n_2]$
	be such that $\ell_{\bar{\theta},w_1} = \embeddingOp{\ext}{w_1}(\ell_{\bar{\theta},s})$ and
	$\ell_{\bar{\theta},w_2}  = \embeddingOp{\ext}{w_2}(\ell_{\bar{\theta},s})$, 
	then we have that 
	$P_{\bar{\theta}} \in 
	\ts(s(\ell_{\bar{\theta},s})) = \ts(w_1(\ell_{\bar{\theta},w_1})) = \ts(w_2(\ell_{\bar{\theta},w_2}))$, 
	$\ell_s < \ell_{\bar{\theta},s}$, and
	$\ell_2 < \ell_{\bar{\theta},w_2}$.
	Since $w_1(\ell_{\bar{\theta},w_1})  \in \Gammarest$
	and $w_1(\ell_{\bar{\theta},w_1}) = s(\ell_{\bar{\theta},s}) = w_2(\ell_{\bar{\theta},w_2})$,
	we have
	$\ell_2 \notin \extremal_{\bar{\theta}}(w_2) = 
	\left\{\max(\positions_{\restrm, P_{\bar{\theta}}}(w_2))\right\}$, 
	in contradiction.

\end{document}